\documentclass[a4paper,cleveref, autoref, thm-restate,svgnames]{lipics-v2019}
\newtheorem*{thm1}{Theorem 1}
\newtheorem*{thm2}{Theorem 2}

\hideLIPIcs{}
\nolinenumbers

\title{Eternal Vertex Cover on Bipartite and Co-Bipartite Graphs} 

\titlerunning{Eternal Vertex Cover on Bipartite and Co-Bipartite Graphs} 

\author{Neeldhara Misra}{Department of Computer Science and Engineering,\and Indian Institute of Technology, Gandhinagar \and \url{http://www.neeldhara.com} }{neeldhara.m@iitgn.ac.in}{https://orcid.org/0000-0003-1727-5388}{}

\author{Saraswati Girish Nanoti}{Department of Mathematics,\and Indian Institute of Technology, Gandhinagar}{nanoti\_saraswati@iitgn.ac.in}{}{}

\authorrunning{N. Misra and S. Nanoti} 

\Copyright{N. Misra and S. Nanoti} 



\category{} 

\relatedversion{} 

\supplement{}

\EventEditors{John Q. Open and Joan R. Access}
\EventNoEds{2}
\EventLongTitle{42nd Conference on Very Important Topics (CVIT 2016)}
\EventShortTitle{CVIT 2016}
\EventAcronym{CVIT}
\EventYear{2016}
\EventDate{December 24--27, 2016}
\EventLocation{Little Whinging, United Kingdom}
\EventLogo{}
\SeriesVolume{42}
\ArticleNo{23}

\usepackage{etoolbox}
\makeatletter
\patchcmd{\BR@backref}{\newblock}{\newblock($\uparrow$~}{}{}
\patchcmd{\BR@backref}{\par}{)\par}{}{}
\makeatother

\usepackage{natbib}
\hypersetup{
    colorlinks=true,
    linkcolor=IndianRed,
    filecolor=magenta,      
    urlcolor=DodgerBlue,
    citecolor = SeaGreen
}

\usepackage{caption}
\usepackage{subcaption}
\usepackage{charter}
\usepackage[backgroundcolor=SkyBlue,linecolor=SkyBlue,obeyFinal]{todonotes}

\makeatletter
\providecommand\@dotsep{5}
\def\listtodoname{}
\def\listoftodos{\@starttoc{tdo}\listtodoname}
\makeatother

\newcounter{nmcomment}

\usepackage{parskip}
\usepackage{comment}
\excludecomment{shortver}
\includecomment{longver}

\usepackage{wrapfig}
\usepackage[font=small]{caption, subcaption}
\usepackage{tikz}
\usetikzlibrary{positioning}
\usetikzlibrary{arrows}
\usetikzlibrary{snakes}
\usetikzlibrary{decorations.pathmorphing}
\usetikzlibrary{decorations.markings}
\usetikzlibrary{decorations.pathreplacing}
\usetikzlibrary{shapes.geometric}
\tikzset{
    small circles/.style={circle,inner sep=2pt,fill=#1},
    hollow circles/.style n args={2}{circle,inner sep=#1,draw=#2,thick},
    stars/.style={star,inner sep=2pt}
}
    
\usepackage{latexsym,amssymb,amsmath,mathtools,eulervm,xspace}

\newcommand{\defparproblem}[4]{
\vspace{1mm}
\begin{center}
\noindent\fbox{

  \begin{minipage}{.9\linewidth}
  \begin{tabular*}{\linewidth}{@{\extracolsep{\fill}}lr} \textsc{#1}  \\ \end{tabular*}
  {\bf{Input:}} #2  \\
  {\bf{Parameter:}} #3  \\
  {\bf{Question:}} #4
  \end{minipage}

  }
\end{center}
  \vspace{1mm}
}

\newcommand{\NPH}{\ensuremath{\mathsf{NP}}-hard\xspace}

\newcommand{\EVC}{\textsc{Eternal Vertex Cover}}
\newcommand{\ECVC}{\textsc{Eternal Connected Vertex Cover}}

\newcommand{\RBDS}{\textsc{Red Blue Dominating Set}}

\begin{document}

\maketitle

\begin{abstract}
The \EVC{} problem is a dynamic variant of the vertex cover problem. We have a two player game in which guards are placed on some vertices of a graph. In every move, one player (the attacker) attacks an edge. In response to the attack, the second player (the defender) moves some of the guards along the edges of the graph in such a manner that at least one guard moves along the attacked edge. If such a movement is not possible, then the attacker wins. If the defender can defend the graph against an infinite sequence of attacks, then the defender wins. 

The minimum number of guards with which the defender has a winning strategy is called the eternal vertex cover number of the graph $G$. On general graphs, the computational problem of determining the minimum eternal vertex cover number is \NPH{} and admits a $2$-approximation algorithm and an exponential kernel. The complexity of the problem on bipartite graphs is open, as is the question of whether the problem admits a polynomial kernel. 

We settle both these questions by showing that Eternal Vertex Cover is \NPH{} and does not admit a polynomial compression even on bipartite graphs of diameter six. We also show that the problem admits a polynomial time algorithm on the class of cobipartite graphs.
\end{abstract}

\newpage 

\section{Introduction}
The \EVC{} problem is a dynamic variant of the vertex cover problem introduced by~\cite{KM09}. The setting is the following. We have a two player game --- between players whom we will refer to as the \emph{attacker} and \emph{defender} --- on a simple, undirected graph $G$. In the beginning, the defender can choose to place guards on some of the vertices of $G$. The attacker's move involve choosing an edge to ``attack''. The defender is able to ``defend'' this attack if she can move the guards along the edges of the graph in such a way that at least one guard moves along the attacked edge. If such a movement is not possible, then the attacker wins. If the defender can defend the graph against an infinite sequence of attacks, then the defender wins (see~\Cref{fig:intro}). The minimum number of guards with which the defender has a winning strategy is called the \emph{eternal vertex cover number} of the graph $G$ and is denoted by $evc(G)$. 

\begin{figure}[ht]
\centering
\begin{subfigure}[b]{0.42\textwidth}
\resizebox{\textwidth}{!}{%
   \begin{tikzpicture}[scale=0.25]
\node[small circles=DodgerBlue] (A1) at (0,0) {};
\node[small circles=DodgerBlue] (B1) [below of=A1] {};

\foreach \x in {2,3,...,6}{
    \pgfmathtruncatemacro{\y}{\x-1}
    \node[small circles=DodgerBlue] (A\x) [right of=A\y] {};
}
\foreach \x in {2,3,...,5}{
    \pgfmathtruncatemacro{\y}{\x-1}
    \node[small circles=DodgerBlue] (B\x) [right of=B\y] {};
};
\node (B6) [right of = B5] {};    
\node[small circles=DodgerBlue] (B7) [right of = B6] {};

\draw (A1) -- (B1) -- (B2) -- (A2) -- (B3) -- (A3) -- (A4) -- (B4)  -- (B5) -- (A4) -- (A5) -- (A6) -- (B5) -- (A6) -- (B7);
\draw (A1) -- (A2);
\draw (B1) -- (A2);
\draw (A1) -- (B2);

\node[star,fill=DarkOrange,inner sep=2pt] at (A1) {};
\node[star,fill=DarkOrange,inner sep=2pt] at (A2) {};
\node[star,fill=DarkOrange,inner sep=2pt] at (B1) {};
\node[star,fill=DarkOrange,inner sep=2pt] at (A3) {};
\node[star,fill=DarkOrange,inner sep=2pt] at (A4) {};
\node[star,fill=DarkOrange,inner sep=2pt] at (B5) {};
\node[star,fill=DarkOrange,inner sep=2pt] at (A6) {};
\node[star,fill=DarkOrange,inner sep=2pt] at (B4) {};
\end{tikzpicture}
}
   \caption{The intial positions of the guards are denoted by the star-shaped vertices.}
   \label{fig:intro1} 
\end{subfigure}
\begin{subfigure}[b]{0.42\textwidth}
\resizebox{\textwidth}{!}{%
   \begin{tikzpicture}[scale=0.25]


\node[small circles=DodgerBlue] (A1) at (0,0) {};
\node[small circles=DodgerBlue] (B1) [below of=A1] {};

\foreach \x in {2,3,...,6}{
    \pgfmathtruncatemacro{\y}{\x-1}
    \node[small circles=DodgerBlue] (A\x) [right of=A\y] {};
}
\foreach \x in {2,3,...,5}{
    \pgfmathtruncatemacro{\y}{\x-1}
    \node[small circles=DodgerBlue] (B\x) [right of=B\y] {};
};
\node (B6) [right of = B5] {};    
\node[small circles=DodgerBlue] (B7) [right of = B6] {};

\draw (A1) -- (B1) -- (B2) -- (A2) -- (B3) -- (A3) -- (A4) -- (B4)  -- (B5) -- (A4) -- (A5) -- (A6) -- (B5) -- (A6) -- (B7);
\draw (A1) -- (A2);
\draw (B1) -- (A2);
\draw (A1) -- (B2);

\node[star,fill=DarkOrange,inner sep=2pt] at (A1) {};
\node[star,fill=DarkOrange,inner sep=2pt] at (A2) {};
\node[star,fill=DarkOrange,inner sep=2pt] at (B1) {};
\node[star,fill=DarkOrange,inner sep=2pt] at (A3) {};
\node[star,fill=DarkOrange,inner sep=2pt] at (A4) {};
\node[star,fill=DarkOrange,inner sep=2pt] at (B5) {};
\node[star,fill=DarkOrange,inner sep=2pt] at (A6) {};
\node[star,fill=DarkOrange,inner sep=2pt] at (B4) {};
\draw[thick,color=IndianRed,decoration = {zigzag,segment length = 3pt,amplitude=1pt},decorate] (A6) -- (B7);
\end{tikzpicture}
}
   \caption{The attackers move targets the edge to the far-right, highighted by a wavy red line.}
   \label{fig:intro2}
   \end{subfigure}
\begin{subfigure}[b]{0.42\textwidth}
\resizebox{\textwidth}{!}{%
   \begin{tikzpicture}


\node[small circles=DodgerBlue] (A1) at (0,0) {};
\node[small circles=DodgerBlue] (B1) [below of=A1] {};

\foreach \x in {2,3,...,6}{
    \pgfmathtruncatemacro{\y}{\x-1}
    \node[small circles=DodgerBlue] (A\x) [right of=A\y] {};
}
\foreach \x in {2,3,...,5}{
    \pgfmathtruncatemacro{\y}{\x-1}
    \node[small circles=DodgerBlue] (B\x) [right of=B\y] {};
};
\node (B6) [right of = B5] {};    
\node[small circles=DodgerBlue] (B7) [right of = B6] {};

\draw (A1) -- (B1) -- (B2) -- (A2) -- (B3) -- (A3) -- (A4) -- (B4)  -- (B5) -- (A4) -- (A5) -- (A6) -- (B5) -- (A6) -- (B7);
\draw (A1) -- (A2);
\draw (B1) -- (A2);
\draw (A1) -- (B2);

\node[star,fill=DarkOrange,inner sep=2pt] at (A1) {};
\node[star,fill=DarkOrange,inner sep=2pt] at (A2) {};
\node[star,fill=DarkOrange,inner sep=2pt] at (B1) {};
\node[star,fill=DarkOrange,inner sep=2pt] at (A3) {};
\node[star,fill=DarkOrange,inner sep=2pt] at (A4) {};
\node[star,fill=DarkOrange,inner sep=2pt] at (B5) {};
\node[star,fill=DarkOrange,inner sep=2pt] at (B7) {};
\node[star,fill=DarkOrange,inner sep=2pt] at (B4) {};
\draw[thick,color=SeaGreen,decoration = {zigzag,segment length = 3pt,amplitude=1pt},decorate] (A6) -- (B7);

\end{tikzpicture}
}
   \caption{The defender moves to defend the attack by moving a guard along the attacked edge.}
   \label{fig:Ng2}
   \end{subfigure}
\begin{subfigure}[b]{0.42\textwidth}
\resizebox{\textwidth}{!}{%
\begin{tikzpicture}


\node[small circles=DodgerBlue] (A1) at (0,0) {};
\node[small circles=DodgerBlue] (B1) [below of=A1] {};

\foreach \x in {2,3,...,6}{
    \pgfmathtruncatemacro{\y}{\x-1}
    \node[small circles=DodgerBlue] (A\x) [right of=A\y] {};
}
\foreach \x in {2,3,...,5}{
    \pgfmathtruncatemacro{\y}{\x-1}
    \node[small circles=DodgerBlue] (B\x) [right of=B\y] {};
};
\node (B6) [right of = B5] {};    
\node[small circles=DodgerBlue] (B7) [right of = B6] {};

\draw (A1) -- (B1) -- (B2) -- (A2) -- (B3) -- (A3) -- (A4) -- (B4)  -- (B5) -- (A4) -- (A5) -- (A6) -- (B5) -- (A6) -- (B7);
\draw (A1) -- (A2);
\draw (B1) -- (A2);
\draw (A1) -- (B2);

\node[star,fill=DarkOrange,inner sep=2pt] at (A1) {};
\node[star,fill=DarkOrange,inner sep=2pt] at (A2) {};
\node[star,fill=DarkOrange,inner sep=2pt] at (B1) {};
\node[star,fill=DarkOrange,inner sep=2pt] at (A3) {};
\node[star,fill=DarkOrange,inner sep=2pt] at (A5) {};
\node[star,fill=DarkOrange,inner sep=2pt] at (B5) {};
\node[star,fill=DarkOrange,inner sep=2pt] at (B7) {};
\node[star,fill=DarkOrange,inner sep=2pt] at (B4) {};
\end{tikzpicture}
}
   \caption{The defender moves another guard to ensure that no edges are left vulnerable. This is the final position of the guards.}
   \label{fig:intro4}
   \end{subfigure}
\caption{An attack that is defended by moving two guards.}
\label{fig:intro}
\end{figure}
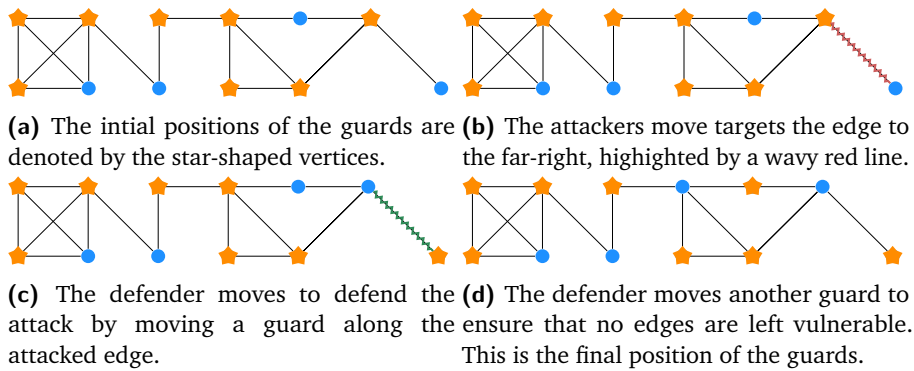

If $S_\ell$ is the subset of vertices that have guards on them after the defender has played her $\ell$-th move, and $S_\ell$ is not a vertex cover of $G$, then the attacker can target any of the uncovered edges to win the game. Therefore, when the defender has a winning strategy, it implies that she can always ``reconfigure'' one vertex cover into another in response to any attack, where the reconfiguration is constrained by the rules of how the guards can move and the requirement that at least one of these guards needs to move along the attacked edge. Therefore, it is clear that $evc(G)\geq mvc(G)$, where $mvc(G)$ denotes the minimum size of a vertex cover of $G$. It also turns out that twice as many vertices as the $mvc(G)$ also suffice the defend against any sequence of attacks --- this might be achieved, for example, by placing guards on both endpoints of any maximum matching. Therefore, we have that $mvc(G)\leq evc(G) \leq 2mvc(G)$.

\cite{KM09} gave a characterization of the graphs for which the upper bound is achieved. A characterization for graphs for which lower bound is achieved remains open, but several special cases have been addressed in the literature~\citep[see, for instance][]{BCFPRW2021}. Also,~\cite{KM2011} study graphs with equal eternal vertex cover and eternal domination numbers, which is a closely related dynamic variant of the dominating set problem.

The natural computational question associated with this parameter is the following: given a graph $G$ and a positive integer $k$, determine if $evc(G) \leq k$. The problem is only known to be in PSPACE in general. \cite{FGGKS2010} show that this problem is NP-hard by a reduction from vertex cover, and admits a $2$-approximation algorithm based on both endpoints of a matching. They also study the problem from a parameterized perspective. In parameterized complexity, one asks if for an instance of size $n$ and a parameter $k$, a problem can be solved in time $f(k) n^{\mathcal{O}(1)}$ where $f$ is an arbitrary computable function independent of~$n$. Problems that can be solved in that time are said to be fixed parameter tractable, and the corresponding complexity class is called FPT. They show that the problem is fixed parameter tractable when parameterized by the number of available guards $k$, by demonstrating an algorithm with running time $\mathcal{O}\left(2^{O\left(k^{2}\right)}+n m\right)$ for \EVC{}, where $n$ is the number of vertices and $m$ the number of edges of the input graph. This work leaves open the question of whether \EVC{} admits a polynomial kernel\footnote{We refer the reader to Section 2 for the definition of the notion of a polynomial kernel.}.  

The comptuational question of \EVC{} is also well studied on special classes of graphs. For instance, it is known to be NP-complete when restricted to locally connected graphs, a graph class which includes all biconnected internally triangulated planar graphs~\citep{BCFPRW2021}. It can also be solved in linear time on the class of cactus graphs~\citep{BPS2020b}, quadratic time on chordal graphs~\citep{BP2020,BPS2020b} and in polynomial time on ``generalized'' trees~\citep{AFI2015}. However, the complexity of the problem on biparitite graphs remains open, and is an intriguing question especially considering that the vertex cover problem is tractable on biparitite graphs. 

\subsection{Our Contributions}

We resolve the question of the complexity of \EVC{} on bipartite graphs by showing NP-hardness even on bipartite graphs of constant diameter. It turns out that the same result can also be used to argue the likely non-existence of a polynomial compression, which resolves the question of whether~\EVC{} has a polynomial kernel in the negative. Finally, we also observe that the hardness results carry over to the related problem of~\ECVC{}~\citep{FN2020}, where we would like the vertex covers at every step to induce connected subgraphs. 

Summarizing, our main result is the following:

\begin{thm1}
Both the \EVC{} and \ECVC{} problems are \NPH{} and do not admit a polynomial compression parameterized by the number of guards (unless $\mathrm{NP} \subseteq \operatorname{coNP}/\text{poly}$), even on bipartite graphs of diameter six. 
\end{thm1}

We also show that~\EVC{} is tractable on the class of cobipartite graphs, a result that is not implied by any of the polynomial time algorithms known so far. 

\begin{thm2}
There is a polynomial-time algorithm for \EVC{} on the class of cobipartite graphs. 
\end{thm2}

\paragraph*{Organization}

We establish notation and provide relevant definitions in~\Cref{sec:prelims}. The proof of Theorem~1 follows from the construction described in~\Cref{lem:bipartite}, and is the main focus of~\Cref{sec:bip}, while the proof of Theorem~2 can be found in~\Cref{sec:cobip} (with an overall schematic of the cases given in~\Cref{fig:overview,fig:overview1,fig:overview2,fig:overview3}). In~\Cref{sec:concl}, we suggest some directions for further work.

\newpage 

\section{Preliminaries and Notations}
\label{sec:prelims}

All graphs in this paper are finite, undirected and without multiple edges and loops. For notation and terminology not defined in this paper we refer to~\cite{graphsbook}. 

Let $G=(V,E)$ be a graph. The set of neighbours of a vertex $v$ in $G$ is denoted by $N_{G}(v)$, or briefly by $N(v)$\footnote{Here, as elsewhere, we drop the index referring to the underlying graph if the reference is clear.}. More generally, for $U \subseteq V$, the neighbours in $V \backslash U$ of vertices in $U$ are called neighbours of $U$; their set is denoted by $N(U)$. A subset $S \subseteq V$ is said to be \emph{independent} if for all $u, v \in S$, $(u,v) \notin E$. 


\sloppypar
A path is a non-empty graph $P=(V,E)$ of the form $V=\left\{x_{0}, x_{1}, \ldots, x_{k}\right\}$ and~$E=\left\{x_{0} x_{1}, x_{1} x_{2}, \ldots, x_{k-1} x_{k}\right\}$, where the $x_{i}$'s are all distinct. The number of edges of a path is its length, and the path of length $k$ is denoted by $P^{k}$. The distance $d_{G}(x, y)$ in $G$ of two vertices $x, y$ is the length of a shortest $x-y$ path in $G$; if no such path exists, we set $d(x, y):=\infty$. The greatest distance between any two vertices in $G$ is the diameter of $G$, denoted by $\operatorname{diam}(G)$. 

A \emph{vertex cover} of a graph $G = (V,E)$ is a subset $S$ of the vertex set such that every edge has at least one of its endpoints in $S$. Note that $V \setminus S$ is an independent set. We use $\operatorname{mvc}(G)$ to denote the size of a minimum vertex cover of $G$. A \emph{dominating set} of a graph $G$ is a subset $X$ of the vertex set such that every vertex of $G$ either belongs to $X$ or has a neighbor in $X$. 

Consider a graph $G=(V, E)$ on $n$ vertices and $m$ edges. Guards are placed on the vertices of the graph in order to protect it from an infinite sequence (which is not known to the guards in advance) of attacks on the edges of the graph. In each round, one edge $u v \in E$ is attacked, and each guard either stays on the vertex it is occupying or moves to a neighboring vertex. 

Moreover, the guards are bound to move in such a way that at least one guard moves from $u$ to $v$ or from $v$ to $u$. The minimum number of guards which can protect all the edges of $G$ is called the eternal vertex cover number of $G$ and is denoted by $\operatorname{evc}(G)$.

A \emph{bipartite} graph is a graph whose vertex set can be partitioned into at most two independent sets. A \textit{co-bipartite graph} is a graph which is the complement of a bipartite graph. In other words, a co-bipartite graph is a graph whose vertex set can be partitioned into at most two cliques.

\paragraph*{Parameterized Complexity.} 
A parameterized problem $L$ is a subset of $\Sigma^{*} \times \mathbb{N}$ for some finite alphabet $\Sigma$. An instance of a parameterized problem consists of $(x, k)$, where $k$ is called the parameter. A central notion in parameterized complexity is fixed parameter tractability (FPT), which means for a given instance $(x, k)$ solvability in time $f(k) \cdot p(|x|)$, where $f$ is an arbitrary function of $k$ and $p$ is a polynomial in the input size. The notions of kernelization and compression are defined as follows.

\begin{definition} A kernelization algorithm, or in short, a kernel for a parameterized problem $Q \subseteq \Sigma^{*} \times \mathbb{N}$ is an algorithm that, given $(x, k) \in \Sigma^{*} \times \mathbb{N}$, outputs in time polynomial in $|x|+k$ a pair $\left(x^{\prime}, k^{\prime}\right) \in \Sigma^{*} \times \mathbb{N}$ such that (a) $(x, k) \in Q$ if and only if $\left(x^{\prime}, k^{\prime}\right) \in Q$ and (b) $\left|x^{\prime}\right|+k^{\prime} \leq g(k)$, where $g$ is an arbitrary computable function. The function $g$ is referred to as the size of the kernel. If $g$ is a polynomial function then we say that $Q$ admits a polynomial kernel.
\end{definition} 

\begin{definition}A polynomial compression of a parameterized language $Q \subseteq \Sigma^{*} \times \mathbb{N}$ into a language $R \subseteq \Sigma^{*}$ is an algorithm that takes as input an instance $(x, k) \in \Sigma^{*} \times \mathbb{N}$, works in time polynomial in $|x|+k$, and returns a string $y$ such that:
\begin{enumerate}
    \item $|y| \leq p(k)$ for some polynomial $p(\cdot)$, and
    \item $y \in R$ if and only if $(x, k) \in Q$.
\end{enumerate} 
\end{definition}
Our focus in this paper is the \EVC{} problem, in which we are interested in computing $\operatorname{evc}(G)$ for a graph $G$, and its parameterized complexity with respect to the number of guards:

\defparproblem{Eternal Vertex Cover}{A graph $G = (V,E)$ and a positive integer $k \in \mathbb{Z}^+$.}{$k$}{Does $G$ have an eternal vertex cover of size at most $k$?} 

\EVC{} is known to admit an exponential kernel of size $4^{k}(k+1)+2 k$~\citep{FGGKS2010}. We use the following standard framework to show that it is unlikely to admit a polynomial compression. 

\begin{definition}Let $P$ and $Q$ be parameterized problems. We say that $P$ is polynomial parameter reducible to $Q$, written $P \leq_{p p t} Q$, if there exists a polynomial time computable function $f: \Sigma^{*} \times \mathbb{N} \rightarrow \Sigma^{*} \times \mathbb{N}$ and a polynomial $p$, such that for all $(x, k) \in \Sigma^{*} \times \mathbb{N}$ (a) $(x, k) \in P$ if and only $\left(x^{\prime}, k^{\prime}\right)=f(x, k) \in Q$ and (b) $k^{\prime} \leq p(k)$. The function $f$ is called polynomial parameter transformation.
\end{definition}

\begin{proposition}
\label{prop:npc} Let $P$ and $Q$ be parameterized problems such that there is a polynomial parameter transformation from $P$ to $Q$. If $Q$ has a polynomial compression, then $P$ also has a polynomial compression.
\end{proposition}

In the \RBDS{} problem, we are given a bipartite graph $G=(B \cup R, E)$ and an integer $k$ and asked whether there exists a vertex set $S \subseteq R$ of size at most $k$ such that every vertex in $B$ has at least one neighbor in $S$. In the literature, the sets $B$ and $R$ are called ``blue vertices'' and ``red vertices'', respectively. It is known~\citep[see][Theorem 4.1]{DLS14} that RBDS parameterized by $(|B|, k)$ does not have a polynomial kernel, and more generally, a polynomial compression~\citep[see][Corollary 19.6]{kernelbook}:

\begin{proposition}[Corollary 19.6~\cite{kernelbook}]\label{prop:rbds-npc} The \RBDS{} problem, parameterized by $|B|+k$, does not admit a polynomial compression unless $\mathrm{coNP} \subseteq \mathrm{NP}/$poly.
\end{proposition} 

Note that based on~\Cref{prop:npc,prop:rbds-npc}, to show that a polynomial compression for \EVC{} parameterized by the number of guards implies $\mathrm{coNP} \subseteq \mathrm{NP}/$poly, it suffices to show a polynomial parameter transformation from \RBDS{} to \EVC{}.  

For more background on parameterized complexity and algorithms, the reader is referred to the books~\cite{pcbook,kernelbook,pcinvitation,flumgrohe,downeyfellows}.

\section{Hardness on Bipartite Graphs}
\label{sec:bip}

In this section we demonstrate the intractability of~\EVC{} on the class of bipartite graphs of diameter six. Our key tool is a reduction from~\RBDS{} which also happens to be a polynomial parameter transformation.

\begin{lemma}
\label{lem:bipartite}
There is a polynomial parameter transformation from \RBDS{} parameterized by $|B|+k$ to \EVC{} parameterized by solution size.
\end{lemma}

\begin{proof}Let $\langle G = (V,E),b+k \rangle$ be an instance of \RBDS{}. We have $V= R\cup B$. We denote the vertices in $R$ by $\{v_1, \ldots, v_r\}$, the vertices in $B$ by $\{u_1, \ldots, u_b\}$  and use $m$ to denote $|E|$. We assume that $G$ is connected, since \RBDS{} does not have a polynomial sized kernel even for connected graphs. We assume that every blue vertex has at least one red neighbour and by returning a trivial \textsc{No}-instance of \EVC{} if some blue vertex has no red neighbour. The correctness of this follows from the fact that if some blue vertex does not have a red neighbour then it cannot be dominated by any subset of $R$. Further, we assume that $k < b$ by returning a trivial \textsc{Yes}-instance of \EVC{} if $k \geq b$. Also we assume $b>1$, since when $b=1$, the instance is easily resolved and we may return an appropriate instance of~\EVC{} (a trivial~\textsc{Yes} instance if $k \geq 1$ and a trivial~\textsc{No} instance otherwise).


\paragraph*{The Construction.} We will develop an instance of \EVC{} which we denote by $\langle H, \ell \rangle$ based on $\langle G,k \rangle$ as follows. First, we introduce $r$ \emph{red} vertices, denoted by $A := \{v_i~|~1 \leq i \leq r\}$ and $n$ \emph{blue} vertices, denoted by $B := \{u_i~|~1 \leq i \leq b\}$. Next, for all $i \in [b]$, we add $b^2+3$ \emph{dependent} vertices of type $i$, denoted by $C_i := \{w^i_j~|~1\leq j \leq b^2+3\}$. Now, we add $b^2+3$ \emph{dependent} vertices of type $\star$, denoted by $D := \{w^i_j~|~1\leq j \leq b^2+3\}$. Finally, we add two special vertices denoted by $\star$ and $\dagger$, which we will refer to as the \emph{universal} and \emph{backup} vertices respectively. To summarize, the vertex set consists of the following $r+(b^3+b^2+4b+5)$ vertices:

$$ V(H) := A \cup B \cup C_1 \cup \cdots \cup C_n \cup D \cup \{\star,\dagger\}.$$

We now describe the edges in $H$:

\begin{itemize}

\item There are $m$ \emph{structural} edges given by $(v_p,u_q)$ for every pair $(p,q)$ such that $(v_p,u_q) \in E(G)$. In other words, for every edge $(v_p,u_q)$ in the graph $G$, the original vertex $v_p$ is adjacent to the partner vertex $u_q$.

\item The dependent vertices of type $i$ are adjacent to the $i^{th}$ blue vertex, i.e, for every $i \in [b]$, we have a \emph{sliding} edge $(u_i,w)$ for each $w \in C_i$.

\item The dependent vertices of type $\star$ are adjacent to the universal vertex, i.e., we have a \emph{sliding} edge $(\star,w)$ for each $w \in D$.

\item The universal vertex $\star$ is adjacent to every red vertex via a \emph{supplier} edge. In particular, for every $i \in [r]$, we have the edge $(v_i,\star)$. 

\item Finally, we have the edge $(\star,\dagger)$, indicating that the backup vertex $\dagger$ is adjacent to the universal vertex. We call this edge a \emph{bridge}. 

\end{itemize}

To summarize, we have the following edges in $H$:

\begin{equation} \label{eq1}
  \begin{split}
  E(H) & = \{(v_p,u_q)~|~ 1 \leq p \leq r; 1 \leq q \leq b; \mbox{ and } (v_p,u_q) \in E(H)\}~\longleftarrow \mbox{ the structural edges}\\
  & \cup \{(u_1,w)~|~ w \in C_i)\} \longleftarrow \mbox{ the type } 1 \mbox{ sliding edges}\\
  & \cup \vdots \\
  & \cup \{(u_i,w)~|~ w \in C_i)\} \longleftarrow \mbox{ the type } i \mbox{ sliding edges}\\
  & \cup \vdots \\
  & \cup \{(u_n,w)~|~ w \in C_i)\}~\longleftarrow \mbox{ the type } n \mbox{ sliding edges}\\ 
  & \cup \{(\star,w)~|~ w \in D)\}~\longleftarrow \mbox{ the type } \star \mbox{ sliding edges}\\ 
  & \cup \{(v_i,\star)~|~ 1 \leq i \leq n\} ~\longleftarrow \mbox{ the supplier edges} \\
  & \cup \{(\star,\dagger)\}~\longleftarrow \mbox{the bridge edge}.
  \end{split}
  \end{equation}

We now let $\ell := b+k+2$, and this completes the description of the reduced instance $\langle H, \ell \rangle$. 



\begin{claim}
  The vertex cover number of $H$ is $b+1$. 
\end{claim}

\begin{proof}
  This follows from the fact that there is a matching of size $b+1$ in $H$, consisting of edges joing each blue vertex and $\star$ to one of their adjacent dependent vertices. (showing the lower bound), and that $B \cup \{\star\}$ is a vertex cover in $H$ (which implies the upper bound).
\end{proof}

\begin{claim}
  \label{claim-invariant}
  Any vertex cover of $H$ that has at most $\ell$ vertices must contain $B \cup \{\star\}$. 
\end{claim}

\begin{proof}
Consider a vertex cover $S \subseteq V(H)$ that does not contain some blue vertex $u_i \in B$. Then $S$ must contain all the dependent vertices in $C_i$, but since $|C_i| = b^2+3$, this contradicts our assumption that $|S| \leq \ell$. 
Consider a vertex cover $S \subseteq V(H)$ that does not contain the universal vertex $\star$. Then $S$ must contain all the dependent vertices in $C_i$, but since $|C_i| = b^2+3$, this contradicts our assumption that $|S| \leq \ell$. 
\end{proof}


\paragraph*{The Backward Direction.} Suppose $\langle H, \ell \rangle$ is a  \textsc{Yes}-instance of \EVC{}. We argue that $\langle G = (V,E),k \rangle$ is a \textsc{Yes}-instance of \RBDS.

We know that any sequence of edge attacks in $H$ can be defended by deploying at most $\ell = n + k + 2$ guards. Let $\mathcal{S}$ denote the initial placement of guards. 

We now consider two cases:

\paragraph*{Case 1. $\mathcal{S}$ contains the backup vertex.}

We already know that $\mathcal{S}$ contains all the blue vertices and the universal vertex by~\Cref{claim-invariant}. This accounts for the positions of $(n+1)$ guards. Additionally, because of the case we are in, we have one guard on the backup vertex. So the remaining $k$ guards occupy either red or dependent vertices. We propose that the corresponding vertices in $G$ form a dominating set. Specifically, let 

$$A^\prime := \{j ~|~ 1 \leq j \leq r \mbox{ and } v_i \in \mathcal{S}\} \mbox{ and } B^\prime := \{j ~|~ 1 \leq j \leq b \mbox{ and } C_j \cap \mathcal{S} \neq \emptyset\}.$$ 

For each $j \in B^\prime$, let $\ell_j$ be such that $v_{\ell_j}$ is an arbitrarily chosen neighbor of $u_j$ in $G$. Note that it is possible that $j_1 \neq j_2$ in $B^\prime$ but $\ell_{j_1} = \ell_{j_2}$. We now define $C^\prime := \{\ell_j ~|~ j\in B^\prime\}$.

Intuitively speaking, our choice of dominating set is made by choosing all red vertices in $G$ for whom the corresponding vertices in $H$ have a guard on them, and additionally, for all blue vertices who have a guard on a dependent neighbor vertex in $H$, we choose an arbitrary red neighbor in $G$ --- while this choice may coincide for some blue vertices, we note that the total number of chosen vertices is no more than the number of guards who are positioned on dependent and red vertices, i.e, $k$. In other words, we have that $|A^\prime \cup C^\prime| \leq k$. 


We now claim that $S := \{v_i ~|~ i \in A^\prime \cup C^\prime\}$ is a  dominating set for the blue vertices in $G$. Suppose not. Then, let $u_t \in B$ be a vertex that is not dominated by $S$. Let us attack a structural edge $(u_t,v_q)$. Note that $v_q$ is not occupied by a guard, and the guard on $u_t$ is forced to move to $v_q$ to defend this attack. However, observe that our assumption that $u_t$ is not dominated in $G$ implies that no neighbor of $u_t$ has a guard in $\mathcal{S}$. Therefore, this configuration now cannot be extended to a vertex cover, contradicting our assumption that every attack can be defended. Therefore, $S$ is indeed a dominating set in $G$ of size at most $k$.


\paragraph*{Case 2. $\mathcal{S}$ does not contain the backup vertex.}

In this case, we attack the bridge. Let $\mathcal{S}^\prime$ denote the placement of the guards obtained by defending this attack. Note that $\mathcal{S}^\prime$ must contain the backup vertex. Now we argue as we did in the previous case. This concludes the proof in the reverse direction. 

\paragraph*{The Forward Direction.} Suppose $\langle G = (V,E),k \rangle$ is a \textsc{Yes}-instance of \RBDS. We argue that $\langle H, \ell \rangle$ is a \textsc{Yes}-instance of \EVC{}.

Let $S \subseteq V(G)$ be a dominating set of $G$. Without loss of generality (by renaming), we assume that $S = \{v_1, \ldots, v_k\}$. Let $S^\star := \{v_1, \ldots, v_k\} \subseteq A$, i.e, the red vertices of $H$ corresponding to the solution in $G$. 

Note that all the following are vertex covers of size $b + k + 2$ for $H$:

\begin{itemize}
\item A backup vertex cover is given by $X := B \cup \{\star\} \cup S^\star \cup \{\dagger\}$.
\item A dependent vertex cover of type $i$ (or type $\star$) is given by $Y_i := B \cup \{\star\} \cup S^\star \cup \{w\}$, where $w$ is a dependent vertex from $C_i$ (or $D$).
\item A red vertex cover is given by $Z_j := B \cup \{\star\} \cup S^\star \cup \{v_j\}$, for some $j > k$.
\end{itemize}

\begin{figure}[t]
    \centering
    \begin{tikzpicture}
    \tikz{

    \draw [DodgerBlue,thick,rounded corners] (1,0) rectangle (7,1);
    \draw [IndianRed,thick,rounded corners] (0,2.5) rectangle (8,3.5);
    
    \node [circle,draw,thick,OliveDrab,fill=white] (global) at (10,0.25) {$\star$}; 
    \node [circle,draw,thick,OliveDrab,fill=white] (backup) at (10,3.25) {$\dagger$};

    \tikzset{decoration={snake,amplitude=.4mm,segment length=2mm,
                       post length=0mm,pre length=0mm}}
    
    \draw [thick,decorate] (global) -- (backup);

    \foreach \x in {1,...,7}
        \draw[dotted] (global) -- (\x, 3);
    
    \draw[Sienna,dashed] (2,0.5) -- (2,-1.5);
    \draw[Sienna,dashed]  (1.75, -1.5) -- (2,0.5);
    \draw[Sienna,dashed]  (2.25, -1.5) -- (2,0.5);
    \draw[Sienna,dashed]  (1.5, -1.5) -- (2,0.5);
    \draw[Sienna,dashed]  (2.5, -1.5) -- (2,0.5);
    
    \draw[Sienna,dashed]  (6, -1.5) -- (6,0.5);
    \draw[Sienna,dashed]  (5.75, -1.5) -- (6,0.5);
    \draw[Sienna,dashed]  (6.25, -1.5) -- (6,0.5);
    \draw[Sienna,dashed]  (5.5, -1.5) -- (6,0.5);
    \draw[Sienna,dashed]  (6.5, -1.5) -- (6,0.5);
    

    \draw[Sienna,dashed]  (10, -1.5) -- (global);
    \draw[Sienna,dashed]  (9.75, -1.5) -- (global);
    \draw[Sienna,dashed]  (10.25, -1.5) -- (global);
    \draw[Sienna,dashed]  (9.5, -1.5) -- (global);
    \draw[Sienna,dashed]  (10.5, -1.5) -- (global);
    
    
    \draw[thick,LightSeaGreen] (2,0.5) -- (3,3); 
    \draw[thick,LightSeaGreen] (6,0.5) -- (7,3);
    \draw[thick,LightSeaGreen] (4,0.5) -- (4,3);
    \draw[thick,LightSeaGreen] (3,0.5) -- (2,3);
    \draw[thick,LightSeaGreen] (5,0.5) -- (4,3);
    \draw[thick,LightSeaGreen] (5,0.5) -- (6,3);
    \draw[thick,LightSeaGreen] (2,0.5) -- (6,3);
    \draw[thick,LightSeaGreen] (6,0.5) -- (2,3);
    \draw[thick,LightSeaGreen] (3,0.5) -- (6,3);
    \draw[thick,LightSeaGreen] (3,0.5) -- (1,3);
    \draw[thick,LightSeaGreen] (4,0.5) -- (5,3);
    \draw[thick,LightSeaGreen] (3,0.5) -- (3,3);
    \draw[thick,LightSeaGreen] (4,0.5) -- (1,3);
    \draw[thick,LightSeaGreen] (5,0.5) -- (2,3);

    \foreach \x in {2,...,6}
        \draw[fill=SteelBlue!42]  (\x, 0.5) circle (0.15cm);
    
    \foreach \x in {1,...,7}
        \draw[fill=Crimson!77] (\x, 3) circle (0.15cm);
        
    \draw[fill=OrangeRed!42]  (2, -1.5) circle (0.07cm);
    \draw[fill=OrangeRed!42]  (1.75, -1.5) circle (0.07cm);
    \draw[fill=OrangeRed!42]  (2.25, -1.5) circle (0.07cm);
    \draw[fill=OrangeRed!42]  (1.5, -1.5) circle (0.07cm);
    \draw[fill=OrangeRed!42]  (2.5, -1.5) circle (0.07cm);
    
    \draw[fill=OrangeRed!42]  (6, -1.5) circle (0.07cm);
    \draw[fill=OrangeRed!42]  (5.75, -1.5) circle (0.07cm);
    \draw[fill=OrangeRed!42]  (6.25, -1.5) circle (0.07cm);
    \draw[fill=OrangeRed!42]  (5.5, -1.5) circle (0.07cm);
    \draw[fill=OrangeRed!42]  (6.5, -1.5) circle (0.07cm);

    \draw[fill=OrangeRed!42]  (10, -1.5) circle (0.07cm);
    \draw[fill=OrangeRed!42]  (9.75, -1.5) circle (0.07cm);
    \draw[fill=OrangeRed!42]  (10.25, -1.5) circle (0.07cm);
    \draw[fill=OrangeRed!42]  (9.5, -1.5) circle (0.07cm);
    \draw[fill=OrangeRed!42]  (10.5, -1.5) circle (0.07cm);

    \node at (0,5.2) {};
    
    

    }
    \end{tikzpicture}
    \caption{A schematic depicting the construction of $(H,\ell)$ starting with an instance $(G,k)$ of \RBDS{}. The red vertices from $G$ instance are shown in the red rectangle on the top while the blue vertices are in the blue rectangle positioned at the bottom. The solid green lines correspond to edges in $E(G)$. The small orange vertices are the dependent vertices (some of them are omitted for clarity), while the global and backup vertices are shown by nodes labeled $\star$ and $\dagger$ respectively. The wavy line shows the bridge, the dotted lines shows the supplier edges while the dashed lines show the sliding edges.}
    \label{fig:my_label}
\end{figure}
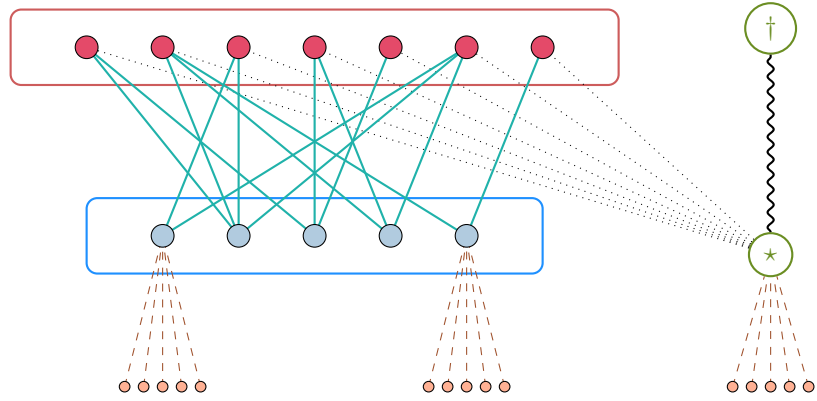

We remark that there may other vertex covers of size $n + k + 2$ that don't ``fit'' into any of the categories listed above. A vertex cover that is either a backup, dependent, or red vertex cover is called a \emph{nice} vertex cover.

We argue that if guards are occupying a nice vertex cover, then any attack can be defended by moving guards along edges in such a way that the new configuration also corresponds to a nice vertex cover. This implies that any sequence of attacks can be defended, by starting with an arbitrary nice vertex cover.

\begin{figure}[t]
    \centering
    \begin{tikzpicture}[square/.style={regular polygon,regular polygon sides=4}]
    \tikz{

    \draw [DodgerBlue,thick,rounded corners] (1,0) rectangle (7,1);
    \draw [IndianRed,thick,rounded corners] (0,2.5) rectangle (8,3.5);
    
    \node [circle,draw,thick,OliveDrab,fill=white] (backup) at (10,3.25) {$\dagger$}; 
    
    \node at (10,0.25) [square,draw,fill=YellowGreen] (global) {$\star$};

    \tikzset{decoration={snake,amplitude=.4mm,segment length=2mm,
                       post length=0mm,pre length=0mm}}
    
    \draw [thick,decorate] (global) -- (backup);

    \foreach \x in {1,...,7}
        \draw[dotted] (global) -- (\x, 3);
    
    \draw[Sienna,dashed] (2,0.5) -- (2,-1.5);
    \draw[Sienna,dashed]  (1.75, -1.5) -- (2,0.5);
    \draw[Sienna,dashed]  (2.25, -1.5) -- (2,0.5);
    \draw[Sienna,dashed]  (1.5, -1.5) -- (2,0.5);
    \draw[Sienna,dashed]  (2.5, -1.5) -- (2,0.5);
    
    \draw[Sienna,dashed]  (6, -1.5) -- (6,0.5);
    \draw[Sienna,dashed]  (5.75, -1.5) -- (6,0.5);
    \draw[Sienna,dashed]  (6.25, -1.5) -- (6,0.5);
    \draw[Sienna,dashed]  (5.5, -1.5) -- (6,0.5);
    \draw[Sienna,dashed]  (6.5, -1.5) -- (6,0.5);
    

    \draw[Sienna,dashed]  (10, -1.5) -- (global);
    \draw[Sienna,dashed]  (9.75, -1.5) -- (global);
    \draw[Sienna,dashed]  (10.25, -1.5) -- (global);
    \draw[Sienna,dashed]  (9.5, -1.5) -- (global);
    \draw[Sienna,dashed]  (10.5, -1.5) -- (global);
    
    
    \draw[thick,LightSeaGreen] (2,0.5) -- (3,3); 
    \draw[thick,LightSeaGreen] (6,0.5) -- (7,3);
    \draw[thick,LightSeaGreen] (4,0.5) -- (4,3);
    \draw[thick,LightSeaGreen] (3,0.5) -- (2,3);
    \draw[thick,LightSeaGreen] (5,0.5) -- (4,3);
    \draw[thick,LightSeaGreen] (5,0.5) -- (6,3);
    \draw[thick,LightSeaGreen] (2,0.5) -- (6,3);
    \draw[thick,LightSeaGreen] (6,0.5) -- (2,3);
    \draw[thick,LightSeaGreen] (3,0.5) -- (6,3);
    \draw[thick,LightSeaGreen] (3,0.5) -- (1,3);
    \draw[thick,LightSeaGreen] (4,0.5) -- (5,3);
    \draw[thick,LightSeaGreen] (3,0.5) -- (3,3);
    \draw[thick,LightSeaGreen] (4,0.5) -- (1,3);
    \draw[thick,LightSeaGreen] (5,0.5) -- (2,3);

    \foreach \x in {2,...,6}
        \node at (\x, 0.5) [square,draw,fill=YellowGreen] () {};
    
    \foreach \x in {4,...,7}
        \draw[fill=Crimson!77] (\x, 3) circle (0.15cm);

    \foreach \x in {1,2,3}
        \node at (\x, 3) [square,draw,fill=YellowGreen] () {};
        
    \draw[fill=OrangeRed!42]  (2, -1.5) circle (0.07cm);
    \draw[fill=OrangeRed!42]  (1.75, -1.5) circle (0.07cm);
    \draw[fill=OrangeRed!42]  (2.25, -1.5) circle (0.07cm);
    \draw[fill=OrangeRed!42]  (1.5, -1.5) circle (0.07cm);
    \draw[fill=OrangeRed!42]  (2.5, -1.5) circle (0.07cm);
    
    \draw[fill=OrangeRed!42]  (6, -1.5) circle (0.07cm);
    \draw[fill=OrangeRed!42]  (5.75, -1.5) circle (0.07cm);
    \draw[fill=OrangeRed!42]  (6.25, -1.5) circle (0.07cm);
    \draw[fill=OrangeRed!42]  (5.5, -1.5) circle (0.07cm);
    \draw[fill=OrangeRed!42]  (6.5, -1.5) circle (0.07cm);

    \draw[fill=OrangeRed!42]  (10, -1.5) circle (0.07cm);
    \draw[fill=OrangeRed!42]  (9.75, -1.5) circle (0.07cm);
    \draw[fill=OrangeRed!42]  (10.25, -1.5) circle (0.07cm);
    \draw[fill=OrangeRed!42]  (9.5, -1.5) circle (0.07cm);
    \draw[fill=OrangeRed!42]  (10.5, -1.5) circle (0.07cm);

    \node at (0,7) {};
    
    
    \draw[thick] (3.5,3.5) -- (3.5,2.5);
    \draw [thick,decorate,decoration={brace,amplitude=4pt}]  (0.5,4) -- (3.5,4);
    \node[scale=0.7] at (2,4.5) {Dominating Set};

    }
    \end{tikzpicture}
    \caption{This figure depicts a possible initial position of $n+k+1$ guards in a reduced instance of \EVC{} that is based on a \textsc{Yes}-instance of \RBDS{}. Note that in the reduced instance we may deploy up to $\ell = n + k + 2$ guards, and the position of the last guard (not shown) determines the type of the vertex cover that we are working with.}
    \label{fig:my_label2}
\end{figure}
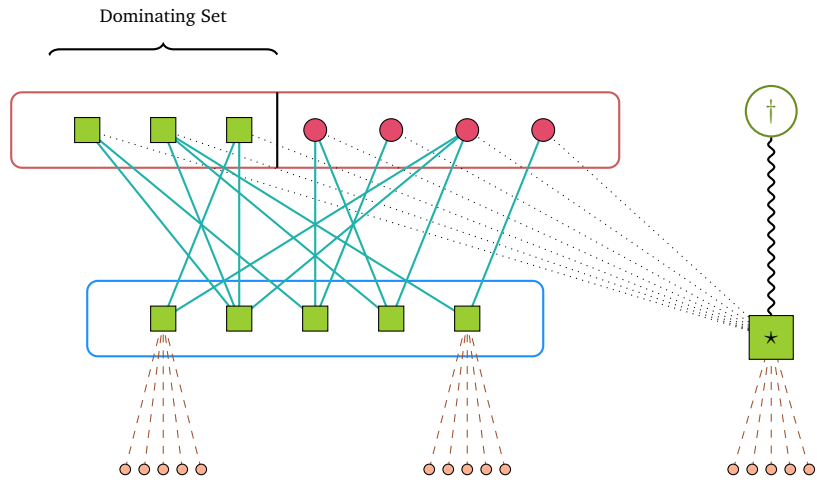


\begin{claim}
  Let $\mathcal{S} \subseteq V(H)$ be a nice vertex cover, and let $e \in E(H)$. If the guards are currently occupying $\mathcal{S}$, there is a legal move that defends $e$ such that the new position of the guards also corresponds to a nice vertex cover. 
\end{claim}

\begin{proof}

  We argue this by an exhaustive case analysis on $(\mathcal{S},e)$.

  \paragraph*{Case 1. $\mathcal{S}$ is a backup vertex cover and $e$ is a structural edge.}

  Let $e = (v_q,u_p)$. If $q \leq k$, then the guards on $v_q$ and $u_p$ exchange positions. Otherwise, if $q>k$, there exists $r\leq k$ such that $(u_p,v_r)\in E(G)$ because $S$ dominates $B$. We perform the following sequence of moves:

  \begin{itemize}
    \item The guard on $u_p$ moves to $v_q$ along the structural edge that was attacked;
    \item the guard on $v_r$ moves to $u_p$ along a structural edge; \item the guard on the universal vertex moves to $v_r$ along a supplier edge; and 
    \item the guard on $\dagger$ moves to $\star$ along the bridge.
  \end{itemize}
  
  The new configuration corresponds to a red vertex cover.

  \paragraph*{Case 2. $\mathcal{S}$ is a backup vertex cover and $e$ is a sliding edge.}

Let $i \in [b]$ and suppose $e = (u_i,w)$ for some $w \in C_i$. 
Let $q \in [k]$ be such that $(v_q,u_i) \in E(G)$. We now perform the following sequence of moves:
  \begin{itemize}
    \item the guard on $u_i$ moves to $w$ along the sliding edge that was attacked;
    \item the guard on $v_q$ moves to $u_i$ along a structural edge;
    \item the guard on the universal vertex moves to $v_q$ along a supplier edge; and 
    \item the guard on $\dagger$ moves to $\star$ along the bridge.
  \end{itemize}

Suppose $e=(\star,w)$ is the sliding edge which is attacked, then the guard on $\star$ moves to $w$ and the guard on $\dagger$ moves to $\star$.
   
The new configuration corresponds to a dependent vertex cover (c.f.~\Cref{fig:case2bip}).

\begin{figure}
    \centering
    \begin{subfigure}[b]{\textwidth}
    \begin{tikzpicture}[square/.style={regular polygon,regular polygon sides=4}]
    \tikz{

    \draw [DodgerBlue,thick,rounded corners] (1,0) rectangle (7,1);
    \draw [IndianRed,thick,rounded corners] (0,2.5) rectangle (8,3.5);
    
    
    \node at (10,0.25) [square,draw,fill=YellowGreen] (global) {$\star$};
    \node at (10,3.25) [square,draw,fill=YellowGreen] (backup) {$\dagger$};

    \tikzset{decoration={snake,amplitude=.4mm,segment length=2mm,
                       post length=0mm,pre length=0mm}}
    
    \draw [thick,decorate] (global) -- (backup);

    \foreach \x in {1,...,7}
        \draw[dotted] (global) -- (\x, 3);
    
    \draw[Sienna,dashed] (2,0.5) -- (2,-1.5);
    \draw[Sienna,dashed]  (1.75, -1.5) -- (2,0.5);
    \draw[Sienna,dashed]  (2.25, -1.5) -- (2,0.5);
    \draw[MediumVioletRed,snake=zigzag,thick]  (1.5, -1.5) -- (2,0.5);
    \draw[Sienna,dashed]  (2.5, -1.5) -- (2,0.5);
    
    \draw[Sienna,dashed]  (6, -1.5) -- (6,0.5);
    \draw[Sienna,dashed]  (5.75, -1.5) -- (6,0.5);
    \draw[Sienna,dashed]  (6.25, -1.5) -- (6,0.5);
    \draw[Sienna,dashed]  (5.5, -1.5) -- (6,0.5);
    \draw[Sienna,dashed]  (6.5, -1.5) -- (6,0.5);
    

    \draw[Sienna,dashed]  (10, -1.5) -- (global);
    \draw[Sienna,dashed]  (9.75, -1.5) -- (global);
    \draw[Sienna,dashed]  (10.25, -1.5) -- (global);
    \draw[Sienna,dashed]  (9.5, -1.5) -- (global);
    \draw[Sienna,dashed]  (10.5, -1.5) -- (global);

    
    \begin{scope}[xshift=-0.42cm]
    \draw [->,Red,thick] (2,0.4) to [out=190,in=170] (1.5, -1.5);
    
    \draw [->,Red,thick] (3,3) to [out=190,in=170] (2,0.6);
    \end{scope}
    
    \draw [->,Red,thick] (global) to (3.25,2.9);
    
    \draw [->,Red,thick,transform canvas={xshift = 0.5cm}] (backup) to (global);
    
    \node [circle,draw,thick,White,fill=Red] at (0,-1) {1};
    \node [circle,draw,thick,White,fill=Red] at (0.5,1.5) {2};
    \node [circle,draw,thick,White,fill=Red] at (8.1,1.5) {3};
    \node [circle,draw,thick,White,fill=Red] at (11,1.5) {4};
    
    
    \draw[thick,LightSeaGreen!42] (2,0.5) -- (3,3); 
    \draw[thick,LightSeaGreen!42] (6,0.5) -- (7,3);
    \draw[thick,LightSeaGreen!42] (4,0.5) -- (4,3);
    \draw[thick,LightSeaGreen!42] (3,0.5) -- (2,3);
    \draw[thick,LightSeaGreen!42] (5,0.5) -- (4,3);
    \draw[thick,LightSeaGreen!42] (5,0.5) -- (6,3);
    \draw[thick,LightSeaGreen!42] (2,0.5) -- (6,3);
    \draw[thick,LightSeaGreen!42] (6,0.5) -- (2,3);
    \draw[thick,LightSeaGreen!42] (3,0.5) -- (6,3);
    \draw[thick,LightSeaGreen!42] (3,0.5) -- (1,3);
    \draw[thick,LightSeaGreen!42] (4,0.5) -- (5,3);
    \draw[thick,LightSeaGreen!42] (3,0.5) -- (3,3);
    \draw[thick,LightSeaGreen!42] (4,0.5) -- (1,3);
    \draw[thick,LightSeaGreen!42] (5,0.5) -- (2,3);

    \foreach \x in {2,...,6}
        \node at (\x, 0.5) [square,draw,fill=YellowGreen] () {};
    
    \foreach \x in {4,...,7}
        \draw[fill=Crimson!77] (\x, 3) circle (0.15cm);

    \foreach \x in {1,2,3}
        \node at (\x, 3) [square,draw,fill=YellowGreen] () {};
        
    \draw[fill=OrangeRed!42]  (2, -1.5) circle (0.07cm);
    \draw[fill=OrangeRed!42]  (1.75, -1.5) circle (0.07cm);
    \draw[fill=OrangeRed!42]  (2.25, -1.5) circle (0.07cm);
    \draw[fill=OrangeRed!42]  (1.5, -1.5) circle (0.07cm);
    \draw[fill=OrangeRed!42]  (2.5, -1.5) circle (0.07cm);
    
    \draw[fill=OrangeRed!42]  (6, -1.5) circle (0.07cm);
    \draw[fill=OrangeRed!42]  (5.75, -1.5) circle (0.07cm);
    \draw[fill=OrangeRed!42]  (6.25, -1.5) circle (0.07cm);
    \draw[fill=OrangeRed!42]  (5.5, -1.5) circle (0.07cm);
    \draw[fill=OrangeRed!42]  (6.5, -1.5) circle (0.07cm);

    \draw[fill=OrangeRed!42]  (10, -1.5) circle (0.07cm);
    \draw[fill=OrangeRed!42]  (9.75, -1.5) circle (0.07cm);
    \draw[fill=OrangeRed!42]  (10.25, -1.5) circle (0.07cm);
    \draw[fill=OrangeRed!42]  (9.5, -1.5) circle (0.07cm);
    \draw[fill=OrangeRed!42]  (10.5, -1.5) circle (0.07cm);

    \node at (0,7) {};
    
    
    \draw[thick] (3.5,3.5) -- (3.5,2.5);
    \draw [thick,decorate,decoration={brace,amplitude=4pt}]  (0.5,4) -- (3.5,4);
    \node[scale=0.7] at (2,4.5) {Dominating Set};
    
    }
    \end{tikzpicture}
    \caption{This figure demonstrates a defense for when $\mathcal{S}$ is a backup vertex cover and a sliding edge is attacked.}
    \label{fig:case2a}
\end{subfigure}

\begin{subfigure}[b]{\textwidth}
    \centering
    \begin{tikzpicture}[square/.style={regular polygon,regular polygon sides=4}]
    \tikz{ 
    
    \draw [DodgerBlue,thick,rounded corners] (1,0) rectangle (7,1);
    \draw [IndianRed,thick,rounded corners] (0,2.5) rectangle (8,3.5);
    
    \node [circle,draw,thick,OliveDrab,fill=white] (backup) at (10,3.25) {$\dagger$}; 
    
    \node at (10,0.25) [square,draw,fill=YellowGreen] (global) {$\star$};

    \tikzset{decoration={snake,amplitude=.4mm,segment length=2mm,
                       post length=0mm,pre length=0mm}}
    
    \draw [thick,decorate] (global) -- (backup);

    \foreach \x in {1,...,7}
        \draw[dotted] (global) -- (\x, 3);
    
    \draw[Sienna,dashed] (2,0.5) -- (2,-1.5);
    \draw[Sienna,dashed]  (1.75, -1.5) -- (2,0.5);
    \draw[Sienna,dashed]  (2.25, -1.5) -- (2,0.5);
    \draw[Gray,snake=zigzag,thick]  (1.25, -1.5) -- (2,0.5);
    \draw[Sienna,dashed]  (2.5, -1.5) -- (2,0.5);
    
    \draw[Sienna,dashed]  (6, -1.5) -- (6,0.5);
    \draw[Sienna,dashed]  (5.75, -1.5) -- (6,0.5);
    \draw[Sienna,dashed]  (6.25, -1.5) -- (6,0.5);
    \draw[Sienna,dashed]  (5.5, -1.5) -- (6,0.5);
    \draw[Sienna,dashed]  (6.5, -1.5) -- (6,0.5);
    

    \draw[Sienna,dashed]  (10, -1.5) -- (global);
    \draw[Sienna,dashed]  (9.75, -1.5) -- (global);
    \draw[Sienna,dashed]  (10.25, -1.5) -- (global);
    \draw[Sienna,dashed]  (9.5, -1.5) -- (global);
    \draw[Sienna,dashed]  (10.5, -1.5) -- (global);
    

    \draw[thick,LightSeaGreen] (2,0.5) -- (3,3); 
    \draw[thick,LightSeaGreen] (6,0.5) -- (7,3);
    \draw[thick,LightSeaGreen] (4,0.5) -- (4,3);
    \draw[thick,LightSeaGreen] (3,0.5) -- (2,3);
    \draw[thick,LightSeaGreen] (5,0.5) -- (4,3);
    \draw[thick,LightSeaGreen] (5,0.5) -- (6,3);
    \draw[thick,LightSeaGreen] (2,0.5) -- (6,3);
    \draw[thick,LightSeaGreen] (6,0.5) -- (2,3);
    \draw[thick,LightSeaGreen] (3,0.5) -- (6,3);
    \draw[thick,LightSeaGreen] (3,0.5) -- (1,3);
    \draw[thick,LightSeaGreen] (4,0.5) -- (5,3);
    \draw[thick,LightSeaGreen] (3,0.5) -- (3,3);
    \draw[thick,LightSeaGreen] (4,0.5) -- (1,3);
    \draw[thick,LightSeaGreen] (5,0.5) -- (2,3);

    \foreach \x in {2,...,6}
        \node at (\x, 0.5) [square,draw,fill=YellowGreen] () {};
    
    \foreach \x in {4,...,7}
        \draw[fill=Crimson!77] (\x, 3) circle (0.15cm);

    \foreach \x in {1,2,3}
        \node at (\x, 3) [square,draw,fill=YellowGreen] () {};
        
    \draw[fill=OrangeRed!42]  (2, -1.5) circle (0.07cm);
    \draw[fill=OrangeRed!42]  (1.75, -1.5) circle (0.07cm);
    \draw[fill=OrangeRed!42]  (2.25, -1.5) circle (0.07cm);
    
    \node at (1.25, -1.5) [square,draw,fill=YellowGreen] () {};
    
    \draw[fill=OrangeRed!42]  (2.5, -1.5) circle (0.07cm);
    
    \draw[fill=OrangeRed!42]  (6, -1.5) circle (0.07cm);
    \draw[fill=OrangeRed!42]  (5.75, -1.5) circle (0.07cm);
    \draw[fill=OrangeRed!42]  (6.25, -1.5) circle (0.07cm);
    \draw[fill=OrangeRed!42]  (5.5, -1.5) circle (0.07cm);
    \draw[fill=OrangeRed!42]  (6.5, -1.5) circle (0.07cm);

    \draw[fill=OrangeRed!42]  (10, -1.5) circle (0.07cm);
    \draw[fill=OrangeRed!42]  (9.75, -1.5) circle (0.07cm);
    \draw[fill=OrangeRed!42]  (10.25, -1.5) circle (0.07cm);
    \draw[fill=OrangeRed!42]  (9.5, -1.5) circle (0.07cm);
    \draw[fill=OrangeRed!42]  (10.5, -1.5) circle (0.07cm);

    \node at (0,7) {};
    
    
    \draw[thick] (3.5,3.5) -- (3.5,2.5);
    \draw [thick,decorate,decoration={brace,amplitude=4pt}]  (0.5,4) -- (3.5,4);
    \node[scale=0.7] at (2,4.5) {Dominating Set};

    }
    \end{tikzpicture}
    \caption{This figure demonstrates a the positions of the guards after the defence in ~\Cref{fig:case2a} is executed.}
    \label{fig:case2b}
    \end{subfigure}
    \caption{This figure demonstrates the case when $\mathcal{S}$ is a backup vertex cover and a sliding edge is attacked.}
    \label{fig:case2bip}
\end{figure}
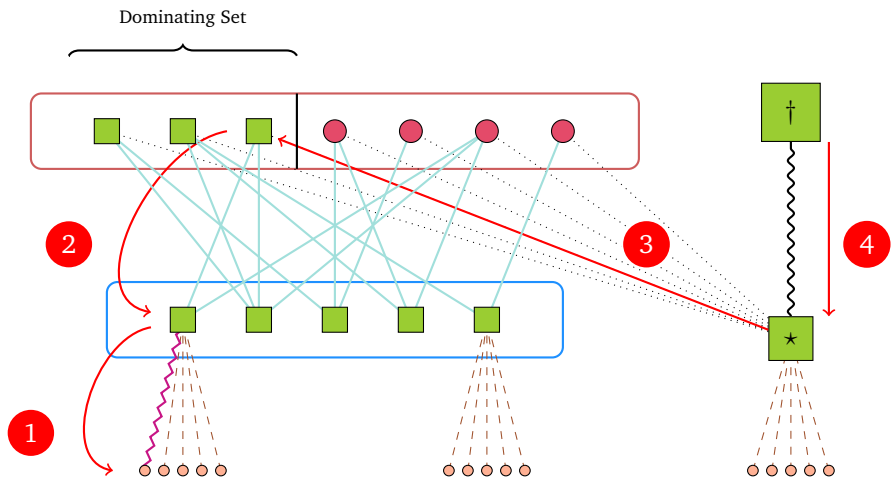
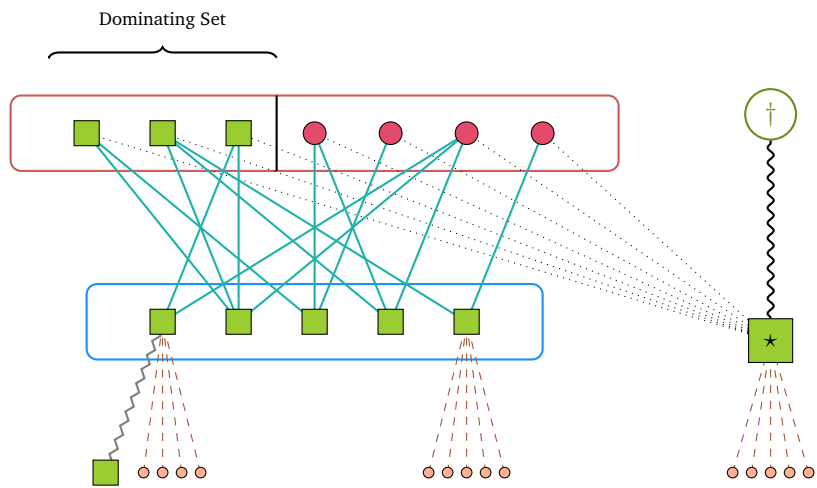

\paragraph*{Case 3. $\mathcal{S}$ is a backup vertex cover and $e$ is a supplier edge.}

  Let $e = (v_i,\star)$. If $i \leq k$, the guards on $v_i$ and $\star$ exchange positions. If $i > k$, we have the guard on the universal vertex move to $v_i$ via the supplier edge that was attacked. The guard on the backup vertex then moves to $\star$ along the bridge, leading to a red vertex cover. 

\paragraph*{Case 4. $\mathcal{S}$ is a backup vertex cover and $e$ is a bridge edge.}

  The guards on $\dagger$ and $\star$ exchange places and the resulting configuration is again a backup vertex cover. 

\paragraph*{Case 5. $\mathcal{S}$ is a red vertex cover and $e$ is a structural edge.}
  
   Let $e = (v_q,u_p)$. If $q \leq k$, then the guards on $v_q$ and $u_p$ exchange positions. If $q > k$, and if $v_q$ has a guard, then the guards on $u_p$ and $v_q$ exchange positions. Otherwise let $j\neq q$ be such that $v_j$ has a guard. Note that $j > k$. 
   Let $r \in [k]$ be such that $(v_r,u_p) \in E(G)$.
  \begin{itemize}
    \item The guard on $u_p$ moves to $v_q$ along the structural edge that was attacked;
    \item the guard on $v_r$ moves to $u_p$ along a structural edge; \item the guard on the universal vertex moves to $v_r$ along a supplier edge; and 
    \item the guard on $v_j$ moves to $\star$ along a supplier edge.
  \end{itemize}

  Note that this new configuration of guards corresponds to a red vertex cover (c.f.~\Cref{fig:case5bip}).

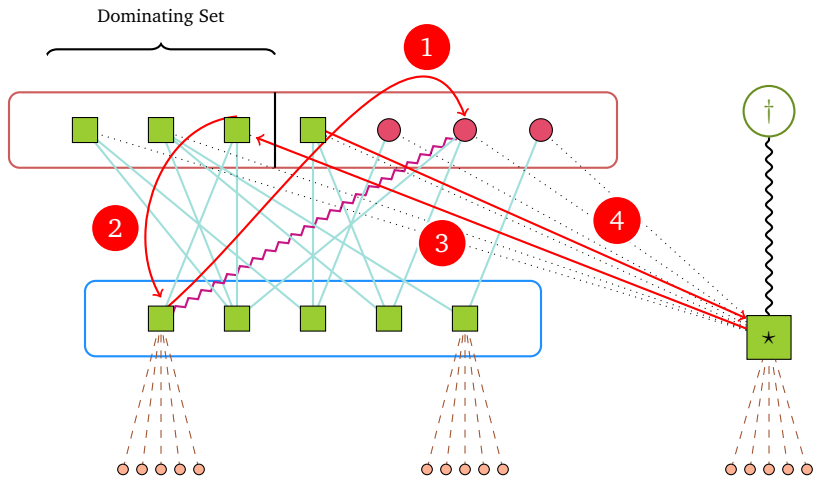
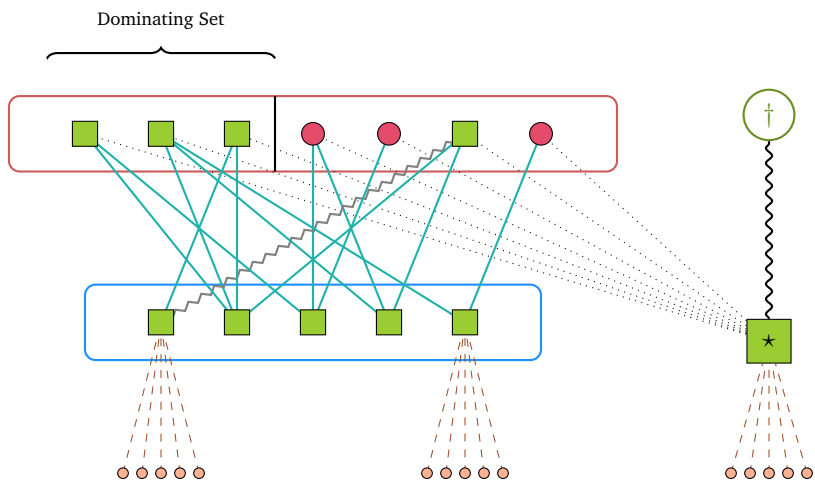
\begin{figure}
    \centering
    \begin{subfigure}[b]{\textwidth}
    \begin{tikzpicture}[square/.style={regular polygon,regular polygon sides=4}]
    \tikz{

    \draw [DodgerBlue,thick,rounded corners] (1,0) rectangle (7,1);
    \draw [IndianRed,thick,rounded corners] (0,2.5) rectangle (8,3.5);
    
    \node [circle,draw,thick,OliveDrab,fill=white] (backup) at (10,3.25) {$\dagger$}; 
    
    \node at (10,0.25) [square,draw,fill=YellowGreen] (global) {$\star$};

    \tikzset{decoration={snake,amplitude=.4mm,segment length=2mm,
                       post length=0mm,pre length=0mm}}
    
    \draw [thick,decorate] (global) -- (backup);
    
    
    \draw[thick,LightSeaGreen!42] (2,0.5) -- (3,3); 
    \draw[thick,LightSeaGreen!42] (6,0.5) -- (7,3);
    \draw[thick,LightSeaGreen!42] (4,0.5) -- (4,3);
    \draw[thick,LightSeaGreen!42] (3,0.5) -- (2,3);
    \draw[thick,LightSeaGreen!42] (5,0.5) -- (4,3);
    \draw[thick,LightSeaGreen!42] (5,0.5) -- (6,3);
    \draw[MediumVioletRed,snake=zigzag,thick] (2,0.5) -- (6,3);
    \draw[thick,LightSeaGreen!42] (6,0.5) -- (2,3);
    \draw[thick,LightSeaGreen!42] (3,0.5) -- (6,3);
    \draw[thick,LightSeaGreen!42] (3,0.5) -- (1,3);
    \draw[thick,LightSeaGreen!42] (4,0.5) -- (5,3);
    \draw[thick,LightSeaGreen!42] (3,0.5) -- (3,3);
    \draw[thick,LightSeaGreen!42] (4,0.5) -- (1,3);
    \draw[thick,LightSeaGreen!42] (5,0.5) -- (2,3);

    \foreach \x in {1,...,7}
        \draw[dotted] (global) -- (\x, 3);
    
    \draw[Sienna,dashed] (2,0.5) -- (2,-1.5);
    \draw[Sienna,dashed]  (1.75, -1.5) -- (2,0.5);
    \draw[Sienna,dashed]  (2.25, -1.5) -- (2,0.5);
    \draw[Sienna,dashed]  (1.5, -1.5) -- (2,0.5);
    \draw[Sienna,dashed]  (2.5, -1.5) -- (2,0.5);
    
    \draw[Sienna,dashed]  (6, -1.5) -- (6,0.5);
    \draw[Sienna,dashed]  (5.75, -1.5) -- (6,0.5);
    \draw[Sienna,dashed]  (6.25, -1.5) -- (6,0.5);
    \draw[Sienna,dashed]  (5.5, -1.5) -- (6,0.5);
    \draw[Sienna,dashed]  (6.5, -1.5) -- (6,0.5);
    

    \draw[Sienna,dashed]  (10, -1.5) -- (global);
    \draw[Sienna,dashed]  (9.75, -1.5) -- (global);
    \draw[Sienna,dashed]  (10.25, -1.5) -- (global);
    \draw[Sienna,dashed]  (9.5, -1.5) -- (global);
    \draw[Sienna,dashed]  (10.5, -1.5) -- (global);

    
    \begin{scope}[yshift=0.18cm]
   \draw [->,Red,thick] (2,0.4) to [out=40,in=110] (6,3);
    
   \draw [->,Red,thick] (3,3) to [out=185,in=120] (2,0.6);
    \end{scope}
    
    \draw [->,Red,thick] (global) to (3.25,2.9);
    
    \draw [->,Red,thick,transform canvas={xshift=-0.1cm,yshift = 0.1}] (4,3.1) to (9.8,0.5);
    
    \node [circle,draw,thick,White,fill=Red] at (1.4,1.7) {2};
    \node [circle,draw,thick,White,fill=Red] at (5.5,4.1) {1};
    \node [circle,draw,thick,White,fill=Red] at (5.7,1.5) {3};
    \node [circle,draw,thick,White,fill=Red] at (8,1.8) {4};

    \foreach \x in {2,...,6}
        \node at (\x, 0.5) [square,draw,fill=YellowGreen] () {};
    
    \foreach \x in {5,...,7}
        \draw[fill=Crimson!77] (\x, 3) circle (0.15cm);

    \foreach \x in {1,2,3,4}
        \node at (\x, 3) [square,draw,fill=YellowGreen] () {};
        
    \draw[fill=OrangeRed!42]  (2, -1.5) circle (0.07cm);
    \draw[fill=OrangeRed!42]  (1.75, -1.5) circle (0.07cm);
    \draw[fill=OrangeRed!42]  (2.25, -1.5) circle (0.07cm);
    \draw[fill=OrangeRed!42]  (1.5, -1.5) circle (0.07cm);
    \draw[fill=OrangeRed!42]  (2.5, -1.5) circle (0.07cm);
    
    \draw[fill=OrangeRed!42]  (6, -1.5) circle (0.07cm);
    \draw[fill=OrangeRed!42]  (5.75, -1.5) circle (0.07cm);
    \draw[fill=OrangeRed!42]  (6.25, -1.5) circle (0.07cm);
    \draw[fill=OrangeRed!42]  (5.5, -1.5) circle (0.07cm);
    \draw[fill=OrangeRed!42]  (6.5, -1.5) circle (0.07cm);

    \draw[fill=OrangeRed!42]  (10, -1.5) circle (0.07cm);
    \draw[fill=OrangeRed!42]  (9.75, -1.5) circle (0.07cm);
    \draw[fill=OrangeRed!42]  (10.25, -1.5) circle (0.07cm);
    \draw[fill=OrangeRed!42]  (9.5, -1.5) circle (0.07cm);
    \draw[fill=OrangeRed!42]  (10.5, -1.5) circle (0.07cm);

    \node at (0,7) {};
    
    
    \draw[thick] (3.5,3.5) -- (3.5,2.5);
    \draw [thick,decorate,decoration={brace,amplitude=4pt}]  (0.5,4) -- (3.5,4);
    \node[scale=0.7] at (2,4.5) {Dominating Set};
    
    }
    \end{tikzpicture}
    \caption{This figure demonstrates a defense for when $\mathcal{S}$ is a red vertex cover and a structural edge is attacked.}
    \label{fig:case2c}
\end{subfigure}
\begin{subfigure}[b]{\textwidth}
    \begin{tikzpicture}[square/.style={regular polygon,regular polygon sides=4}]
    \tikz{ 
    
    \draw [DodgerBlue,thick,rounded corners] (1,0) rectangle (7,1);
    \draw [IndianRed,thick,rounded corners] (0,2.5) rectangle (8,3.5);
    
    \node [circle,draw,thick,OliveDrab,fill=white] (backup) at (10,3.25) {$\dagger$}; 
    
    \node at (10,0.25) [square,draw,fill=YellowGreen] (global) {$\star$};

    \tikzset{decoration={snake,amplitude=.4mm,segment length=2mm,
                       post length=0mm,pre length=0mm}}
    
    \draw [thick,decorate] (global) -- (backup);

    \foreach \x in {1,...,7}
        \draw[dotted] (global) -- (\x, 3);
    
    \draw[Sienna,dashed] (2,0.5) -- (2,-1.5);
    \draw[Sienna,dashed]  (1.75, -1.5) -- (2,0.5);
    \draw[Sienna,dashed]  (2.25, -1.5) -- (2,0.5);
    \draw[Sienna,dashed]  (1.5, -1.5) -- (2,0.5);
    \draw[Sienna,dashed]  (2.5, -1.5) -- (2,0.5);
    
    \draw[Sienna,dashed]  (6, -1.5) -- (6,0.5);
    \draw[Sienna,dashed]  (5.75, -1.5) -- (6,0.5);
    \draw[Sienna,dashed]  (6.25, -1.5) -- (6,0.5);
    \draw[Sienna,dashed]  (5.5, -1.5) -- (6,0.5);
    \draw[Sienna,dashed]  (6.5, -1.5) -- (6,0.5);
    

    \draw[Sienna,dashed]  (10, -1.5) -- (global);
    \draw[Sienna,dashed]  (9.75, -1.5) -- (global);
    \draw[Sienna,dashed]  (10.25, -1.5) -- (global);
    \draw[Sienna,dashed]  (9.5, -1.5) -- (global);
    \draw[Sienna,dashed]  (10.5, -1.5) -- (global);
    

    \draw[thick,LightSeaGreen] (2,0.5) -- (3,3); 
    \draw[thick,LightSeaGreen] (6,0.5) -- (7,3);
    \draw[thick,LightSeaGreen] (4,0.5) -- (4,3);
    \draw[thick,LightSeaGreen] (3,0.5) -- (2,3);
    \draw[thick,LightSeaGreen] (5,0.5) -- (4,3);
    \draw[thick,LightSeaGreen] (5,0.5) -- (6,3);
    \draw[Gray,snake=zigzag,thick] (2,0.5) -- (6,3);
    \draw[thick,LightSeaGreen] (6,0.5) -- (2,3);
    \draw[thick,LightSeaGreen] (3,0.5) -- (6,3);
    \draw[thick,LightSeaGreen] (3,0.5) -- (1,3);
    \draw[thick,LightSeaGreen] (4,0.5) -- (5,3);
    \draw[thick,LightSeaGreen] (3,0.5) -- (3,3);
    \draw[thick,LightSeaGreen] (4,0.5) -- (1,3);
    \draw[thick,LightSeaGreen] (5,0.5) -- (2,3);

    \foreach \x in {2,...,6}
        \node at (\x, 0.5) [square,draw,fill=YellowGreen] () {};
    
    \foreach \x in {4,5,7}
        \draw[fill=Crimson!77] (\x, 3) circle (0.15cm);

    \foreach \x in {1,2,3,6}
        \node at (\x, 3) [square,draw,fill=YellowGreen] () {};
        
    \draw[fill=OrangeRed!42]  (2, -1.5) circle (0.07cm);
    \draw[fill=OrangeRed!42]  (1.75, -1.5) circle (0.07cm);
    \draw[fill=OrangeRed!42]  (2.25, -1.5) circle (0.07cm);
    
    \draw[fill=OrangeRed!42]  (1.5, -1.5) circle (0.07cm);

    \draw[fill=OrangeRed!42]  (2.5, -1.5) circle (0.07cm);
    
    \draw[fill=OrangeRed!42]  (6, -1.5) circle (0.07cm);
    \draw[fill=OrangeRed!42]  (5.75, -1.5) circle (0.07cm);
    \draw[fill=OrangeRed!42]  (6.25, -1.5) circle (0.07cm);
    \draw[fill=OrangeRed!42]  (5.5, -1.5) circle (0.07cm);
    \draw[fill=OrangeRed!42]  (6.5, -1.5) circle (0.07cm);

    \draw[fill=OrangeRed!42]  (10, -1.5) circle (0.07cm);
    \draw[fill=OrangeRed!42]  (9.75, -1.5) circle (0.07cm);
    \draw[fill=OrangeRed!42]  (10.25, -1.5) circle (0.07cm);
    \draw[fill=OrangeRed!42]  (9.5, -1.5) circle (0.07cm);
    \draw[fill=OrangeRed!42]  (10.5, -1.5) circle (0.07cm);

    \node at (0,7) {};
    
    
    \draw[thick] (3.5,3.5) -- (3.5,2.5);
    \draw [thick,decorate,decoration={brace,amplitude=4pt}]  (0.5,4) -- (3.5,4);
    \node[scale=0.7] at (2,4.5) {Dominating Set};

    }
    \end{tikzpicture}
    \caption{This figure demonstrates the positions of the guards after the defence in ~\Cref{fig:case2c} is executed.}
    \label{fig:case2d}
    \end{subfigure}
    \caption{This figure demonstrates the case when $\mathcal{S}$ is a red vertex cover and a structural edge is attacked.}
    \label{fig:case5bip}
\end{figure}

  \paragraph*{Case 6. $\mathcal{S}$ is a red vertex cover and $e$ is a sliding edge.}
  
   Let $i \in [b]$ and suppose $e = (u_i,w)$ for some $w \in C_i$. Let $r \in [k]$ be such that $(v_r,u_i) \in E(G)$. Let $j>k$ such that $v_j$ has a guard. Then we have the following sequence of moves. 

  \begin{itemize}
    \item the guard on $u_i$ moves to $w$ along the sliding edge that was attacked;
    \item the guard on $v_r$ moves to $u_i$ along a structural edge;
    \item the guard on the universal vertex moves to $v_r$ along a supplier edge; and 
    \item the guard on $v_j$ moves to $\star$ along a supplier edge.
  \end{itemize}
  
  Now suppose $e = (\star,w)$ for some $w \in D$. Let $j>k$ such that $v_j$ has a guard.
  
  \begin{itemize}
    \item the guard on $\star$ moves to $w$ along the sliding edge that was attacked and
    \item the guard on $v_j$ moves to $\star$ along a supplier edge.
  \end{itemize}

  In both cases, the new configuration of the guards corresponds to a dependent vertex cover.
  
  \paragraph*{Case 7. $\mathcal{S}$ is a red vertex cover and $e$ is a supplier edge.}
  
  Let $e = (v_i,\star)$. If $i \leq k$, the guards on $v_i$ and $\star$ exchange positions. If $i > k$ and $v_i$ has a guard, the guard on $v_i$ and $\star$ exchange positions. If $i>k$ and $v_i$ does not have a guard, let $j>k$ and $j\neq i$ be such that $v_j$ has a guard. We have the guard on the universal vertex move to $v_i$ via the supplier edge that was attacked. The guard on $v_j$ then moves to $\star$ along a supplier edge, leading to another red vertex cover. 

  \paragraph*{Case 8. $\mathcal{S}$ is a red vertex cover and $e$ is a bridge edge.}

  Let $j>k$ be such that $v_j$ has a guard. The guard on $\star$ moves to $\dagger$ along the attacked bridge edge. The guard on $v_j$ moves to $\star$ along a supplier edge and the resulting configuration is a backup vertex cover. 

  \paragraph*{Case 9. $\mathcal{S}$ is a dependent vertex cover and $e$ is a structural edge.}
  
  Suppose $w\in C_j$ is such that $w$ has a guard.
  Suppose $e=(v_q,u_j)$. If $q<k$, guards on $v_q$ and $u_j$ exchange and we have the same dependent vertex cover. If $q>k$, guard on $u_j$ moves to $v_q$ along the attacked structural edge and guard on $w$ moves to $u_j$ along a sliding edge. Thus we have a red vertex cover.
  
  Let $e = (v_q,u_p)$ for some $p\neq j$. If $q \leq k$, then the guards on $v_q$ and $u_p$ exchange positions and we have the same dependent vertex cover. 
  If $q > k$, let $r \in [k]$ be such that $(v_r,u_p) \in E(G)$. If $(v_r,u_j)\in E(G)$, we perform the following sequence of moves:

  \begin{itemize}
    \item The guard on $u_p$ moves to $v_q$ along the structural edge that was attacked;
    \item the guard on $v_r$ moves to $u_p$ along a structural edge;
    \item the guard on $u_j$ moves to $v_r$ along a structural edge; and
    \item the guard on $w$ moves to $u_j$ along a sliding edge.
  \end{itemize}

  Note that this new configuration of guards corresponds to a red vertex cover.
   
  If $(v_r,u_j)\notin E(G)$, there exists $s\in [k]$ such that $(v_r,u_s) \in E(G)$. We perform the following sequence of moves:
  \begin{itemize}
    \item The guard on $u_p$ moves to $v_q$ along the structural edge that was attacked;
    \item the guard on $v_r$ moves to $u_p$ along a structural edge;
    \item the guard on $\star$ moves to $v_r$ along a supplier edge;
    \item the guard on $v_s$ moves to $\star$ along a supplier edge;
    \item the guard on $u_j$ moves to $v_s$ along a structural edge;
    and
    \item the guard on $w$ moves to $u_j$ along a sliding edge.
  \end{itemize}

  Note that this new configuration of guards corresponds to a red vertex cover.
  
  Suppose $w\in D$ is such that $w$ has a guard. Let $e = (v_q,u_p)$. If $q \leq k$, then the guards on $v_q$ and $u_p$ exchange positions and we have the same dependent vertex cover.  If $q > k$, let $r \in [k]$ be such that $(v_r,u_p) \in E(G)$. We perform the following sequence of moves:

  \begin{itemize}
    \item The guard on $u_p$ moves to $v_q$ along the structural edge that was attacked;
    \item the guard on $v_r$ moves to $u_p$ along a structural edge;
    \item the guard on $\star$ moves to $v_r$ along a supplier edge; and
    \item the guard on $w$ moves to $\star$ along a sliding edge.
  \end{itemize}
  
  Note that this new configuration of guards corresponds to a red vertex cover.
  
\paragraph*{Case 10. $\mathcal{S}$ is a dependent vertex cover and $e$ is a sliding edge.}
  
  Let $w\in C_j$ be such that $w$ has a guard.
  Suppose $e= (u_j,w)$ is the sliding edge which was attacked, then the guards exchange their positions and we have the same dependent vertex cover.
  Suppose $e= (u_j,z)$ for some $z\neq w$ is the sliding edge which was attacked, then the guard on $u_j$ moves to $z$ to defend the attack and the guard from $w$ moves to $u_j$, leading to a different dependent vertex cover. 
  Let $i \in [b]$ and suppose $e = (u_i,z)$ for some $z \in C_i$.  Let $q \in [k]$ be such that $(v_q,u_i) \in E(G)$ and $r\in [k]$ such that $(v_r,u_j)\in E(G)$. If $r\neq q$, we perform the following sequence of moves:
  
  \begin{itemize}
    \item the guard on $u_i$ moves to $z$ along the sliding edge that was attacked;
    \item the guard on $v_q$ moves to $u_i$ along a structural edge;
    \item the guard on the universal vertex moves to $v_q$ along a supplier edge; 
    \item the guard on $v_r$ moves to $\star$ along a supplier edge;
    \item the guard on $u_j$ moves to $v_r$ along a structural edge; and
    \item the guard on $w$ moves to $u_j$ along a sliding edge.
  \end{itemize}
  Note that the new configuration corresponds to a dependent vertex cover (c.f.~\Cref{fig:case10bip}).
  
  If $u_i$ and $u_j$ are dominated by a common vertex in $G$, that is, there exists $q \in [k]$ be such that $(v_q,u_i) \in E(G)$ and $(v_q,u_j) \in E(G)$, then we perform the following sequence of moves:
  
  \begin{itemize}
    \item the guard on $u_i$ moves to $z$ along the sliding edge that was attacked;
    \item the guard on $v_q$ moves to $u_i$ along a structural edge;
    \item the guard on $u_j$ moves to $v_q$ along another structural edge;
    \item the guard on $w$ moves to $u_j$.
  \end{itemize}
 The new configuration corresponds to a dependent vertex cover. 
 
 Let $w\in D$ be such that $w$ has a guard. Suppose $e= (\star,w)$ is the sliding edge which was attacked, then the guards exchange their positions and we have the same dependent vertex cover.
  Suppose $e= (\star,z)$ for some $z\neq w$ is the sliding edge which was attacked, then the guard on $\star$ moves to $z$ to defend the attack and the guard from $w$ moves to $\star$, leading to a different dependent vertex cover. 
  Let $i \in [b]$ and suppose $e = (u_i,z)$ for some $z \in C_i$.  Let $q \in [k]$ be such that $(v_q,u_i) \in E(G)$. We perform the following sequence of moves:
  
  \begin{itemize}
    \item the guard on $u_i$ moves to $z$ along the sliding edge that was attacked;
    \item the guard on $v_q$ moves to $u_i$ along a structural edge;
    \item the guard on the universal vertex moves to $v_q$ along a supplier edge;  and
    \item the guard on $w$ moves to $\star$ along a sliding edge.
  \end{itemize}
  Note that the new configuration corresponds to a dependent vertex cover.

\begin{figure}
    \centering
    \begin{subfigure}[b]{0.8\textwidth}
    \begin{tikzpicture}[square/.style={regular polygon,regular polygon sides=4}]
    \tikz{

    \draw [DodgerBlue,thick,rounded corners] (1,0) rectangle (7,1);
    \draw [IndianRed,thick,rounded corners] (0,2.5) rectangle (8,3.5);
    
    \node [circle,draw,thick,OliveDrab,fill=white] (backup) at (10,3.25) {$\dagger$}; 
    
    \node at (10,0.25) [square,draw,fill=YellowGreen] (global) {$\star$};

    
    \draw[thick,LightSeaGreen!42] (2,0.5) -- (3,3); 
    \draw[thick,LightSeaGreen!42] (6,0.5) -- (7,3);
    \draw[thick,LightSeaGreen!42] (4,0.5) -- (4,3);
    \draw[thick,LightSeaGreen!42] (3,0.5) -- (2,3);
    \draw[thick,LightSeaGreen!42] (5,0.5) -- (4,3);
    \draw[thick,LightSeaGreen!42] (5,0.5) -- (6,3);
    \draw[thick,LightSeaGreen!42] (2,0.5) -- (6,3);
    \draw[thick,LightSeaGreen!42] (6,0.5) -- (2,3);
    \draw[thick,LightSeaGreen!42] (3,0.5) -- (6,3);
    \draw[thick,LightSeaGreen!42] (3,0.5) -- (1,3);
    \draw[thick,LightSeaGreen!42] (4,0.5) -- (5,3);
    \draw[thick,LightSeaGreen!42] (3,0.5) -- (3,3);
    \draw[thick,LightSeaGreen!42] (4,0.5) -- (1,3);
    \draw[thick,LightSeaGreen!42] (5,0.5) -- (2,3);
    
    \tikzset{decoration={snake,amplitude=.4mm,segment length=2mm,
                       post length=0mm,pre length=0mm}}
    
    \draw [thick,decorate] (global) -- (backup);

    \foreach \x in {1,...,7}
        \draw[dotted] (global) -- (\x, 3);
    
    \draw[Sienna,dashed] (2,0.5) -- (2,-1.5);
    \draw[Sienna,dashed]  (1.75, -1.5) -- (2,0.5);
    \draw[Sienna,dashed]  (2.25, -1.5) -- (2,0.5);
    \draw[Sienna,dashed]  (1.25, -1.5) -- (2,0.5);
    \draw[Sienna,dashed]  (2.5, -1.5) -- (2,0.5);
    
    \draw[Sienna,dashed]  (6, -1.5) -- (6,0.5);
    \draw[Sienna,dashed]  (5.75, -1.5) -- (6,0.5);
    \draw[Sienna,dashed]  (6.25, -1.5) -- (6,0.5);
    \draw[MediumVioletRed,snake=zigzag,thick]  (5.5, -1.5) -- (6,0.5);
    \draw[Sienna,dashed]  (6.5, -1.5) -- (6,0.5);
    

    \draw[Sienna,dashed]  (10, -1.5) -- (global);
    \draw[Sienna,dashed]  (9.75, -1.5) -- (global);
    \draw[Sienna,dashed]  (10.25, -1.5) -- (global);
    \draw[Sienna,dashed]  (9.5, -1.5) -- (global);
    \draw[Sienna,dashed]  (10.5, -1.5) -- (global);

    
    \begin{scope}[yshift=0.1cm]
   \draw [->,Red,thick] (6,0.4) to [out=60,in=120] (5.25,-1.5);
    
  \draw [->,Red,thick] (2.1,2.8) to [out=330,in=120] (6,0.6);
    \end{scope}
    
    \draw [->,Red,thick] (global) to (2.25,2.9);
    
    \draw [->,Red,thick,transform canvas={xshift=-0.1cm,yshift = 0.1}] (3,3.1) to (9.8,0.5);
    
    \draw [->,Red,thick] (2,0.5) to [out=70,in=120] (3,3.2);
    
    \draw [->,Red,thick] (1.3,-1.5) to [out=150,in=190] (1.7,0.5);
    
    \node [circle,draw,thick,White,fill=Red] at (4.8,-0.7) {1};
    \node [circle,draw,thick,White,fill=Red] at (5,1.3) {2};
    \node [circle,draw,thick,White,fill=Red] at (8,0.5) {3};
    \node [circle,draw,thick,White,fill=Red] at (9,1.2) {4};
    \node [circle,draw,thick,White,fill=Red] at (1.7,1.5) {5};
    \node [circle,draw,thick,White,fill=Red] at (0.5,-0.5) {6};

    \foreach \x in {2,...,6}
        \node at (\x, 0.5) [square,draw,fill=YellowGreen] () {};
    
    \foreach \x in {4,...,7}
        \draw[fill=Crimson!77] (\x, 3) circle (0.15cm);

    \foreach \x in {1,2,3}
        \node at (\x, 3) [square,draw,fill=YellowGreen] () {};
        
    \draw[fill=OrangeRed!42]  (2, -1.5) circle (0.07cm);
    \draw[fill=OrangeRed!42]  (1.75, -1.5) circle (0.07cm);
    \draw[fill=OrangeRed!42]  (2.25, -1.5) circle (0.07cm);
    \draw[fill=OrangeRed!42]  (2.5, -1.5) circle (0.07cm);
    \node at (1.25, -1.5) [square,draw,fill=YellowGreen] () {};
    
    \draw[fill=OrangeRed!42]  (6, -1.5) circle (0.07cm);
    \draw[fill=OrangeRed!42]  (5.75, -1.5) circle (0.07cm);
    \draw[fill=OrangeRed!42]  (6.25, -1.5) circle (0.07cm);
    \draw[fill=OrangeRed!42]  (5.5, -1.5) circle (0.07cm);
    \draw[fill=OrangeRed!42]  (6.5, -1.5) circle (0.07cm);

    \draw[fill=OrangeRed!42]  (10, -1.5) circle (0.07cm);
    \draw[fill=OrangeRed!42]  (9.75, -1.5) circle (0.07cm);
    \draw[fill=OrangeRed!42]  (10.25, -1.5) circle (0.07cm);
    \draw[fill=OrangeRed!42]  (9.5, -1.5) circle (0.07cm);
    \draw[fill=OrangeRed!42]  (10.5, -1.5) circle (0.07cm);

    \node at (0,7) {};
    
    
    \draw[thick] (3.5,3.5) -- (3.5,2.5);
    \draw [thick,decorate,decoration={brace,amplitude=4pt}]  (0.5,4) -- (3.5,4);
    \node[scale=0.7] at (2,4.5) {Dominating Set};
    
    }
    \end{tikzpicture}
    \caption{This figure demonstrates a defense for when $\mathcal{S}$ is a dependent vertex cover and a sliding edge is attacked.}
    \label{fig:case2e}
\end{subfigure}
\begin{subfigure}[b]{0.8\textwidth}
    \centering
    \begin{tikzpicture}[square/.style={regular polygon,regular polygon sides=4}]
    \tikz{ 
    
    \draw [DodgerBlue,thick,rounded corners] (1,0) rectangle (7,1);
    \draw [IndianRed,thick,rounded corners] (0,2.5) rectangle (8,3.5);
    
    \node [circle,draw,thick,OliveDrab,fill=white] (backup) at (10,3.25) {$\dagger$}; 
    
    \node at (10,0.25) [square,draw,fill=YellowGreen] (global) {$\star$};

    \tikzset{decoration={snake,amplitude=.4mm,segment length=2mm,
                       post length=0mm,pre length=0mm}}
    
    \draw [thick,decorate] (global) -- (backup);

    \foreach \x in {1,...,7}
        \draw[dotted] (global) -- (\x, 3);
    
    \draw[Sienna,dashed] (2,0.5) -- (2,-1.5);
    \draw[Sienna,dashed]  (1.75, -1.5) -- (2,0.5);
    \draw[Sienna,dashed]  (2.25, -1.5) -- (2,0.5);
    \draw[Sienna,dashed]  (1.5, -1.5) -- (2,0.5);
    \draw[Sienna,dashed]  (2.5, -1.5) -- (2,0.5);
    
    \draw[Sienna,dashed]  (6, -1.5) -- (6,0.5);
    \draw[Sienna,dashed]  (5.75, -1.5) -- (6,0.5);
    \draw[Sienna,dashed]  (6.25, -1.5) -- (6,0.5);
    \draw[Gray,snake=zigzag,thick]  (5.25, -1.5) -- (6,0.5);
    \draw[Sienna,dashed]  (6.5, -1.5) -- (6,0.5);
    

    \draw[Sienna,dashed]  (10, -1.5) -- (global);
    \draw[Sienna,dashed]  (9.75, -1.5) -- (global);
    \draw[Sienna,dashed]  (10.25, -1.5) -- (global);
    \draw[Sienna,dashed]  (9.5, -1.5) -- (global);
    \draw[Sienna,dashed]  (10.5, -1.5) -- (global);
    

    \draw[thick,LightSeaGreen] (2,0.5) -- (3,3); 
    \draw[thick,LightSeaGreen] (6,0.5) -- (7,3);
    \draw[thick,LightSeaGreen] (4,0.5) -- (4,3);
    \draw[thick,LightSeaGreen] (3,0.5) -- (2,3);
    \draw[thick,LightSeaGreen] (5,0.5) -- (4,3);
    \draw[thick,LightSeaGreen] (5,0.5) -- (6,3);
    \draw[thick,LightSeaGreen] (2,0.5) -- (6,3);
    \draw[thick,LightSeaGreen] (6,0.5) -- (2,3);
    \draw[thick,LightSeaGreen] (3,0.5) -- (6,3);
    \draw[thick,LightSeaGreen] (3,0.5) -- (1,3);
    \draw[thick,LightSeaGreen] (4,0.5) -- (5,3);
    \draw[thick,LightSeaGreen] (3,0.5) -- (3,3);
    \draw[thick,LightSeaGreen] (4,0.5) -- (1,3);
    \draw[thick,LightSeaGreen] (5,0.5) -- (2,3);

    \foreach \x in {2,...,6}
        \node at (\x, 0.5) [square,draw,fill=YellowGreen] () {};
    
    \foreach \x in {4,...,7}
        \draw[fill=Crimson!77] (\x, 3) circle (0.15cm);

    \foreach \x in {1,2,3}
        \node at (\x, 3) [square,draw,fill=YellowGreen] () {};
        
    \draw[fill=OrangeRed!42]  (2, -1.5) circle (0.07cm);
    \draw[fill=OrangeRed!42]  (1.75, -1.5) circle (0.07cm);
    \draw[fill=OrangeRed!42]  (2.25, -1.5) circle (0.07cm);
    
    \draw[fill=OrangeRed!42]  (1.5, -1.5) circle (0.07cm);

    \draw[fill=OrangeRed!42]  (2.5, -1.5) circle (0.07cm);
    
    \draw[fill=OrangeRed!42]  (6, -1.5) circle (0.07cm);
    \draw[fill=OrangeRed!42]  (5.75, -1.5) circle (0.07cm);
    \draw[fill=OrangeRed!42]  (6.25, -1.5) circle (0.07cm);
    \draw[fill=OrangeRed!42]  (6.5, -1.5) circle (0.07cm);
    \node at (5.25, -1.5) [square,draw,fill=YellowGreen] () {};
    
    \draw[fill=OrangeRed!42]  (10, -1.5) circle (0.07cm);
    \draw[fill=OrangeRed!42]  (9.75, -1.5) circle (0.07cm);
    \draw[fill=OrangeRed!42]  (10.25, -1.5) circle (0.07cm);
    \draw[fill=OrangeRed!42]  (9.5, -1.5) circle (0.07cm);
    \draw[fill=OrangeRed!42]  (10.5, -1.5) circle (0.07cm);

    \node at (0,7) {};
    
    
    \draw[thick] (3.5,3.5) -- (3.5,2.5);
    \draw [thick,decorate,decoration={brace,amplitude=4pt}]  (0.5,4) -- (3.5,4);
    \node[scale=0.7] at (2,4.5) {Dominating Set};

    }
    \end{tikzpicture}
    \caption{This figure demonstrates a the positions of the guards after the defence in ~\Cref{fig:case2e} is executed.}
    \label{fig:case2f}
    \end{subfigure}
    \caption{This figure demonstrates the case when $\mathcal{S}$ is a dependent vertex cover and a sliding edge is attacked.}
    \label{fig:case10bip}
\end{figure}
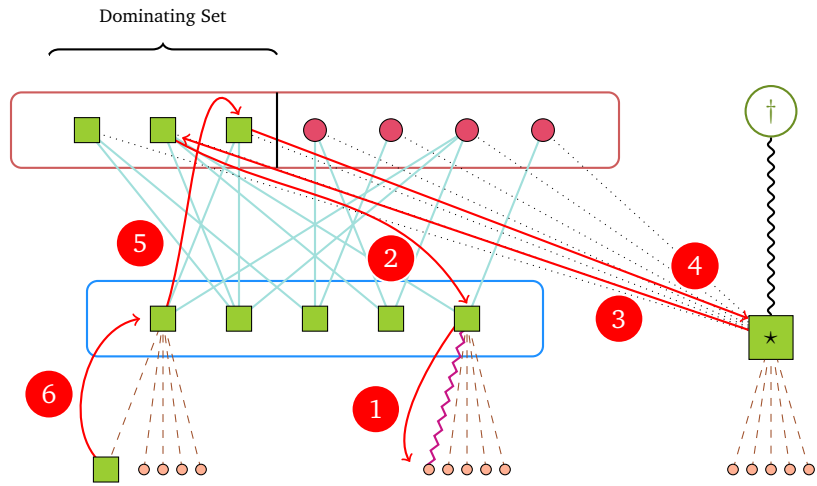
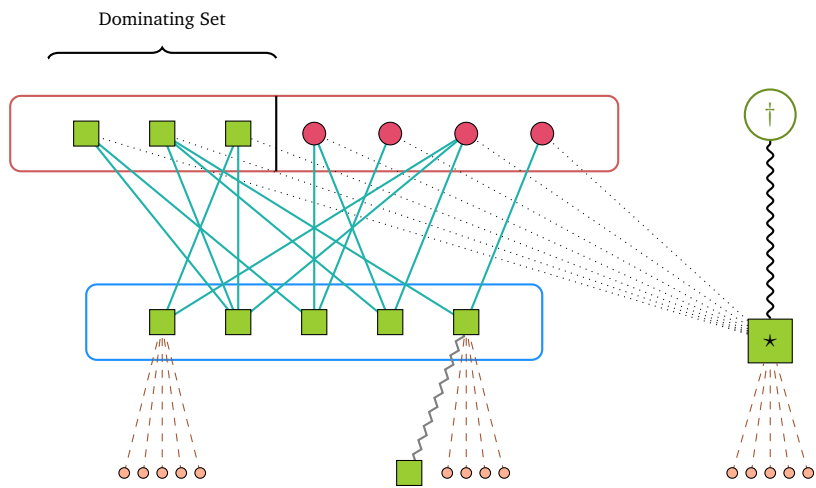

 \paragraph*{Case 11. $\mathcal{S}$ is a dependent vertex cover and $e$ is a supplier edge.}
  
  Let $w\in C_j$ be such that $w$ has a guard. 
  Let $e = (v_i,\star)$. If $i \leq k$, the guards on $v_i$ and $\star$ exchange positions. Suppose $i>k$, there exists $q\in[k]$ such that $(v_q,u_j)\in E(G)$. We have the following sequence of moves:
  \begin{itemize}
      \item The guard on the universal vertex moves to $v_i$ via the supplier edge that was attacked.
      \item The guard on $v_q$ moves to $\star$ along a supplier edge;
      \item The guard on $u_j$ moves to $v_q$ along a structural edge; and 
      \item The guard on $w$ moves to $u_j$ along a sliding edge.
  \end{itemize}
  
  Note that the resulting configuration is a red vertex cover.
   
  Let $w\in D$ be such that $w$ has a guard. Let $e = (v_i,\star)$. If $i \leq k$, the guards on $v_i$ and $\star$ exchange positions. Let $i>k$. We now perform the following sequence of moves:
  \begin{itemize}
      \item The guard on $\star$ moves to $v_i$ along the attacked supplier edge; and
      \item The guard on $w$ moves to $\star$ along a sliding edge.
  \end{itemize} 
  
  Note that the resulting configuration is a red vertex cover.
  
\paragraph*{Case 12. $\mathcal{S}$ is a dependent vertex cover and $e$ is a bridge edge.}
 Let $w\in C_j$ be such that $w$ has a guard. There exists $q\in[k]$ such that $(v_q,u_j)\in E(G)$. We now perform the following sequence of moves:
 \begin{itemize}
     \item The guard on $\star$ moves to $\dagger$ along the attacked bridge edge;
     \item The guard on $v_q$ moves to $\star$ along a supplier edge;
     \item The guard on $u_j$ moves to $v_q$ along a structural edge;
     and
     \item The guard on $w$ moves to $u_j$ along a sliding edge.
 \end{itemize}
 
Note that the resulting configuration is a backup vertex cover. 

 Let $w\in D$ be such that $w$ has a guard. The guard on $\star$ moves to $\dagger$ along an attacked bridge edge and the guard on $w$ moves to $\star$ along a sliding edge. The resulting configuration is a backup vertex cover.
\end{proof}
Having argued that any attack on a nice vertex cover can be defended in such a way that the resulting configuration is also a nice vertex cover, we have completed the argument in the forward direction as well. 
\end{proof}

Observe that the instance that we construct in the proof of~\Cref{lem:bipartite} is both bipartite and has diameter at most six. 

It is also easily checked that all the vertex covers used by the defense in the forward direction induced connected subgraphs, since every vertex cover contains all the blue vertices, a dominating set for the blue vertices, and a universal vertex that is adjacent to all the vertices in the dominating set; and any other vertex is adjacent to one of the blue vertices (or the universal vertex). Therefore, the reduction above also serves to demonstrate the hardness of \ECVC{} on bipartite graphs --- note that the argument for the reverse direction is exactly the same since every connected vertex cover is also a vertex cover.

Overall,~\Cref{lem:bipartite} along with~\Cref{prop:rbds-npc} and the remarks above lead to our main result.

\begin{thm1}
Both the \EVC{} and \ECVC{} problems are \NPH{} and do not admit a polynomial compression parameterized by the number of guards (unless $\mathrm{NP} \subseteq \operatorname{coNP}/\text{poly}$), even on bipartite graphs of diameter six. 
\end{thm1}






\section{A Polynomial-time Algorithm for Co-bipartite Graphs}
\label{sec:cobip}

In this section, we focus on a proof of Theorem 2.

\begin{thm2}
There is a polynomial-time algorithm for \EVC{} on the class of cobipartite graphs. 
\end{thm2}

To the best of our knowledge, this result is not subsumed by any of the known polynomial-time algorithms for special classes of graphs. In particular, it is easily checked that the class of cobipartite graphs is not contained in any of the following classes: chordal graphs, cactus graphs, and generalized trees\footnote{The notion of generalized trees in the context of eternal vertex cover was considered by~\cite{AFI2015}. Such graphs are characterized by the following property: every block is an elementary bipartite graph or a clique having at most two cut-vertices in it. Note that a cobipartite graph with four vertices in both parts with two disjoint edges across the parts is not a generalized tree.}.

Let $G = (V = A \uplus B, E)$ be a cobipartite graph with bipartition $A,B$. Recall that $G[A]$ and $G[B]$ are cliques. Consider that $A$ has $p$ vertices $\{a_1,a_2,\ldots,a_p\}$ and $B$ has $q$ vertices $\{b_1,b_2,\ldots,b_q\}$. Without loss of generality we assume that $p\leq q$. Without loss of generality, we assume that there are no global vertices in $A$. (If there is some global vertex in $A$, simply shift that vertex to $B$). We also assume throughout that $p \geq 1$ --- if $p = 0$ then $G$ is a clique and $evc(G) = mvc(G) = |V(G)|-1$. 

Since the cliques require $p-1$ and $q-1$ vertices respectively for a vertex cover, we have $mvc(G)\geq p+q-2$. Since $p = |A|\geq 1$, there exists a (non-global) vertex $a_i$ on the $A$ side and therefore it has at least one non-neighbor (say $b_j$) and thus we have a vertex cover of size $p+q-2$ given by $V(G) \setminus \{a_i,b_j\}$. Therefore, $mvc(G)= p+q-2$. We make a note of this fact in the following claim.

\begin{claim}\label{vc}
For any co-bipartite graph $G$ with bipartitions $A$ and $B$ with all the notations as described above, if there are no global vertices in $A$ and $|A|\geq 1$, $mvc(G)=p+q-2$.
\end{claim}

Consider a configuration where all the vertices have a guard except one. This will always be a vertex cover because only one vertex is uncovered, no two endpoints of any edge can be uncovered. Now if any edge not adjacent to this particular vertex is attacked, the guards can exchange places. If any edge adjacent to this vertex is attacked, the guard on the other endpoint of the edge will cross the edge and come to this vertex. So, we are again left with a situation where all the vertices have a guard except one. Therefore, if the number of guards is one less than the number of vertices i.e. $p+q-1$, the defender always has a winning strategy. Therefore, $evc(G)\leq p+q-1$.

Thus, for all co-bipartite graphs we have:

$$mvc(G) = p+q-2 \mbox{ and } p+q-2\leq evc(G)\leq p+q-1.$$ 

We will now derive the characterization of co-bipartite graphs which have $evc(G)=mvc(G)=p+q-2$.


Consider any non-global vertex on the $A$ side (or $B$ side), then its non-neighbour must exist on the $B$ side (respectively $A$ side) because $G[A]$ (respectively $G[B]$) is a clique. For any vertex $a_i$, if $a_ib_j\notin E(G)$, then $b_j$ is said to be a \textit{friend} of $a_i$. We will also use the phrases: ``$a_i$ is a friend of $b_j$'', or ``$a_i$ and $b_j$ are friends'' to mean that $a_i$ and $b_j$ are non-neighbours (and thus belong to different sides of the bipartition). Note that a vertex is non-global if and only if it has at least one friend.

Denote $A\backslash\{a_i\}\cup B\backslash \{b_j\}$ by $S_{ij}$. Note that $S_{ij}$ will be a vertex cover if and only if $a_i$ and $b_j$ are friends. We now argue a series of claims that account for all possible scenarios for co-bipartite graphs. A summary of the cases can be found in~\Cref{fig:overview,overviewcontd}. 

\begin{figure}
    \centering
    \includegraphics[width=\textwidth]{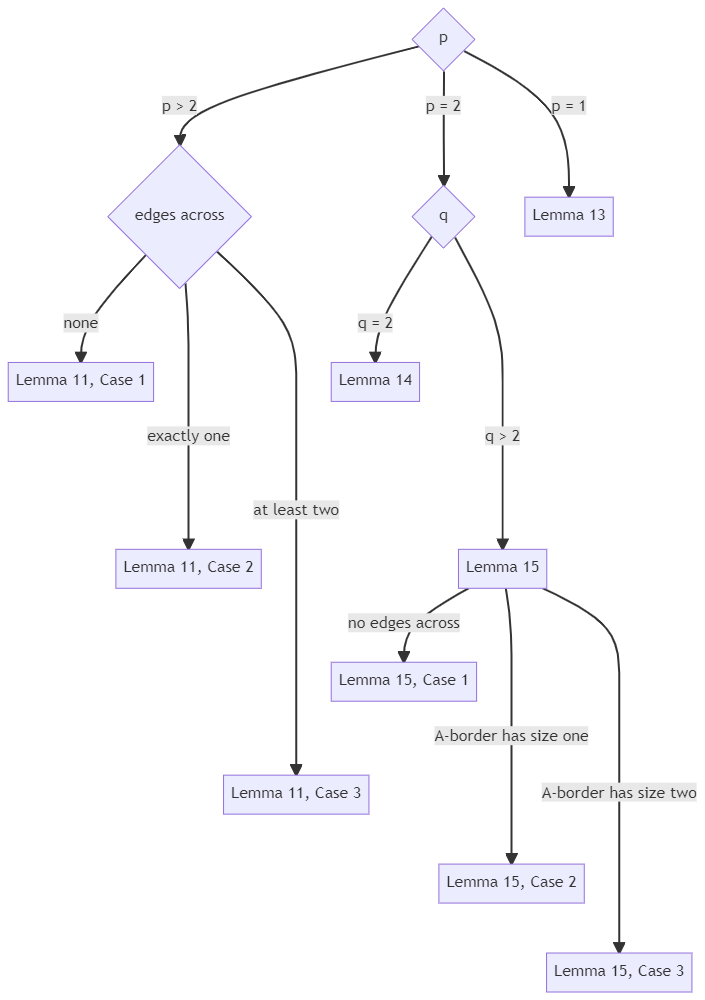}
    \caption{An overview of the proof of Theorem 2.}
    \label{fig:overview}
\end{figure}

\begin{figure}
\centering
\begin{subfigure}[b]{\textwidth}
\centering
   \includegraphics[width=\textwidth]{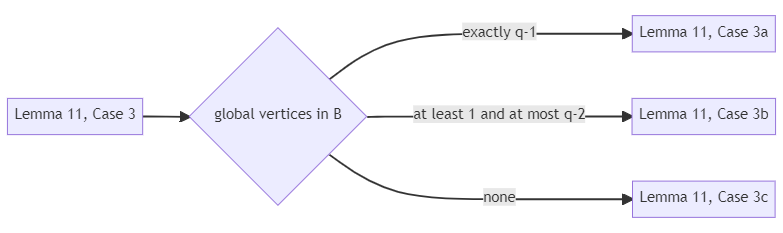}
   \caption{}
   \label{fig:overview1} 
\end{subfigure}
\hfill
\begin{subfigure}[b]{\textwidth}
\centering
   \includegraphics[width=\textwidth]{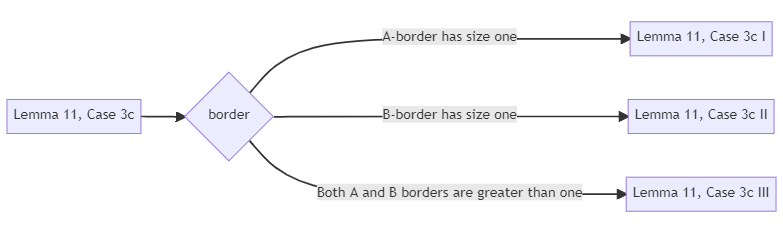}
   \caption{}
   \label{fig:overview2}
\end{subfigure}
\hfill
\begin{subfigure}[b]{\textwidth}
\centering
   \includegraphics[width=\textwidth]{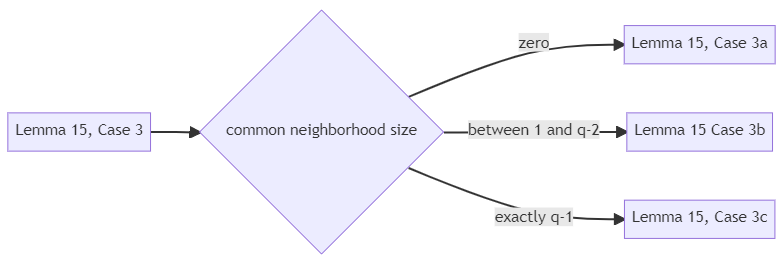}
   \caption{}
   \label{fig:overview3}
\end{subfigure}
\caption{An overview of the cases in the proof of Theorem 2 (continued from~\Cref{fig:overview}).}
\label{overviewcontd}
\end{figure}


\begin{lemma}\label{big}
For any co-bipartite graph $G$ with all the notations as described above (with no global vertices on the $A$ side) and $|A|\geq 3$, we have $evc(G)\neq mvc(G)$ if and only if there is exactly one non-global vertex on the $B$ side.
\end{lemma}
\begin{proof}
Using \Cref{vc}, we have $mvc(G)=p+q-2$. We will show that except when $B$ side has exactly one non-global vertex, $p+q-2$ guards are sufficient to defend any sequence of attacks and when $B$ side has exactly one non-global vertex, we show that the attacker has a winning strategy no matter where we place the $p+q-2$ guards.

Note that since $G[A]$ and $G[B]$ are cliques, any vertex cover of size $p+q-2$ will have $p-1$ vertices from $A$ and $q-1$ vertices from $B$. In each of the following cases, we start with an arbitrary vertex cover $S_{ij}$ where $a_i$ and $b_j$ are friends. For each type of attack, we show that it is possible to reconfigure the guards in such a way that at least one guard moves along the attacked edge and the resulting configuration is again a vertex cover.

\paragraph*{Case $1$: There is no edge between $A$ and $B$.}
Let $S_{ij}$ be the initial vertex cover. If any edge which is not incident to either $a_i$ or $b_j$ is attacked, the guards on its endpoints exchange positions along the attacked edge to take care of the attack and we again have the same vertex cover. 

The following types of edges can be attacked which are incident to $a_i$ or $b_j$. 
\begin{enumerate}
    \item \textbf{Some edge $a_ia_r$ on the $A$ side is attacked:}
    
        The guard on $a_r$ moves to $a_i$ along the attacked edge. The resulting vertex cover is $S_{rj}$. We know that $a_r$ is a friend of $b_j$ because there is no edge between $A$ and $B$.  
    \item \textbf{Some edge $b_sb_j$ on the $B$ side is attacked:}
    
        The guard on $b_s$ moves to $b_j$ along the attacked edge. The resulting vertex cover is $S_{is}$. We know that $a_i$ is a friend of $b_s$ because there is no edge between $A$ and $B$.
\end{enumerate}
Thus $p+q-2$ guards are sufficient to defend any attack.

\paragraph*{Case $2$: There is just one edge between $A$ and $B$.}
Let the edge between $A$ and $B$ be $a_1b_1$. We ensure that there will always be a guard on both $a_1$ and $b_1$. Let $S_{ij}$ be the initial vertex cover where $i,j\neq 1$. If any edge which is not incident to either $a_i$ or $b_j$ is attacked, the guards on its endpoints exchange positions along the attacked edge to take care of the attack and we again have the same vertex cover.

The following types of edges can be attacked which are incident to $a_i$ or $b_j$. 

\begin{enumerate}
    \item \textbf{Some edge $a_ia_r$ with $r\neq 1$ on the $A$ side is attacked:}
    
    The guard on $a_r$ moves to $a_i$ along the attacked edge. The resulting vertex cover is $S_{rj}$. We know that $a_r$ is a friend of $b_j$ because $r,j\neq 1$. 
    \item \textbf{Some edge $a_1a_i$ is attacked:}
    
    The guard on $a_1$ moves to $a_i$ and a guard on $a_r$ moves to $a_1$ for some $r\neq 1,i$. The resulting vertex cover is $S_{rj}$ where $a_r$ is a friend of $b_j$ because $r,j\neq 1$~(c.f.~\Cref{fig:case2lemma11}).
    \item \textbf{Some edge $b_sb_j$ on  the $B$ side is attacked:}
    
    The guard on $b_s$ moves to $b_j$ along the attacked edge. The resulting vertex cover is $S_{is}$. We know that $a_i$ is a friend of $b_s$ because $s,i\neq 1$. 
    \item \textbf{Some edge $b_1b_j$ is attacked:}
    
    The guard on $b_1$ moves to $b_j$ along the attacked edge and a guard on $b_s$ moves to $b_1$ for some $s\neq 1,j$. The resulting vertex cover is $S_{is}$ where $a_i$ is a friend of $b_s$ because $i,s\neq 1$.
\end{enumerate}
Thus $p+q-2$ guards are sufficient to defend any attack.

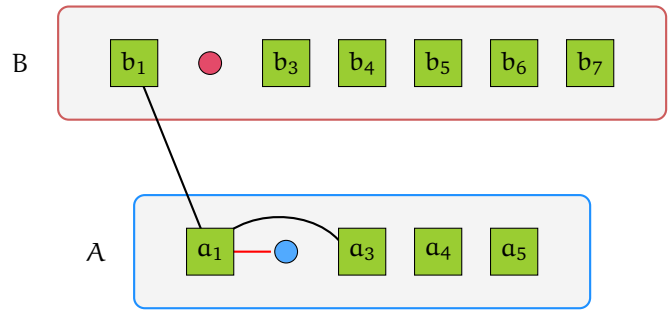
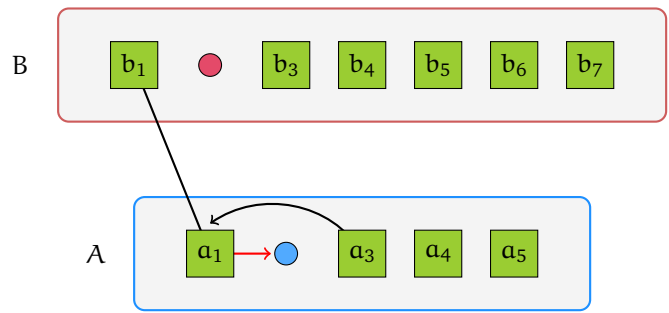
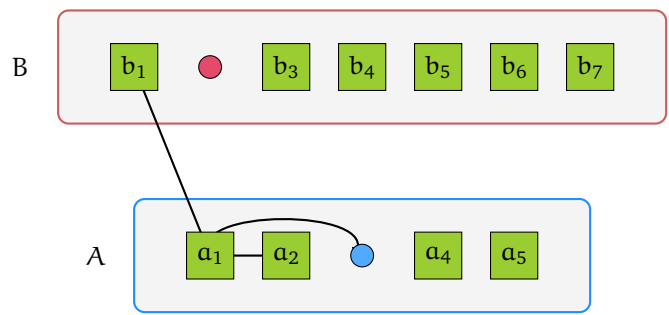
\begin{figure}
    \centering

    \begin{subfigure}[b]{\textwidth}
    \centering
    \begin{tikzpicture}[square/.style={regular polygon,regular polygon sides=4}]
    \tikz{ 
    
    \draw [DodgerBlue,thick,rounded corners,fill=LightGray!25] (1,-0.25) rectangle (7,1.25);
    \draw [IndianRed,thick,rounded corners,fill=LightGray!25] (0,2.25) rectangle (8,3.75);

    
    \draw[thick] (2,0.5) -- (1,3); 
    
    \draw[thick,Red] (2,0.5) -- (2.8,0.5);
    \draw[->,thick] (2,0.5) to [out=60,in=120](3.8,0.5);

    \foreach \x in {2,4,5,6}
        \node at (\x, 0.5) [square,draw,fill=YellowGreen,minimum size=25pt] () {};
        
    \node at (2,0.5) {$a_1$};
    \node at (4,0.5) {$a_3$};
    \node at (5,0.5) {$a_4$};
    \node at (6,0.5) {$a_5$};

    \foreach \x in {1,3,4,5,6,7}
        \node at (\x, 3) [square,draw,fill=YellowGreen,minimum size=25pt] () {};

    \node at (1,3) {$b_1$};
    \node at (3,3) {$b_3$};
    \node at (4,3) {$b_4$};
    \node at (5,3) {$b_5$};
    \node at (6,3) {$b_6$};
    \node at (7,3) {$b_7$};
    
    \draw[fill=DodgerBlue!77] (3, 0.5) circle (0.15cm);
    \draw[fill=Crimson!77] (2, 3) circle (0.15cm);
    
    \node at (0.5,0.5) {$A$};
    \node at (-0.5,3) {$B$};
    
    \node at (0,4) {};

    }
    \end{tikzpicture}
    \caption{This figure demonstrates the case when there is only one edge between $A$ and $B$ and an edge on the $A$ side (say $a_1a_2$) is attacked.}
    \label{fig:lem11case2-1}
\end{subfigure}

\begin{subfigure}[b]{\textwidth}
    \centering
    \begin{tikzpicture}[square/.style={regular polygon,regular polygon sides=4}]
    \tikz{ 
    
    \draw [DodgerBlue,thick,rounded corners,fill=LightGray!25] (1,-0.25) rectangle (7,1.25);
    \draw [IndianRed,thick,rounded corners,fill=LightGray!25] (0,2.25) rectangle (8,3.75);

    
    \draw[thick] (2,0.5) -- (1,3); 
    
    \draw[->,thick,Red] (2,0.5) -- (2.8,0.5);
    \draw[->,thick] (4,0.5) to [out=120,in=35](2,0.9);

    \foreach \x in {2,4,5,6}
        \node at (\x, 0.5) [square,draw,fill=YellowGreen,minimum size=25pt] () {};
        
    \node at (2,0.5) {$a_1$};
    \node at (4,0.5) {$a_3$};
    \node at (5,0.5) {$a_4$};
    \node at (6,0.5) {$a_5$};

    \foreach \x in {1,3,4,5,6,7}
        \node at (\x, 3) [square,draw,fill=YellowGreen,minimum size=25pt] () {};

    \node at (1,3) {$b_1$};
    \node at (3,3) {$b_3$};
    \node at (4,3) {$b_4$};
    \node at (5,3) {$b_5$};
    \node at (6,3) {$b_6$};
    \node at (7,3) {$b_7$};
    
    \draw[fill=DodgerBlue!77] (3, 0.5) circle (0.15cm);
    \draw[fill=Crimson!77] (2, 3) circle (0.15cm);
    
    \node at (0.5,0.5) {$A$};
    \node at (-0.5,3) {$B$};
    
    \node at (0,4) {};

    }
    \end{tikzpicture}
    \caption{This figure demonstrates the defense for the attack in  \Cref{fig:lem11case2-1}. The guard on $a_3$ moves to $a_1$ and the guard on $a_1$ moves to $a_2$.}
    \label{fig:lem11case2-2}
\end{subfigure}

\begin{subfigure}[b]{\textwidth}
    \centering
    \begin{tikzpicture}[square/.style={regular polygon,regular polygon sides=4}]
    \tikz{ 
    
    \draw [DodgerBlue,thick,rounded corners,fill=LightGray!25] (1,-0.25) rectangle (7,1.25);
    \draw [IndianRed,thick,rounded corners,fill=LightGray!25] (0,2.25) rectangle (8,3.75);

    
    \draw[thick] (2,0.5) -- (1,3); 
    
    \draw[thick] (2,0.5) -- (2.8,0.5);
    \draw[->,thick] (2,0.5) to [out=120,in=60](3.9,0.5);

    \foreach \x in {2,3,5,6}
        \node at (\x, 0.5) [square,draw,fill=YellowGreen,minimum size=25pt] () {};
        
    \node at (2,0.5) {$a_1$};
    \node at (3,0.5) {$a_2$};
    \node at (5,0.5) {$a_4$};
    \node at (6,0.5) {$a_5$};

    \foreach \x in {1,3,4,5,6,7}
        \node at (\x, 3) [square,draw,fill=YellowGreen,minimum size=25pt] () {};

    \node at (1,3) {$b_1$};
    \node at (3,3) {$b_3$};
    \node at (4,3) {$b_4$};
    \node at (5,3) {$b_5$};
    \node at (6,3) {$b_6$};
    \node at (7,3) {$b_7$};
    
    \draw[fill=DodgerBlue!77] (4, 0.5) circle (0.15cm);
    \draw[fill=Crimson!77] (2, 3) circle (0.15cm);
    
    \node at (0.5,0.5) {$A$};
    \node at (-0.5,3) {$B$};
    
    \node at (0,4) {};

    }
    \end{tikzpicture}
    \caption{This figure demonstrates the position of guards after executing the defense in \Cref{fig:lem11case2-2}.}
    \label{fig:lem11case2-3}
\end{subfigure}
\caption{This figure demonstrates Case $2$ from~\Cref{big}.}
\label{fig:case2lemma11}
\end{figure}

\paragraph*{Case $3$: There are at least two edges between $A$ and $B$.}
\paragraph*{Case $a$: There is exactly one vertex on the $B$ side which is not global.}
Let $b_1$ be the vertex on the $B$ side which is not global. All the other vertices on the $B$ side are global. Since by assumption there is no global vertex on $A$ side, $b_1$ is a friend of all the vertices on $A$ side and any vertex on $A$ side does not have any other friend. That is why, any arbitrary vertex cover of size $p+q-2$ is of the form $S_{i1}$ for some $i\in [1,p]$. Let all the vertices of $S_{i1}$ be occupied by guards. We show that the attacker has a winning strategy. 

Attack $b_2b_1$. The guard on $b_2$ is forced to move to $b_1$ along the attacked edge and cannot move any further. So there must be a guard on $b_1$. If some guard from the $A$ side moves to the $B$ side, there will be $p-2$ guards on the $A$ side and thus some edge on the $A$ side will be vulnerable. If some guard on the $B$ side moves to the $A$ side, there will be $q-2$ guards on the $B$ side and thus some edge on the $B$ side will be vulnerable. If there are $p-1$ guards on the $A$ side and $q-1$ guards on the $B$ side after the attack, then some vertex $a_r$ ($r$ may be equal to $i$) must be empty and some vertex $b_s$ where $s\neq1$ must be empty. So now the edge $a_rb_s$ is vulnerable. (This edge exists because $b_s$ is a global vertex). Thus the defender is not able to defend this attack. Thus $p+q-2$ guards are not sufficient and $evc(G)\neq mvc(G)$ (c.f.~\Cref{fig:case3alem11}).

\begin{figure}
    \centering
    \begin{subfigure}[b]{\textwidth}
    \centering
    \begin{tikzpicture}[square/.style={regular polygon,regular polygon sides=4}]
    \tikz{ 
    
    \draw [DodgerBlue,thick,rounded corners,fill=LightGray!25] (1,-0.25) rectangle (7,1.25);
    \draw [IndianRed,thick,rounded corners,fill=LightGray!25] (0,2.25) rectangle (8,3.75);

    
    \draw[thick] (2,0.5) -- (2,3);
    \draw[thick] (2,0.5) -- (3,3);
    \draw[thick] (2,0.5) -- (4,3);
    \draw[thick] (2,0.5) -- (5,3);
    \draw[thick] (2,0.5) -- (6,3);
    \draw[thick] (2,0.5) -- (7,3);
    \draw[thick] (3,0.5) -- (2,3);
    \draw[thick] (3,0.5) -- (3,3);
    \draw[thick] (3,0.5) -- (4,3);
    \draw[thick] (3,0.5) -- (5,3);
    \draw[thick] (3,0.5) -- (6,3);
    \draw[thick] (3,0.5) -- (7,3);
    \draw[thick] (4,0.5) -- (2,3);
    \draw[thick] (4,0.5) -- (3,3);
    \draw[thick] (4,0.5) -- (4,3);
    \draw[thick] (4,0.5) -- (5,3);
    \draw[thick] (4,0.5) -- (6,3);
    \draw[thick] (4,0.5) -- (7,3);
    \draw[thick] (5,0.5) -- (2,3);
    \draw[thick] (5,0.5) -- (3,3);
    \draw[thick] (5,0.5) -- (4,3);
    \draw[thick] (5,0.5) -- (5,3);
    \draw[thick] (5,0.5) -- (6,3);
    \draw[thick] (5,0.5) -- (7,3);
    \draw[thick] (6,0.5) -- (2,3);
    \draw[thick] (6,0.5) -- (3,3);
    \draw[thick] (6,0.5) -- (4,3);
    \draw[thick] (6,0.5) -- (5,3);
    \draw[thick] (6,0.5) -- (6,3);
    \draw[thick] (6,0.5) -- (7,3);
    
    \draw[thick,Red] (2,3) -- (1,3);
    \draw[->,thick] (2,0.5) to [out=-50,in=-130](5,0.5);

    \foreach \x in {2,3,4,6}
        \node at (\x, 0.5) [square,draw,fill=YellowGreen,minimum size=20pt] () {};
        
    \node at (2,0.5) {$a_1$};
    \node at (3,0.5) {$a_2$};
    \node at (4,0.5) {$a_3$};
    \node at (6,0.5) {$a_5$};

    \foreach \x in {2,3,4,5,6,7}
        \node at (\x, 3) [square,draw,fill=YellowGreen,minimum size=20pt] () {};

    \node at (2,3) {$b_2$};
    \node at (3,3) {$b_3$};
    \node at (4,3) {$b_4$};
    \node at (5,3) {$b_5$};
    \node at (6,3) {$b_6$};
    \node at (7,3) {$b_7$};
    
    \draw[fill=DodgerBlue!77] (5, 0.5) circle (0.15cm);
    \draw[fill=Crimson!77] (1, 3) circle (0.15cm);
    
    \node at (0.5,0.5) {$A$};
    \node at (-0.5,3) {$B$};
    
    \node at (0,4.5) {};

    }
    \end{tikzpicture}
    \caption{This figure demonstrates the case when there are at least two edges between $A$ and $B$, exactly $1$ non-global vertex and an edge on the $B$ side (say $b_1b_2$) is attacked.}
    \label{fig:lem11case3a-1}
\end{subfigure}

\begin{subfigure}[b]{\textwidth}
    \centering
    \begin{tikzpicture}[square/.style={regular polygon,regular polygon sides=4}]
    \tikz{ 
    
    \draw [DodgerBlue,thick,rounded corners,fill=LightGray!25] (1,-0.25) rectangle (7,1.25);
    \draw [IndianRed,thick,rounded corners,fill=LightGray!25] (0,2.25) rectangle (8,3.75);

    
    \draw[->,thick] (2,0.5) -- (2,2.7);
    \draw[thick] (2,0.5) -- (3,3);
    \draw[thick] (2,0.5) -- (4,3);
    \draw[thick] (2,0.5) -- (5,3);
    \draw[thick] (2,0.5) -- (6,3);
    \draw[thick] (2,0.5) -- (7,3);
    \draw[thick] (3,0.5) -- (2,3);
    \draw[thick] (3,0.5) -- (3,3);
    \draw[thick] (3,0.5) -- (4,3);
    \draw[thick] (3,0.5) -- (5,3);
    \draw[thick] (3,0.5) -- (6,3);
    \draw[thick] (3,0.5) -- (7,3);
    \draw[thick] (4,0.5) -- (2,3);
    \draw[thick] (4,0.5) -- (3,3);
    \draw[thick] (4,0.5) -- (4,3);
    \draw[thick] (4,0.5) -- (5,3);
    \draw[thick] (4,0.5) -- (6,3);
    \draw[thick] (4,0.5) -- (7,3);
    \draw[thick] (5,0.5) -- (2,3);
    \draw[thick] (5,0.5) -- (3,3);
    \draw[thick] (5,0.5) -- (4,3);
    \draw[thick] (5,0.5) -- (5,3);
    \draw[thick] (5,0.5) -- (6,3);
    \draw[thick] (5,0.5) -- (7,3);
    \draw[thick] (6,0.5) -- (2,3);
    \draw[thick] (6,0.5) -- (3,3);
    \draw[thick] (6,0.5) -- (4,3);
    \draw[thick] (6,0.5) -- (5,3);
    \draw[thick] (6,0.5) -- (6,3);
    \draw[thick] (6,0.5) -- (7,3);
    
    \draw[->,thick,Red] (1.8,3) -- (1.2,3);
    \draw[->,thick] (5,0.5) to [out=-40,in=-130](2,0.5);

    \foreach \x in {2,3,4,6}
        \node at (\x, 0.5) [square,draw,fill=YellowGreen,minimum size=20pt] () {};
        
    \node at (2,0.5) {$a_1$};
    \node at (3,0.5) {$a_2$};
    \node at (4,0.5) {$a_3$};
    \node at (6,0.5) {$a_5$};

    \foreach \x in {2,3,4,5,6,7}
        \node at (\x, 3) [square,draw,fill=YellowGreen,minimum size=20pt] () {};

    \node at (2,3) {$b_2$};
    \node at (3,3) {$b_3$};
    \node at (4,3) {$b_4$};
    \node at (5,3) {$b_5$};
    \node at (6,3) {$b_6$};
    \node at (7,3) {$b_7$};
    
    \draw[fill=DodgerBlue!77] (5, 0.5) circle (0.15cm);
    \draw[fill=Crimson!77] (1, 3) circle (0.15cm);
    
    \node at (0.5,0.5) {$A$};
    \node at (-0.5,3) {$B$};
    
    \node at (0,4.5) {};

    }
    \end{tikzpicture}
    \caption{This figure demonstrates the defense for the attack in \Cref{fig:lem11case3a-1}. The guard on $b_2$ moves to $b_1$, and the guard on $a_1$ moves to $b_2$.}
    \label{fig:lem11case3a-2}
\end{subfigure}

\begin{subfigure}[b]{\textwidth}
    \centering
    \begin{tikzpicture}[square/.style={regular polygon,regular polygon sides=4}]
    \tikz{ 
    
    \draw [DodgerBlue,thick,rounded corners,fill=LightGray!25] (1,-0.25) rectangle (7,1.25);
    \draw [IndianRed,thick,rounded corners,fill=LightGray!25] (0,2.25) rectangle (8,3.75);

    
   \draw[thick] (2,0.5) -- (2,3);
    \draw[thick] (2,0.5) -- (3,3);
    \draw[thick] (2,0.5) -- (4,3);
    \draw[thick] (2,0.5) -- (5,3);
    \draw[thick] (2,0.5) -- (6,3);
    \draw[thick] (2,0.5) -- (7,3);
    \draw[thick] (3,0.5) -- (2,3);
    \draw[thick] (3,0.5) -- (3,3);
    \draw[thick] (3,0.5) -- (4,3);
    \draw[thick] (3,0.5) -- (5,3);
    \draw[thick] (3,0.5) -- (6,3);
    \draw[thick] (3,0.5) -- (7,3);
    \draw[thick] (4,0.5) -- (2,3);
    \draw[thick] (4,0.5) -- (3,3);
    \draw[thick] (4,0.5) -- (4,3);
    \draw[thick] (4,0.5) -- (5,3);
    \draw[thick] (4,0.5) -- (6,3);
    \draw[thick] (4,0.5) -- (7,3);
    \draw[thick] (5,0.5) -- (2,3);
    \draw[thick] (5,0.5) -- (3,3);
    \draw[thick] (5,0.5) -- (4,3);
    \draw[thick] (5,0.5) -- (5,3);
    \draw[thick] (5,0.5) -- (6,3);
    \draw[thick] (5,0.5) -- (7,3);
    \draw[thick] (6,0.5) -- (2,3);
    \draw[thick] (6,0.5) -- (3,3);
    \draw[thick] (6,0.5) -- (4,3);
    \draw[thick] (6,0.5) -- (5,3);
    \draw[thick] (6,0.5) -- (6,3);
    \draw[thick] (6,0.5) -- (7,3);
    
    \draw[thick] (1.8,3) -- (1.2,3);
   \draw[->,thick] (2,0.5) to [out=-40,in=-130](5,0.5);

    \foreach \x in {3,4,6}
        \node at (\x, 0.5) [square,draw,fill=YellowGreen,minimum size=25pt] () {};

    \node at (3,0.5) {$a_2$};
    \node at (4,0.5) {$a_3$};
    \node at (6,0.5) {$a_5$};

    \foreach \x in {1,2,3,4,5,6,7}
        \node at (\x, 3) [square,draw,fill=YellowGreen,minimum size=25pt] () {};

    \node at (1,3) {$b_1$};
    \node at (2,3) {$b_2$};
    \node at (3,3) {$b_3$};
    \node at (4,3) {$b_4$};
    \node at (5,3) {$b_5$};
    \node at (6,3) {$b_6$};
    \node at (7,3) {$b_7$};
    
    \draw[fill=DodgerBlue!77] (5, 0.5) circle (0.15cm);
    \draw[fill=DodgerBlue!77] (2, 0.5) circle (0.15cm);
    
    \node at (0.5,0.5) {$A$};
    \node at (-0.5,3) {$B$};
    
    \node at (0,4.5) {};
    }
    \end{tikzpicture}
    \caption{This figure demonstrates the position of guards after executing the defense in \Cref{fig:lem11case3a-2}. The edge $a_1a_4$ is now vulnerable.}
    \label{fig:lem11case3a-3}
\end{subfigure}
\caption{This figure demonstrates Case $3a$ from \Cref{big}.}
\label{fig:case3alem11}
\end{figure}

\paragraph*{Case $b$: There is at least one and at most $q-2$ global vertices on the $B$ side.}
Let $B=B_1\cup B_2$ where $B_1=\{b_1,b_2\ldots,b_k\}$ is the set of non-global vertices and $B_2=\{b_{k+1},b_{k+2},\ldots,b_{q}\}$ is the set of global vertices. Note that $k\geq 2$ and $q\geq k+1$. Thus each vertex in $A$ has at least one friend in $B_1$ and each vertex in $B_1$ has at least one friend in $A$. Hence, any arbitrary vertex cover of size $p+q-2$ will be of the form $S_{ij}$ where $i\in [1,p]$ and $j\in [1,k]$. We show how to defend any attack when guards are occupying a vertex cover of this type.

Suppose some edge whose both endpoints have guards on them is attacked, then the guards exchange positions and we again have the same vertex cover.

The following types of edges can be attacked which are incident to $a_i$ or $b_j$. 
\begin{enumerate}
    \item \textbf{Some edge $a_ra_i$ is attacked:}
    
    The guard on $a_r$ moves to $a_i$ along the attacked edge. If $b_j$ is a friend of $a_r$, we are done and we have a new vertex cover $S_{rj}$. If not, then there exists $s\in[1,k]$ such that $s\neq j$ and $b_s$ is a friend of $a_r$. In this case, the guard on $b_s$ moves to $b_j$. The resulting vertex cover will be $S_{rs}$.
    \item \textbf{Some edge $b_sb_j$ is attacked where $s\leq k$:}
    
    The guard on $b_s$ moves to $b_j$ along the attacked edge. If $a_i$ is a friend of $b_s$, we are done and we have a new vertex cover $S_{is}$. If not, then there exists $r\in[1,p]$ such that $r\neq i$ and $a_r$ is a friend of $b_s$. In this case, the guard on $a_r$ moves to $a_i$. The resulting vertex cover will be $S_{rs}$.
    \item \textbf{Some edge $b_sb_j$ is attacked where $s>k$:}
    
    The guard on $b_s$ moves to $b_j$ along the attacked edge. If there exists $t\in[1,k]$ and $t\neq j$ such that $a_i$ is a friend of $b_t$, the guard on $b_t$ moves to $b_s$. The resulting vertex cover will be $S_{it}$. If not, then consider any $t\in[1,k]$ and $t\neq j$. There exists $r\in[1,p]$ and $r\neq i$ such that $a_r$ is a friend of $b_t$. The guard on $b_t$ moves to $b_s$ and the guard on $a_r$ moves to $a_i$. The resulting vertex cover will be $S_{rt}$.
    \item \textbf{Some edge $a_ib_s$ is attacked where $s\leq k$:}
    
    The guard on $b_s$ moves to $a_i$ along the attacked edge. There exists $r\in[1,p]$ such that $a_r$ is a friend of $b_s$ and $r\neq i$ (because $a_i$ is a neighbour of $b_s$ and $b_s$ is non-global). Now, the guard on $b_{k+1}$ moves to $b_s$ and the guard on $a_r$ moves to $b_{k+1}$. The resulting vertex cover will be $S_{rs}$.
    \item \textbf{Some edge $a_ib_s$ is attacked where $s>k$:}
    
    The guard on $b_s$ moves to $a_i$ along the attacked edge. If there exists $r\in[1,p]$ and $r\neq i$ such that $a_r$ is a friend of $b_j$, the guard on $a_r$ moves to $b_s$. The resulting vertex cover will be $S_{rj}$. If not, consider any $r\neq i$ and $r\in[1,p]$. There exists $t\in [1,k]$ and $t\neq j$ such that $b_t$ is a friend of $a_r$. The guard on $b_t$ moves to $b_j$ and the guard on $a_r$ moves to $b_s$. The resulting vertex cover will be $S_{rt}$.
    \item \textbf{Some edge $a_rb_j$ is attacked:}
    
    The guard on $a_r$ moves to $b_j$ along the attacked edge. Since $k\geq 2$, there exists $t\neq j$ such that $b_t$ is also non-global. If $b_t$ is a friend of $a_r$, the guard on $b_{k+1}$ moves to $a_i$, the guard on $b_t$ moves to $b_{k+1}$. The resulting vertex cover is $S_{rt}$. If not, and if $b_t$ is a friend of $a_i$, the guard on $b_t$ moves to $a_r$. The resulting vertex cover is $S_{rt}$. Otherwise $\exists s\in[1,p]$ such that $s\neq i,r$ and $a_s$ is a friend of $b_t$. In this case, the guard on $b_t$ moves to $a_r$ and the guard on $a_s$ moves to $a_i$. The resulting vertex cover is $S_{st}$.
\end{enumerate}
Therefore, in each case we can defend any attack using $p+q-2$ guards and thus $evc(G)=mvc(G)=p+q-2$ (c.f.~\Cref{fig:lem11case3b}). 
\begin{figure}
    \centering
    \begin{subfigure}[b]{\textwidth}
    \centering
    \begin{tikzpicture}[square/.style={regular polygon,regular polygon sides=4}]
    \tikz{ 
    
    \draw [DodgerBlue,thick,rounded corners,fill=LightGray!25] (1,-0.25) rectangle (7,1.25);
    \draw [IndianRed,thick,rounded corners,fill=LightGray!25] (0,2.25) rectangle (8,3.75);

    
    \draw[thick] (2,0.5) -- (3,3);
    \draw[thick] (2,0.5) -- (4,3);
    \draw[thick] (2,0.5) -- (7,3);
    \draw[thick] (3,0.5) -- (2,3);
    \draw[thick] (3,0.5) -- (3,3);
    \draw[thick] (3,0.5) -- (4,3);
    \draw[thick] (3,0.5) -- (7,3);
    \draw[thick] (4,0.5) -- (1,3);
    \draw[thick] (4,0.5) -- (5,3);
    \draw[thick] (4,0.5) -- (6,3);
    \draw[thick] (4,0.5) -- (7,3);
    \draw[thick] (5,0.5) -- (2,3);
    \draw[thick] (5,0.5) -- (3,3);
    \draw[thick] (5,0.5) -- (7,3);
    \draw[thick] (6,0.5) -- (5,3);
    \draw[thick] (6,0.5) -- (6,3);
    \draw[thick] (6,0.5) -- (7,3);
    
    \draw[thick,Red] (2,0.5) -- (1,3);
    \draw[->,thick] (5,3) to [out=80,in=100](7,3);

    \foreach \x in {2,3,4,6}
        \node at (\x, 0.5) [square,draw,fill=YellowGreen,minimum size=20pt] () {};
        
    \node at (2,0.5) {$a_1$};
    \node at (3,0.5) {$a_2$};
    \node at (4,0.5) {$a_3$};
    \node at (6,0.5) {$a_5$};

    \foreach \x in {2,3,4,5,6,7}
        \node at (\x, 3) [square,draw,fill=YellowGreen,minimum size=20pt] () {};

    \node at (2,3) {$b_2$};
    \node at (3,3) {$b_3$};
    \node at (4,3) {$b_4$};
    \node at (5,3) {$b_5$};
    \node at (6,3) {$b_6$};
    \node at (7,3) {$b_7$};
    
    \draw[fill=DodgerBlue!77] (5, 0.5) circle (0.15cm);
    \draw[fill=Crimson!77] (1, 3) circle (0.15cm);
    
    \node at (0.5,0.5) {$A$};
    \node at (-0.5,3) {$B$};
    
    \node at (0,4) {};

    }
    \end{tikzpicture}
    \caption{This figure demonstrates the case when there are at least two edges between $A$ and $B$, between $2$ and $6$ non-global vertices on the $B$ side and an edge between the $A$ side and the $B$ side (say $a_1b_1$) is attacked.}
    \label{fig:lem11case3b-1}
\end{subfigure}

\begin{subfigure}[b]{\textwidth}
    \centering
    \begin{tikzpicture}[square/.style={regular polygon,regular polygon sides=4}]
    \tikz{ 
    
    \draw [DodgerBlue,thick,rounded corners,fill=LightGray!25] (1,-0.25) rectangle (7,1.25);
    \draw [IndianRed,thick,rounded corners,fill=LightGray!25] (0,2.25) rectangle (8,3.75);

    
    \draw[thick] (2,0.5) -- (3,3);
    \draw[thick] (2,0.5) -- (4,3);
    \draw[thick] (2,0.5) -- (7,3);
    \draw[thick] (3,0.5) -- (2,3);
    \draw[thick] (3,0.5) -- (3,3);
    \draw[thick] (3,0.5) -- (4,3);
    \draw[thick] (3,0.5) -- (7,3);
    \draw[thick] (4,0.5) -- (1,3);
    \draw[thick] (4,0.5) -- (5,3);
    \draw[thick] (4,0.5) -- (6,3);
    \draw[thick] (5,0.5) -- (2,3);
    \draw[thick] (5,0.5) -- (3,3);
    \draw[thick] (4,0.5) -- (7,3);
    \draw[thick] (6,0.5) -- (5,3);
    \draw[thick] (6,0.5) -- (6,3);
    \draw[thick] (6,0.5) -- (7,3);
    
    \draw[->,thick,Red] (2,0.5) -- (1,2.8);
    \draw[->,thick] (5,3) to [out=60,in=150](7,3.4);
    \draw[->,thick] (7,3)--(5.1,0.65);
    
    \foreach \x in {2,3,4,6}
        \node at (\x, 0.5) [square,draw,fill=YellowGreen,minimum size=25pt] () {};
        
    \node at (2,0.5) {$a_1$};
    \node at (3,0.5) {$a_2$};
    \node at (4,0.5) {$a_3$};
    \node at (6,0.5) {$a_5$};

    \foreach \x in {2,3,4,5,6,7}
        \node at (\x, 3) [square,draw,fill=YellowGreen,minimum size=25pt] () {};

    \node at (2,3) {$b_2$};
    \node at (3,3) {$b_3$};
    \node at (4,3) {$b_4$};
    \node at (5,3) {$b_5$};
    \node at (6,3) {$b_6$};
    \node at (7,3) {$b_7$};
    
    \draw[fill=DodgerBlue!77] (5, 0.5) circle (0.15cm);
    \draw[fill=Crimson!77] (1, 3) circle (0.15cm);
    
    \node at (0.5,0.5) {$A$};
    \node at (-0.5,3) {$B$};
    
    \node at (0,4) {};

    }
    \end{tikzpicture}
    \caption{This figure demonstrates the defense for the attack in \Cref{fig:lem11case3b-1}. The guard on $a_1$ moves to $b_1$, the guard on $b_7$ moves to $a_4$ and the guard on $b_5$ moves to $b_7$.}
    \label{fig:lem11case3b-2}
\end{subfigure}

\begin{subfigure}[b]{\textwidth}
    \centering
    \begin{tikzpicture}[square/.style={regular polygon,regular polygon sides=4}]
    \tikz{ 
    
    \draw [DodgerBlue,thick,rounded corners,fill=LightGray!25] (1,-0.25) rectangle (7,1.25);
    \draw [IndianRed,thick,rounded corners,fill=LightGray!25] (0,2.25) rectangle (8,3.75);

    \draw[thick] (2,0.5) -- (3,3);
    \draw[thick] (2,0.5) -- (4,3);
    \draw[thick] (2,0.5) -- (7,3);
    \draw[thick] (3,0.5) -- (2,3);
    \draw[thick] (3,0.5) -- (3,3);
    \draw[thick] (3,0.5) -- (4,3);
    \draw[thick] (3,0.5) -- (7,3);
    \draw[thick] (4,0.5) -- (1,3);
    \draw[thick] (4,0.5) -- (5,3);
    \draw[thick] (4,0.5) -- (6,3);
    \draw[thick] (4,0.5) -- (7,3);
    \draw[thick] (5,0.5) -- (2,3);
    \draw[thick] (5,0.5) -- (3,3);
    \draw[thick] (5,0.5) -- (7,3);
    \draw[thick] (6,0.5) -- (5,3);
    \draw[thick] (6,0.5) -- (6,3);
    \draw[thick] (6,0.5) -- (7,3);
    
    \draw[thick,Red] (2,0.5) -- (1,3);
    \draw[->,thick] (5,3) to [out=80,in=100](7,3);

    \foreach \x in {3,4,5,6}
        \node at (\x, 0.5) [square,draw,fill=YellowGreen,minimum size=25pt] () {};

    \node at (3,0.5) {$a_2$};
    \node at (4,0.5) {$a_3$};
    \node at (5,0.5) {$a_4$};
    \node at (6,0.5) {$a_5$};

    \foreach \x in {1,2,3,4,6,7}
        \node at (\x, 3) [square,draw,fill=YellowGreen,minimum size=25pt] () {};

    \node at (1,3) {$b_1$};
    \node at (2,3) {$b_2$};
    \node at (3,3) {$b_3$};
    \node at (4,3) {$b_4$};
    \node at (6,3) {$b_6$};
    \node at (7,3) {$b_7$};
    
    \draw[fill=DodgerBlue!77] (2, 0.5) circle (0.15cm);
    \draw[fill=Crimson!77] (5, 3) circle (0.15cm);
    
    \node at (0.5,0.5) {$A$};
    \node at (-0.5,3) {$B$};
    
    \node at (0,4) {};

    }
    \end{tikzpicture}
    \caption{This figure demonstrates the position of guards after executing the defense in \Cref{fig:lem11case3b-2}.}
    \label{fig:lem11case3b-3}
\end{subfigure}
\caption{This figure demonstrates Case $3b$ from \Cref{big}.}
\label{fig:lem11case3b}
\end{figure}
\paragraph*{Case $c$: There are no global vertices on the $B$ side.}
\textbf{Case I: Only one vertex on the $A$ side has neighbours on the $B$ side.}

Let $a_1$ be the vertex on the $A$ side which has neighbours on the $B$ side. Since there are at least two edges between $A$ and $B$, $a_1$ has at least two neighbours on the $B$ side, let $\{b_1,b_2,\ldots,b_k\}$ be the neighbours of $a_1$ on the $B$ side. Note that $k\geq 2$ and $k<q$ (which is because $a_1$ is not global). Consider a vertex cover of size $p+q-2$ given by $S_{ij}$ where $i\neq 1$. We show that every attack can be defended and we can again get a vertex cover of the same type, i.e., we ensure that $a_1$ will always have a guard. Suppose some edge whose both endpoints have guards on them is attacked, then the guards exchange positions and we again have the same vertex cover. The following types of edges can be attacked which are incident to $a_i$ or $b_j$. 
\begin{enumerate}
    \item \textbf{Some edge $a_ia_r$ is attacked where $r\neq 1$:}
    
    The guard on $a_r$ moves to $a_i$ along the attacked edge. The new vertex cover will be $S_{rj}$. 
    \item\textbf{Some edge $a_1a_i$ is attacked:}
    
    The guard on $a_1$ moves to $a_i$ along the attacked edge and the guard on some $a_r$ (where $r\neq 1,i$) moves to $a_1$. The new vertex cover will be $S_{rj}$.
    \item \textbf{Some edge $b_jb_s$ is attacked:}
    
    The guard on $b_s$ moves to $b_j$ along the attacked edge. The new vertex cover will be $S_{is}$.
    \item \textbf{Some edge $a_1b_j$ is attacked (for $j\leq k$):}
    
    The guard on $a_1$ moves to $b_j$ along the attacked edge. The guard on some $b_s$ for $s<k,s\neq j$ moves to $a_1$. The new vertex cover will be $S_{is}$.
\end{enumerate}
Thus $p+q-2$ guards are sufficient to defend any attack on $G$ and $evc(G)=mvc(G)=p+q-2$ (c.f.~\Cref{fig:lem11case3cI}).

\begin{figure}
    \centering
    \begin{subfigure}[b]{\textwidth}
    \begin{tikzpicture}[square/.style={regular polygon,regular polygon sides=4}]
    \tikz{ 
    
    \draw [DodgerBlue,thick,rounded corners,fill=LightGray!25] (1,-0.25) rectangle (7,1.25);
    \draw [IndianRed,thick,rounded corners,fill=LightGray!25] (0,2.25) rectangle (8,3.75);

    
    \draw[thick] (2,0.5) -- (3,3);
    \draw[thick] (2,0.5) -- (4,3);
    \draw[thick] (2,0.5) -- (5,3);
    \draw[thick] (2,0.5) -- (2,3);
    \draw[thick] (2,0.5) -- (1,3);
    
    \draw[thick,Red] (2,0.5) -- (1,3);

    \foreach \x in {2,3,4,6}
        \node at (\x, 0.5) [square,draw,fill=YellowGreen,minimum size=20pt] () {};
        
    \node at (2,0.5) {$a_1$};
    \node at (3,0.5) {$a_2$};
    \node at (4,0.5) {$a_3$};
    \node at (6,0.5) {$a_5$};

    \foreach \x in {2,3,4,5,6,7}
        \node at (\x, 3) [square,draw,fill=YellowGreen,minimum size=20pt] () {};

    \node at (2,3) {$b_2$};
    \node at (3,3) {$b_3$};
    \node at (4,3) {$b_4$};
    \node at (5,3) {$b_5$};
    \node at (6,3) {$b_6$};
    \node at (7,3) {$b_7$};
    
    \draw[fill=DodgerBlue!77] (5, 0.5) circle (0.15cm);
    \draw[fill=Crimson!77] (1, 3) circle (0.15cm);
    
    \node at (0.5,0.5) {$A$};
    \node at (-0.5,3) {$B$};
    
    \node at (0,4) {};

    }
    \end{tikzpicture}
    \caption{This figure demonstrates the case when there are at least two edges between $A$ and $B$, no global vertices on the $B$ side and an edge between the $A$ side and the $B$ side (say $a_1b_1$) is attacked.}
    \label{fig:lem11case3cI-1}
\end{subfigure}

    \begin{subfigure}[b]{\textwidth}
    \begin{tikzpicture}[square/.style={regular polygon,regular polygon sides=4}]
    \tikz{ 
    
    \draw [DodgerBlue,thick,rounded corners,fill=LightGray!25] (1,-0.25) rectangle (7,1.25);
    \draw [IndianRed,thick,rounded corners,fill=LightGray!25] (0,2.25) rectangle (8,3.75);

    \draw[thick] (2,0.5) -- (3,3);
    \draw[thick] (2,0.5) -- (2,3);
    \draw[thick] (2,0.5) -- (5,3);
    \draw[->,thick,Red] (2,0.5) -- (1,2.8);
    \draw[->,thick](4,3)--(2.3,0.9);
    
    \foreach \x in {2,3,4,6}
        \node at (\x, 0.5) [square,draw,fill=YellowGreen,minimum size=25pt] () {};
        
    \node at (2,0.5) {$a_1$};
    \node at (3,0.5) {$a_2$};
    \node at (4,0.5) {$a_3$};
    \node at (6,0.5) {$a_5$};

    \foreach \x in {2,3,4,5,6,7}
        \node at (\x, 3) [square,draw,fill=YellowGreen,minimum size=25pt] () {};

    \node at (2,3) {$b_2$};
    \node at (3,3) {$b_3$};
    \node at (4,3) {$b_4$};
    \node at (5,3) {$b_5$};
    \node at (6,3) {$b_6$};
    \node at (7,3) {$b_7$};
    
    \draw[fill=DodgerBlue!77] (5, 0.5) circle (0.15cm);
    \draw[fill=Crimson!77] (1, 3) circle (0.15cm);
    
    \node at (0.5,0.5) {$A$};
    \node at (-0.5,3) {$B$};
    
    \node at (0,4) {};

    }
    \end{tikzpicture}
    \caption{This figure demonstrates the defense for the attack in \Cref{fig:lem11case3cI-1}. The guard on $a_1$ moves to $b_1$ and the guard on $b_4$ moves to $a_1$.}
    \label{fig:lem11case3cI-2}
\end{subfigure}

\begin{subfigure}[b]{\textwidth}
    \centering
    \begin{tikzpicture}[square/.style={regular polygon,regular polygon sides=4}]
    \tikz{ 
    
    \draw [DodgerBlue,thick,rounded corners,fill=LightGray!25] (1,-0.25) rectangle (7,1.25);
    \draw [IndianRed,thick,rounded corners,fill=LightGray!25] (0,2.25) rectangle (8,3.75);

    \draw[thick] (2,0.5) -- (3,3);
    \draw[thick] (2,0.5) -- (2,3);
    \draw[thick] (2,0.5) -- (5,3);
    \draw[thick] (2,0.5) -- (4,3);

    \draw[thick,Red] (2,0.5) -- (1,3);

    \foreach \x in {2,3,5,6}
        \node at (\x, 0.5) [square,draw,fill=YellowGreen,minimum size=25pt] () {};

    \node at (3,0.5) {$a_2$};
    \node at (2,0.5) {$a_1$};
    \node at (5,0.5) {$a_4$};
    \node at (6,0.5) {$a_5$};

    \foreach \x in {1,2,3,5,6,7}
        \node at (\x, 3) [square,draw,fill=YellowGreen,minimum size=25pt] () {};

    \node at (1,3) {$b_1$};
    \node at (2,3) {$b_2$};
    \node at (3,3) {$b_3$};
    \node at (5,3) {$b_5$};
    \node at (6,3) {$b_6$};
    \node at (7,3) {$b_7$};
    
    \draw[fill=DodgerBlue!77] (4, 0.5) circle (0.15cm);
    \draw[fill=Crimson!77] (4, 3) circle (0.15cm);
    
    \node at (0.5,0.5) {$A$};
    \node at (-0.5,3) {$B$};
    
    \node at (0,4) {};
    \node at (13,0) {};

    }
    \end{tikzpicture}
    \caption{This figure demonstrates the position of guards after executing the defense in \Cref{fig:lem11case3cI-2}.}
    \label{fig:lem11case3cI-3}
\end{subfigure}
\caption{This figure demonstrates Case $3cI$ from \Cref{big}.}
\label{fig:lem11case3cI}
\end{figure}

\textbf{Case II: Only one vertex on the $B$ side has neighbours on the $A$ side.}

Let $b_1$ be the vertex on the $B$ side which has neighbours on the $A$ side. Since there are at least two edges between $A$ and $B$, $b_1$ has at least two neighbours on the $B$ side, let $\{a_1,a_2,\ldots,a_k\}$ be the neighbours of $b_1$ on the $A$ side. Note that $k\geq 2$ and $k<p$ (which is because $b_1$ is not global). Consider a vertex cover of size $p+q-2$ given by $S_{ij}$ where $j\neq 1$. It can be shown every attack on $G$ can be defended and we can again get a vertex cover of the same type, i.e., we can always ensure that $b_1$ will always have a guard and maintain a vertex cover of size $p+q-2$ by using a strategy symmetric to the \textbf{Case I} above. Here, $evc(G)=mvc(G)=p+q-2$. 

\textbf{Case III: Both sides have at least two vertices which have neighbours on the other side.}

Since there are no global vertices, every vertex has at least one friend. Consider an arbitrary vertex cover $S_{ij}$ of size $p+q-2$ where $a_i$ is a friend of $b_j$. We show that any attack can be defended and we again get a vertex cover.

Suppose some edge whose both endpoints have guards on them is attacked, then the guards exchange positions and we again have the same vertex cover.
The following types of edges can be attacked which are incident to $a_i$ or $b_j$. 
\begin{enumerate}
    \item \textbf{Some edge $a_ra_i$ is attacked:}
    
    The guard on $a_i$ moves to $a_r$ along the attacked edge. If $a_r$ is a friend of $b_j$, we are done and we have a vertex cover $S_{rj}$. If not, then there exists some $s\in[1,q]$ such that $b_s$ is a friend of $a_r$. The guard on $b_s$ moves to $b_j$ and we have the vertex cover $S_{rs}$.
    \item \textbf{Some edge $b_jb_s$ is attacked:}
    
    The guard on $b_s$ moves to $b_j$. If $a_i$ is a friend of $b_s$, we are done and we have a vertex cover $S_{is}$. If not, then there exists some $r\in[1,p]$ such that $a_r$ is a friend of $b_s$. The guard on $a_r$ moves to $a_i$ and we have the vertex cover $S_{rs}$.
    \item \textbf{Some edge $a_ib_s$ is attacked ($s\neq j$):}
    
    The guard on $b_s$ moves to $a_i$. Since $a_i$ is not the only vertex on the $A$ side which has a neighbour on the $B$ side, there exists another vertex $a_r$ with $(r\neq a)$ which has a neighbour on the $B$ side.
    
    If there exists an edge $a_rb_s$, the guard from $a_r$ moves to $b_s$. If $a_r$ is a friend of $b_j$, we are done and we have a new vertex cover $S_{rj}$. If not, then there exists $b_t$ with $t\neq s,j$ such that $b_t$ is a friend of $a_r$. The guard on $b_t$ moves to $b_j$ and we have the new vertex cover $S_{rt}$.
    
    If there does not exist an edge $a_rb_s$, then there exists some edge $a_rb_t$ for $t\neq s$. If $t=j$, i.e., there exists some edge $a_rb_j$, the guard on $a_r$ moves to $b_j$ and we have the vertex cover $S_{rs}$. If $t\neq j$, the guard on $b_t$ moves to $b_j$, the guard on $a_r$ moves to $b_t$ and we have the vertex cover $S_{rs}$.
    \item \textbf{Some edge $a_rb_j$ is attacked ($r\neq i$):}
    
    We can show using an argument which is symmetric to the above one that $p+q-2$ guards are able to defend the attack and again form a vertex cover.
\end{enumerate}
Thus $p+q-2$ guards are sufficient to defend any attack on $G$ and $evc(G)=mvc(G)=p+q-2$.
Thus we have shown that except \textbf{Case 3$a$}, $evc(G)=mvc(G)=p+q-2$ and in \textbf{Case 3$a$} $evc(G)=p+q-1$ and $mvc(G)=p+q-2$ which completes the proof of \Cref{big}.
\end{proof}

\begin{figure}
    \centering
    \begin{subfigure}[b]{\textwidth}
    \begin{tikzpicture}[square/.style={regular polygon,regular polygon sides=4}]
    \tikz{ 
    
    \draw [DodgerBlue,thick,rounded corners,fill=LightGray!25] (1,-0.25) rectangle (7,1.25);
    \draw [IndianRed,thick,rounded corners,fill=LightGray!25] (0,2.25) rectangle (8,3.75);

    
    \draw[thick] (2,0.5) -- (3,3);
    \draw[thick] (2,0.5) -- (4,3);
    \draw[thick] (3,0.5) -- (2,3);
    \draw[thick] (3,0.5) -- (3,3);
    \draw[thick] (3,0.5) -- (4,3);
    \draw[thick] (3,0.5) -- (7,3);
    \draw[thick] (4,0.5) -- (1,3);
    \draw[thick] (4,0.5) -- (5,3);
    \draw[thick] (4,0.5) -- (6,3);
    \draw[thick] (5,0.5) -- (2,3);
    \draw[thick] (5,0.5) -- (3,3);
    \draw[thick] (5,0.5) -- (4,3);
    \draw[thick] (5,0.5) -- (7,3);
    \draw[thick] (6,0.5) -- (5,3);
    \draw[thick] (6,0.5) -- (6,3);
    \draw[->,thick](4,0.5) to (5,0.5);
    
    \draw[thick,Red] (2,0.5) -- (1,3);

    \foreach \x in {2,3,4,6}
        \node at (\x, 0.5) [square,draw,fill=YellowGreen,minimum size=20pt] () {};
        
    \node at (2,0.5) {$a_1$};
    \node at (3,0.5) {$a_2$};
    \node at (4,0.5) {$a_3$};
    \node at (6,0.5) {$a_5$};

    \foreach \x in {2,3,4,5,6,7}
        \node at (\x, 3) [square,draw,fill=YellowGreen,minimum size=20pt] () {};

    \node at (2,3) {$b_2$};
    \node at (3,3) {$b_3$};
    \node at (4,3) {$b_4$};
    \node at (5,3) {$b_5$};
    \node at (6,3) {$b_6$};
    \node at (7,3) {$b_7$};
    
    \draw[fill=DodgerBlue!77] (5, 0.5) circle (0.15cm);
    \draw[fill=Crimson!77] (1, 3) circle (0.15cm);
    
    \node at (0.5,0.5) {$A$};
    \node at (-0.5,3) {$B$};
    
    \node at (0,4) {};

    }
    \end{tikzpicture}
    \caption{This figure demonstrates the case when there are at least two edges between $A$ and $B$, no global vertices on the $B$ side, only one edge on the $A$ side is connected to the $B$ side and an edge between the $A$ side and the $B$ side (say $a_1b_1$) is attacked.}
    \label{fig:lem11case3cIII-1}
\end{subfigure}

\begin{subfigure}[b]{\textwidth}
    \begin{tikzpicture}[square/.style={regular polygon,regular polygon sides=4}]
    \tikz{ 
    
    \draw [DodgerBlue,thick,rounded corners,fill=LightGray!25] (1,-0.25) rectangle (7,1.25);
    \draw [IndianRed,thick,rounded corners,fill=LightGray!25] (0,2.25) rectangle (8,3.75);

    
    \draw[thick] (2,0.5) -- (3,3);
    \draw[thick] (3,0.5) -- (2,3);
    \draw[thick] (3,0.5) -- (3,3);
    \draw[thick] (3,0.5) -- (4,3);
    \draw[thick] (3,0.5) -- (7,3);
    \draw[thick] (4,0.5) -- (1,3);
    \draw[thick] (4,0.5) -- (5,3);
    \draw[thick] (4,0.5) -- (6,3);
    \draw[thick] (5,0.5) -- (2,3);
    \draw[thick] (5,0.5) -- (3,3);
    \draw[thick] (5,0.5) -- (4,3);
    \draw[thick] (5,0.5) -- (7,3);
    \draw[thick] (6,0.5) -- (5,3);
    \draw[thick] (6,0.5) -- (6,3);
    
    \draw[->,thick,Red] (2,0.5) -- (1,2.8);
    \draw[->,thick](4,3)--(2.3,0.9);
    \draw[->,thick](4,0.5) to (4.8,0.5);
    \draw[thick] (7,3)--(5,0.5);
    
    \foreach \x in {2,3,4,6}
        \node at (\x, 0.5) [square,draw,fill=YellowGreen,minimum size=25pt] () {};
        
    \node at (2,0.5) {$a_1$};
    \node at (3,0.5) {$a_2$};
    \node at (4,0.5) {$a_3$};
    \node at (6,0.5) {$a_5$};

    \foreach \x in {2,3,4,5,6,7}
        \node at (\x, 3) [square,draw,fill=YellowGreen,minimum size=25pt] () {};

    \node at (2,3) {$b_2$};
    \node at (3,3) {$b_3$};
    \node at (4,3) {$b_4$};
    \node at (5,3) {$b_5$};
    \node at (6,3) {$b_6$};
    \node at (7,3) {$b_7$};
    
    \draw[fill=DodgerBlue!77] (5, 0.5) circle (0.15cm);
    \draw[fill=Crimson!77] (1, 3) circle (0.15cm);
    
    \node at (0.5,0.5) {$A$};
    \node at (-0.5,3) {$B$};
    
    \node at (0,4) {};

    }
    \end{tikzpicture}
    \caption{This figure demonstrates the defense for the attack in \Cref{fig:lem11case3cIII-1}. The guard on $a_1$ moves to $b_1$, the guard on $b_4$ moves to $a_1$ and the guard on $a_3$ moves to $a_4$.}
    \label{fig:lem11case3cIII-2}
\end{subfigure}

\begin{subfigure}[b]{\textwidth}
    \begin{tikzpicture}[square/.style={regular polygon,regular polygon sides=4}]
    \tikz{ 
    
    \draw [DodgerBlue,thick,rounded corners,fill=LightGray!25] (1,-0.25) rectangle (7,1.25);
    \draw [IndianRed,thick,rounded corners,fill=LightGray!25] (0,2.25) rectangle (8,3.75);

    \draw[thick] (2,0.5) -- (3,3);
    \draw[thick] (2,0.5) -- (4,3);
    \draw[thick] (3,0.5) -- (2,3);
    \draw[thick] (3,0.5) -- (3,3);
    \draw[thick] (3,0.5) -- (4,3);
    \draw[thick] (3,0.5) -- (7,3);
    \draw[thick] (4,0.5) -- (1,3);
    \draw[thick] (4,0.5) -- (5,3);
    \draw[thick] (4,0.5) -- (6,3);
    \draw[thick] (5,0.5) -- (2,3);
    \draw[thick] (5,0.5) -- (3,3);
    \draw[thick] (5,0.5) -- (4,3);
    \draw[thick] (5,0.5) -- (7,3);
    \draw[thick] (6,0.5) -- (5,3);
    \draw[thick] (6,0.5) -- (6,3);
    \draw[->,thick](4,0.5) to (5,0.5);
    
    \draw[thick,Red] (2,0.5) -- (1,3);
    
    \draw[thick,Red] (2,0.5) -- (1,3);

    \foreach \x in {2,3,5,6}
        \node at (\x, 0.5) [square,draw,fill=YellowGreen,minimum size=25pt] () {};

    \node at (3,0.5) {$a_2$};
    \node at (2,0.5) {$a_1$};
    \node at (5,0.5) {$a_4$};
    \node at (6,0.5) {$a_5$};

    \foreach \x in {1,2,3,5,6,7}
        \node at (\x, 3) [square,draw,fill=YellowGreen,minimum size=25pt] () {};

    \node at (1,3) {$b_1$};
    \node at (2,3) {$b_2$};
    \node at (3,3) {$b_3$};
    \node at (5,3) {$b_5$};
    \node at (6,3) {$b_6$};
    \node at (7,3) {$b_7$};
    
    \draw[fill=DodgerBlue!77] (4, 0.5) circle (0.15cm);
    \draw[fill=Crimson!77] (4, 3) circle (0.15cm);
    
    \node at (0.5,0.5) {$A$};
    \node at (-0.5,3) {$B$};
    
    \node at (0,4) {};

    }
    \end{tikzpicture}
    \caption{This figure demonstrates the position of guards after executing the defense in \Cref{fig:lem11case3cIII-2}.}
    \label{fig:lem11case3cIII-3}
\end{subfigure}
\caption{This figure demonstrates Case $3c III$ from \Cref{big}.}
\end{figure}

Now we determine the value of $evc(G)$ for a co-bipartite graph $G$ with all the notations as described above and $p\leq 3$.
\begin{lemma}\label{p=0}
For a co-bipartite graph $G$ with all the notations as described above and $p=0$, we have $evc(G)=mvc(G)=p+q-1$.
\end{lemma}
\begin{proof}
If $p=0$, then the graph is a clique with $q$ vertices and thus $evc(G)=mvc(G)=q-1=p+q-1$. 
\end{proof}
\begin{lemma}\label{p=1}
For a co-bipartite graph $G$ with all the notations as described above and $p=1$, $evc(G)=mvc(G)=p+q-2$ if and only if $a_1$ has no neighbours on the $B$ side.
\end{lemma}
\begin{proof}
Since $q\geq p$ and $p=1$, we have $q\geq 1$. If $q=1$, then if $a_1b_1\in E(G)$, then $a_1$ will be a global vertex which is not allowed according to the convention. Therefore $a_1$ has no neighbours on the $B$ side. Since there are no edges, no guards are required to protect $G$. Therefore $evc(G)=mvc(G)=0=p+q-2$.

Now consider the case where $q>1$.
If $a_1$ has no neighbours on the $B$ side, then $evc(G)=mvc(G)=q-1=p+q-2$ because $B$ is a clique of size $q$ and $A$ is an isolated vertex.

If $a_1$ has at least one neighbour on the other side, then we show that $p+q-2=q-1$ guards are not sufficient to protect $G$. Let $\{b_1,b_2,\ldots,b_k\}$ be the neighbours of $a_1$ on the $B$ side where $1\leq k<q$. A vertex cover of $G$ of size $q-1$ will be of the form $S_{1j}$ where $j>k$. Now attack $a_1b_1$. To defend this attack, the guard on $b_1$ is forced to move to $a_1$. So there are only $q-2$ guards on $B$ to protect a clique of size $q$. Therefore some edge must be left vulnerable. Thus we have $evc(G)=p+q-1=q$ and $evc(G)\neq mvc(G)$.
\end{proof}
\begin{lemma}\label{p=q=2}
For a co-bipartite graph $G$ with all the notations as described above and $p=q=2$, we have $evc(G)=mvc(G)=p+q-2=2$ if and only if $a_1$ and $a_2$ each have one neighbour on the other side which is unique or $a_1$ and $a_2$ have no neighbour on the other side.
\end{lemma}
\begin{proof}
Since $p=q=2$ and no vertex on the $A$ side is global according to the convention, we have the following possibilities:
\begin{enumerate}
    \item $a_1$ and $a_2$ have no neighbour on the $B$ side.
    \item $a_1$ and $a_2$ have one neighbour each on the $B$ side which is unique.
    \item $a_1$ and $a_2$ have one common neighbour on the $B$ side.
    \item Exactly one of $a_1$ and $a_2$ has a neighbour on the $B$ side.
\end{enumerate}
\paragraph*{Case $1$: $a_1$ and $a_2$ have no neighbour on the $B$ side.}
Since we have two disconnected edges, i.e., two disjoint cliques of size $1$, so we have $evc(G)=mvc(G)=2=p+q-2$.
\paragraph*{Case $2$: $a_1$ and $a_2$ have one neighbour each on the $B$ side which is unique.}
Let $b_1$ be the neighbour of $a_1$ and $b_2$ be the neighbour of $a_2$. The two vertex covers of size $p+q-2$ are $S_{12}$ and $S_{21}$. Without loss of generality we start with the vertex cover $S_{12}$ and show that any attack can be defended.

If any edge with guards on both of its endpoints is attacked, the guards exchange positions and we again have the same vertex cover.

The following types of edges can be attacked which are incident to $a_i$ or $b_j$. 
\begin{enumerate}
    \item \textbf{The edge $a_1a_2$ is attacked:}
    
    The guard on $a_2$ moves to $a_1$ along the attacked edge and the guard on $b_1$ moves to $b_2$ and thus the resultant vertex cover is $S_{21}$.
    \item \textbf{The edge $b_1b_2$ is attacked:}
    
    The guard on $b_1$ moves to $b_2$ and the guard on $a_2$ moves to $a_1$ and the resultant vertex cover is $S_{21}$.
    \item \textbf{The edge $a_1b_1$ is attacked:}
    
    The guard on $b_1$ moves to $a_1$ and the guard on $a_2$ moves to $b_2$. The resultant vertex cover is $S_{21}$. 
    \item \textbf{The edge $a_2b_2$ is attacked:}
    
    The guard on $a_2$ moves to $b_2$ and the guard on $b_1$ moves to $a_1$. The resultant vertex cover is $S_{21}$.
\end{enumerate}

Thus we have $evc(G)=mvc(G)=2$ (c.f.~\Cref{fig:lem14case2}).

\begin{figure}
\centering
\begin{subfigure}[b]{\textwidth}
\begin{tikzpicture}[square/.style={regular polygon,regular polygon sides=4}]
    \tikz{ 
    
    \draw [DodgerBlue,thick,rounded corners,fill=LightGray!25] (1.5,-0.25) rectangle (4,1.25);
    \draw [IndianRed,thick,rounded corners,fill=LightGray!25] (1.5,2.25) rectangle (4,3.75);

    \draw[thick] (3,0.5) -- (3,3);

    \draw[thick,Red] (2,0.5) -- (2,3);
    \draw[thick] (2,0.5) -- (3,0.5);
    \draw[thick] (2,3) -- (3,3);

    \foreach \x in {3}
        \node at (\x, 0.5) [square,draw,fill=YellowGreen,minimum size=20pt] () {};

    \node at (3,0.5) {$a_2$};

    \foreach \x in {2}
        \node at (\x, 3) [square,draw,fill=YellowGreen,minimum size=20pt] () {};

    \node at (2,3) {$b_1$};
    
    \draw[fill=DodgerBlue!77] (2, 0.5) circle (0.15cm);
    \draw[fill=Crimson!77] (3, 3) circle (0.15cm);
    
    \node at (0.5,0.5) {$A$};
    \node at (0.5,3) {$B$};
    
    \node at (0,4) {};

    }
    \end{tikzpicture}
   \caption{This figure demonstrates Case $2$, \Cref{p=q=2} when the edge $a_1b_1$ is attacked.}
   \label{fig:lem14case2-1} 
\end{subfigure}

\begin{subfigure}[b]{\textwidth}
\begin{tikzpicture}[square/.style={regular polygon,regular polygon sides=4}]
    \tikz{ 
    
    \draw [DodgerBlue,thick,rounded corners,fill=LightGray!25] (1.5,-0.25) rectangle (4,1.25);
    \draw [IndianRed,thick,rounded corners,fill=LightGray!25] (1.5,2.25) rectangle (4,3.75);

    \draw[->,thick] (3,0.5) -- (3,2.8);

    \draw[->,thick,Red] (2,3) -- (2,0.7);
    \draw[thick] (2,0.5) -- (3,0.5);
    \draw[thick] (2,3) -- (3,3);

    \foreach \x in {3}
        \node at (\x, 0.5) [square,draw,fill=YellowGreen,minimum size=20pt] () {};

    \node at (3,0.5) {$a_2$};

    \foreach \x in {2}
        \node at (\x, 3) [square,draw,fill=YellowGreen,minimum size=20pt] () {};

    \node at (2,3) {$b_1$};
    
    \draw[fill=DodgerBlue!77] (2, 0.5) circle (0.15cm);
    \draw[fill=Crimson!77] (3, 3) circle (0.15cm);
    
    \node at (0.5,0.5) {$A$};
    \node at (0.5,3) {$B$};

    \node at (0,4) {};

    }
    \end{tikzpicture}
   \caption{This figure shows the defense for the attack in \Cref{fig:lem14case2-1}. Here, the guard on $b_1$ moves to $a_1$ and the guard on $a_2$ moves to $b_2$.}
   \label{fig:lem14case2-2}
\end{subfigure}

\begin{subfigure}[b]{\textwidth}
\begin{tikzpicture}[square/.style={regular polygon,regular polygon sides=4}]
    \tikz{ 
    
    \draw [DodgerBlue,thick,rounded corners,fill=LightGray!25] (1.5,-0.25) rectangle (4,1.25);
    \draw [IndianRed,thick,rounded corners,fill=LightGray!25] (1.5,2.25) rectangle (4,3.75);

    \draw[thick] (3,0.5) -- (3,3);

    \draw[thick,Red] (2,0.5) -- (2,3);
    \draw[thick] (2,0.5) -- (3,0.5);
    \draw[thick] (2,3) -- (3,3);

    \foreach \x in {2}
        \node at (\x, 0.5) [square,draw,fill=YellowGreen,minimum size=20pt] () {};

    \node at (2,0.5) {$a_1$};

    \foreach \x in {3}
        \node at (\x, 3) [square,draw,fill=YellowGreen,minimum size=20pt] () {};

    \node at (3,3) {$b_2$};
    
    \draw[fill=DodgerBlue!77] (3, 0.5) circle (0.15cm);
    \draw[fill=Crimson!77] (2, 3) circle (0.15cm);
    
    \node at (0.5,0.5) {$A$};
    \node at (0.5,3) {$B$};

    \node at (0,4) {};

    }
    \end{tikzpicture}
   \caption{This figure shows the state of the graph after the defense in figure \ref{fig:lem14case2-2} has been executed.}
   \label{fig:lem14case2-3}
\end{subfigure}

\caption{A figure corresponding to Case $2$ from \Cref{p=q=2}.
\label{fig:lem14case2}
}
\end{figure}

\paragraph*{Case $3$: $a_1$ and $a_2$ have one common neighbour on the $B$ side.}
Let $b_1$ be the common neighbour of $a_1$ and $a_2$. Since both $a_1$ and $a_2$ are non-global, we know that $a_1$ and $a_2$ are not adjacent with $b_2$. Any vertex cover of size $p+q-2$ will be of the form $S_{i2}$ where $i\in\{1,2\}$. Now attack $b_1b_2$, the guard on $b_1$ is forced to move to $b_2$. If the guard on $a_{3-i}$ moves to $b_1$, $a_1a_2$ will be unprotected, otherwise one of $a_1b_1$ and $a_2b_1$ will be unprotected.

Therefore, $p+q-2$ guards are not sufficient and $evc(G)\neq mvc(G)$.
\paragraph*{Case $4$: Exactly one of $a_1$ and $a_2$ has a neighbour on the $B$ side.}
Let $a_1$ be the vertex which has a neighbour (say $b_1$) on the $B$ side. The vertex covers of size $2$ will be $S_{12},S_{21}$ and $S_{22}$. If we have $S_{12}$ as the initial vertex cover, attack $a_1b_1$. The guard on $b_1$ is forced to move to $a_1$. The guard on $a_2$ cannot move anywhere and thus $b_1b_2$ is left unprotected. Similarly, if we have $S_{21}$ as the initial vertex cover, attack $a_1b_1$. The guard on $a_1$ is forced to move to $b_1$. The guard on $b_2$ cannot move anywhere and thus $a_1a_2$ is left unprotected. If we have $S_{22}$ as the initial vertex cover, attack $a_1a_2$. The guard on $a_1$ is forced to move to $a_2$ and to maintain a vertex cover, the guard on $b_1$ is forced to remain in his position. So we have the vertex cover $S_{12}$. We have already shown that it is possible for the attacker to win from the vertex cover $S_{12}$.

Thus $p+q-2$ guards are not sufficient to protect the graph and $evc(G)\neq mvc(G)$.
\end{proof}
\begin{lemma}\label{p=2q>2}
For a co-bipartite graph $G$ with all the notations as described above and $p=2,q>2$, we have $evc(G)\neq mvc(G)$ if and only if exactly one vertex from $A$ is adjacent to at least one vertex from the $B$ side, or there is exactly one non-global vertex on the $B$ side.
\end{lemma}
\begin{proof}
Since $p=2$, by \Cref{vc} we have $mvc(G)=p+q-2=2+q-2=q$. 
\paragraph*{Case $1$: There is no edge between $A$ and $B$.}
Since $G$ consists of two disjoint cliques of size $2$ and $q$, we have $evc(G)=mvc(G)=2-1+q-1=q=p+q-2$.

\paragraph*{Case $2$: Exactly one vertex from $A$ is adjacent to at least one vertex from the $B$ side.}
Let $a_1$ be the vertex from the $A$ side which is adjacent to at least one vertex from the $B$ side. Let $\{b_1,b_2,\ldots,b_k\}$ be the vertices on the $B$ side which are adjacent to $a_1$ where $1\leq k<q$. Any vertex cover of size $q$ will be of the form $S_{1j}$ where $j>k$ or $S_{2i}$ where $i\in[1,q]$. Suppose we have an initial vertex cover $S_{2i}$, now attack $b_1b_2$. The guard on $b_1$ is forced to move to $b_2$. So now we have a vertex cover which contains $b_2$, i.e., the guards will rearrange to get a vertex cover of the form $S_{1j}$ where $j>k$. Therefore, we can assume without loss of generality that the initial vertex cover is of the form $S_{1j}$ where $j>k$. Now attack $a_1b_1$. The guard on $b_1$ is forced to move to $a_1$. The guard on $a_2$ cannot move to the $B$ side. So there are $q-2$ guards on $B$ side to protect a clique of size $q$, thus some edge is unprotected. Thus $q$ guards are not sufficient to protect $G$ and $evc(G)\neq mvc(G)$.
\paragraph*{Case $3$: Both $a_1$ and $a_2$ have at least one neighbour from the $B$ side}
Let $B=B_1\cup B_2 \cup B_3 \cup B_4$ where $B_1$ is the set of vertices which are adjacent to $a_1$ but not adjacent to $a_2$, $B_2$ is the set of vertices which are adjacent to $a_2$ but not adjacent to $a_1$, $B_3$ is the set of vertices which are adjacent to both $a_1$ and $a_2$; and $B_4$ is the set of vertices which are neither adjacent to $a_1$ nor adjacent to $a_2$.

Since both $a_1$ and $a_2$ have at least one neighbour from the $B$ side, we have $|B_1\cup B_3|\geq 1$ and $|B_2\cup B_3|\geq 1$. Since both $a_1$ and $a_2$ are non-global, they both have at least one friend,i.e., $|B_2\cup B_4|\geq 1$ and $|B_1\cup B_4|\geq 1$. Therefore, $|B_3|<q$.

Any vertex cover of size $p+q-2$ is given by $S_{ij}$ where $b_j\in B_{3-i}\cup B_4$. 
\paragraph*{Case $a$: $B_3=\phi$}
Since $B_3=\phi$, we have $|B_1|\geq 1$ and $|B_2|\geq 1$.
Without loss of generality we start with a vertex cover $S_{1j}$ where $b_j \in B_2\cup B_4$. We show that each type of attack on this vertex cover can be defended and we will again get a vertex cover of size $p+q-2$. 

If any edge whose both endpoints have guards on them is attacked, then the guards exchange positions and we again have the same vertex cover.

The following types of edges can be attacked which are incident to $a_1$ or $b_j$.
\begin{enumerate}
    \item \textbf{The edge $a_1a_2$ is attacked:}
    
    The guard on $a_2$ moves to $a_1$ along the attacked edge. If $b_j \in B_4$, we are done and we have the new vertex cover $S_{2j}$. If not, then $b_j\in B_2$ and there exists $b_k\in B_1$ (because $|B_1|\geq 1$). Now the guard on $b_k$ moves to $b_j$ and we have the new vertex cover $S_{2k}$.
    \item \textbf{The edge $b_kb_j$ is attacked for $b_k\in B_2\cup B_4$:}
    
    The guard on $b_k$ moves to $b_j$ along the attacked edge. The resultant vertex cover will be $S_{1k}$.
    \item \textbf{The edge $b_kb_j$ is attacked for $b_k \in B_1$:}
    
    The guard on $b_k$ moves to $b_j$ along the attacked edge and the guard on $a_2$ moves to $a_1$. The resultant vertex cover will be $S_{2k}$.
    \item \textbf{The edge $a_1b_k$ is attacked for $b_k \in B_1$:}
    
    The guard on $b_k$ moves to $a_1$ along the attacked edge. If $b_j\in B_2$, the guard on $a_2$ moves to $b_j$ and we have the vertex cover $S_{2k}$. If $b_j \in B_4$, since $|B_2|\geq 1$, there exists $b_r\in B_2$. The guard on $b_r$ moves to $b_k$ and the guard on $a_2$ moves to $b_r$. The new vertex cover will be $S_{2j}$.
    \item \textbf{The edge $a_2b_j$ is attacked for $b_j\in B_2$:}
    
    The guard on $a_2$ moves to $b_j$ along the attacked edge. There exists $b_k\in B_1$ because $|B_1|\geq 1$. The guard on $b_k$ moves to $a_1$ and we have the new vertex cover $S_{2k}$. 
\end{enumerate}
Thus we have shown that each type of attack can be defended on a vertex cover of the form $S_{1j}$ where $b_j\in B_2\cup B_4$. So, $evc(G)=mvc(G)=p+q-2=q$.

\paragraph*{Case $b$: $1\leq|B_3|\leq q-2$}

Without loss of generality we start with a vertex cover $S_{1j}$ where $b_j \in B_2\cup B_4$. We show that each type of attack on this vertex cover can be defended and we will again get a vertex cover of size $p+q-2$. 

If any edge whose both endpoints have guards on them is attacked, the guards exchange positions along the attacked edge and we again have the same vertex cover.

The following types of edges can be attacked which are incident to $a_1$ or $b_j$.
\begin{enumerate}
    \item \textbf{The edge $a_1a_2$ is attacked:}
    
    The guard on $a_2$ moves to $a_1$ along the attacked edge. If $b_j \in B_4$, we are done and we have the new vertex cover $S_{2j}$. If not, then $b_j\in B_2$ and there exists $b_k\in B_1\cup B_4$ (because $|B_1\cup B_4|\geq 1$). Now the guard on $b_k$ moves to $b_j$ and we have  the new vertex cover $S_{2k}$.
    \item \textbf{The edge $b_kb_j$ is attacked for $b_k\in B_2\cup B_4$:}
    
    The guard on $b_k$ moves to $b_j$. The resultant vertex cover will be $S_{1k}$.
    \item \textbf{The edge $b_kb_j$ is attacked for $b_k \in B_1$:}
    
    The guard on $b_k$ moves to $b_j$ along the attacked edge and the guard on $a_2$ moves to $a_1$. The resultant vertex cover will be $S_{2k}$. If $b_kb_j$ is attacked for $b_k\in B_3$, guard on $b_k$ moves to $b_j$, there exists $b_r\in B_1\cup B_2\cup B_4$ where $r\neq j$ because $|B_3|\leq q-2$. Suppose there exists $b_r\in B_2\cup B_4$ with $r\neq j$, then guard on $b_r$ moves to $b_k$ and we have the vertex cover $S_{1r}$. If not, then there exists $b_r\in B_1$. Then the guard on $b_r$ moves to $b_k$ and the guard on $a_2$ moves to $a_1$. The resultant vertex cover will be $S_{2r}$.
    \item \textbf{The edge $a_1b_k$ is attacked for $b_k \in B_1$:}
    
    The guard on $b_k$ moves to $a_1$ along the attacked edge. If $b_j\in B_2$, the guard on $a_2$ moves to $b_j$ and we have the vertex cover $S_{2k}$. If $b_j \in B_4$, since $|B_2\cup B_3|\geq 1$, there exists $b_r\in B_2\cup B_3$. The guard on $b_r$ moves to $b_k$ and the guard on $a_2$ moves to $b_r$. The new vertex cover will be $S_{2j}$.
    \item \textbf{The edge $a_1b_k$ is attacked for $b_k \in B_3$:}
    
    The guard on $b_k$ moves to $a_1$ along the attacked edge. If $b_j\in B_4$, the guard on $a_2$ moves to $b_k$ and we have the vertex cover $S_{2j}$. Otherwise, $b_j\in B_2$ and since $|B_1\cup B_4|\geq 1$, there exists $b_r\in B_1\cup B_4$. Now the guard on $b_r$ moves to $b_k$ and the guard on $a_2$ moves to $b_j$, thus we have the vertex cover $S_{2r}$.
    \item \textbf{The edge $a_2b_j$ is attacked for $b_j\in B_2$:}
    
    The guard on $a_2$ moves to $b_j$ along the attacked edge. There exists $b_k\in B_1\cup B_4$ because $|B_1\cup B_4|\geq 1$. If there exists $b_k\in B_1$, the guard on $b_k$ moves to $a_1$ and we have the new vertex cover $S_{2k}$. Otherwise, $b_k\in B_4$. Consider some $b_r\in B_3$. The guard on $b_r$ moves to $a_2$ and the guard on $b_k$ moves to $b_r$. The new vertex cover will be $S_{1k}$.
\end{enumerate}

Thus we have shown that each type of attack can be defended on a vertex cover of the form $S_{1j}$ where $b_j\in B_2\cup B_4$. So, $evc(G)=mvc(G)=p+q-2=q$ (c.f.~\Cref{fig:lem15case3b}).

\begin{figure}
    \centering
    \begin{subfigure}[b]{\textwidth}
    \begin{tikzpicture}[square/.style={regular polygon,regular polygon sides=4}]
    \tikz{ 
    
    \draw [DodgerBlue,thick,rounded corners,fill=LightGray!25] (1.5,-0.25) rectangle (4,1.25);
    \draw [IndianRed,thick,rounded corners,fill=LightGray!25] (1.5,2.25) rectangle (8,3.75);

    \draw[thick] (2,0.5) -- (2,3);
    \draw[thick] (3,0.5) -- (3,3);
    \draw[thick] (3,0.5) -- (4,3);
    \draw[thick] (3,0.5) -- (5,3);
    \draw[thick] (2,0.5) -- (6,3);
    \draw[thick] (3,0.5) -- (6,3);

    \draw[thick,Red] (2,3) -- (3,3);
    \draw[thick] (2,0.5) -- (3,0.5);
   
    \foreach \x in {3}
        \node at (\x, 0.5) [square,draw,fill=YellowGreen,minimum size=20pt] () {};

    \node at (3,0.5) {$a_2$};

    \foreach \x in {2,4,5,6,7}
        \node at (\x, 3) [square,draw,fill=YellowGreen,minimum size=20pt] () {};

    \node at (2,3) {$b_1$};
    \node at (4,3) {$b_3$};
    \node at (5,3) {$b_4$};
    \node at (6,3) {$b_5$};
    \node at (7,3) {$b_6$};
    
    \draw[fill=DodgerBlue!77] (2, 0.5) circle (0.15cm);
    \draw[fill=Crimson!77] (3, 3) circle (0.15cm);
    
    \node at (0.5,0.5) {$A$};
    \node at (0.5,3) {$B$};
    
    \node at (0,4) {};

    }
    \end{tikzpicture}
    \caption{This figure demonstrates the case when $p=2,q>2$, between $1$ and $q-2$ global vertices on the $B$ side and an edge on the $B$ side (say $b_1b_2$) is attacked.}
    \label{fig:lem15case3b-1}
\end{subfigure}
    \begin{subfigure}[b]{\textwidth}
    \begin{tikzpicture}[square/.style={regular polygon,regular polygon sides=4}]
    \tikz{ 
    \draw [DodgerBlue,thick,rounded corners,fill=LightGray!25] (1.5,-0.25) rectangle (4,1.25);
    \draw [IndianRed,thick,rounded corners,fill=LightGray!25] (1.5,2.25) rectangle (8,3.75);

    \draw[thick] (2,0.5) -- (2,3);
    \draw[thick] (3,0.5) -- (3,3);
    \draw[thick] (3,0.5) -- (4,3);
    \draw[thick] (3,0.5) -- (5,3);

    \draw[->,thick,Red] (2,3) -- (2.8,3);
    \draw[thick] (2,0.5) -- (3,0.5);
    \draw[->,thick] (6,3) -- (2.1,0.7);
    \draw[->,thick] (3,0.5) -- (5.7,2.7);
   
    \foreach \x in {3}
        \node at (\x, 0.5) [square,draw,fill=YellowGreen,minimum size=20pt] () {};

    \node at (3,0.5) {$a_2$};

    \foreach \x in {2,4,5,6,7}
        \node at (\x, 3) [square,draw,fill=YellowGreen,minimum size=20pt] () {};

    \node at (2,3) {$b_1$};
    \node at (4,3) {$b_3$};
    \node at (5,3) {$b_4$};
    \node at (6,3) {$b_5$};
    \node at (7,3) {$b_6$};
    
    \draw[fill=DodgerBlue!77] (2, 0.5) circle (0.15cm);
    \draw[fill=Crimson!77] (3, 3) circle (0.15cm);
    
    \node at (0.5,0.5) {$A$};
    \node at (0.5,3) {$B$};

    \node at (0,4) {};

    }
    \end{tikzpicture}
    \caption{This figure demonstrates the defense for the attack in \Cref{fig:lem15case3b-1}. The guard on $b_1$ moves to $b_2$, the guard on $b_5$ moves to $a_1$ and the guard on $a_2$ moves to $b_5$.}
    \label{fig:lem15case3b-2}
\end{subfigure}
\begin{subfigure}[b]{\textwidth}
    \centering
    \begin{tikzpicture}[square/.style={regular polygon,regular polygon sides=4}]
    \tikz{ 
    \draw [DodgerBlue,thick,rounded corners,fill=LightGray!25] (1.5,-0.25) rectangle (4,1.25);
    \draw [IndianRed,thick,rounded corners,fill=LightGray!25] (1.5,2.25) rectangle (8,3.75);
    \draw[thick] (2,0.5) -- (2,3);
    \draw[thick] (3,0.5) -- (3,3);
    \draw[thick] (3,0.5) -- (4,3);
    \draw[thick] (3,0.5) -- (5,3);
    \draw[thick] (2,0.5) -- (6,3);
    \draw[thick] (3,0.5) -- (6,3);

    \draw[thick,Red] (2,3) -- (3,3);
    \draw[thick] (2,0.5) -- (3,0.5);
   
    \foreach \x in {2}
        \node at (\x, 0.5) [square,draw,fill=YellowGreen,minimum size=20pt] () {};

    \node at (2,0.5) {$a_1$};

    \foreach \x in {3,4,5,6,7}
        \node at (\x, 3) [square,draw,fill=YellowGreen,minimum size=20pt] () {};

    \node at (3,3) {$b_2$};
    \node at (4,3) {$b_3$};
    \node at (5,3) {$b_4$};
    \node at (6,3) {$b_5$};
    \node at (7,3) {$b_6$};
    
    \draw[fill=DodgerBlue!77] (3, 0.5) circle (0.15cm);
    \draw[fill=Crimson!77] (2, 3) circle (0.15cm);
    
    \node at (0.5,0.5) {$A$};
    \node at (0.5,3) {$B$};

    \node at (0,4) {};
     
    \node at (14,4) {};

    }
    \end{tikzpicture}
    \caption{This figure demonstrates the position of guards after executing the defense in \Cref{fig:lem15case3b-2}.}
    \label{fig:lem15case3b-3}
\end{subfigure}
\caption{This figure demonstrates Case $3b$ from \Cref{p=2q>2}.}
\label{fig:lem15case3b}
\end{figure}

\paragraph*{Case $c$: $|B_3|=q-1$}
Let $b_1\notin B_3$. Therefore, $b_1$ has at least one friend, say $a_1$ (without loss of generality). So, $b_1\in B_2\cup B_4$. Suppose $b_1\in B_2$. Then $a_2$ will be a global vertex which is not possible according to the convention. So $b_1 \in B_4$. Thus any vertex cover of size $p+q-2$ will be of the form $S_{i1}$ where $i\in \{1,2\}$. Now attack $b_1b_2$. The guard on $b_2$ will be forced to move to $b_1$. If some guard from the $A$ side comes to the $B$ side, $a_1a_2$ will be unprotected. If not, then there are $q-1$ guards on the $B$ side out of which one is on $b_1$. Therefore some $b_j\in B_3$ is left empty and since some $a_r$ from the $A$ side will be empty ($r$ may be equal to $i$), the edge $a_rb_j$ will be vulnerable.

Thus $p+q-2$ guards are not sufficient to protect $G$ and thus $evc(G)\neq mvc(G)$.
\end{proof}
 
It follows from~\Cref{big,p=0,p=1,p=q=2,p=2q>2} that given a co-bipartite graph, its eternal vertex cover number can be calculated based on the number of vertices and degree of each vertex. Thus we have a polynomial time algorithm for \EVC{} on co-bipartite graphs. This concludes the proof of Theorem 2.


\section{Concluding Remarks}
\label{sec:concl}

We established the hardness of~\EVC{} on bipartite graphs of constant diameter. We also showed that, under standard complexity-theoretic assumptions, the problem does not admit a polynomial compression on these graph classes when parameterized by the number of guards. Because of the relationship between $mvc(G)$ and $evc(G)$, this also implies hardness when paramterized by the vertex cover number. In the light of these developments, it will be interesting to pursue improved FPT and approximation algorithms for these classes of graphs. It is also unclear if~\EVC{} is in NP even on these classes of graphs. 

\newpage 

\bibliography{refs}

\begin{thebibliography}{}

\bibitem[Araki et~al., 2015]{AFI2015}
Araki, H., Fujito, T., and Inoue, S. (2015).
\newblock On the eternal vertex cover numbers of generalized trees.
\newblock {\em {IEICE} Trans. Fundam. Electron. Commun. Comput. Sci.},
  98-A(6):1153--1160.

\bibitem[Babu et~al., 2021]{BCFPRW2021}
Babu, J., Chandran, L.~S., Francis, M.~C., Prabhakaran, V., Rajendraprasad, D.,
  and Warrier, N.~J. (2021).
\newblock On graphs whose eternal vertex cover number and vertex cover number
  coincide.
\newblock {\em Discrete Applied Mathematics (in press)}.

\bibitem[Babu and Prabhakaran, 2020]{BP2020}
Babu, J. and Prabhakaran, V. (2020).
\newblock A new lower bound for the eternal vertex cover number of graphs.
\newblock In {\em Proceedings of the 26th International Conference on Computing
  and Combinatorics {COCOON}}, pages 27--39.

\bibitem[Babu et~al., 2020]{BPS2020b}
Babu, J., Prabhakaran, V., and Sharma, A. (2020).
\newblock A substructure based lower bound for eternal vertex cover number.
\newblock In {\em Proceedings of the 21st Italian Conference on Theoretical
  Computer Science}, pages 1--14.

\bibitem[Cygan et~al., 2015]{pcbook}
Cygan, M., Fomin, F.~V., Kowalik, {\L}., Lokshtanov, D., Marx, D., Pilipczuk,
  M., Pilipczuk, M., and Saurabh, S. (2015).
\newblock {\em Parameterized Algorithms}.
\newblock Springer.

\bibitem[Diestel, 2017]{graphsbook}
Diestel, R. (2017).
\newblock {\em Graph theory}.
\newblock Springer.

\bibitem[Dom et~al., 2014]{DLS14}
Dom, M., Lokshtanov, D., and Saurabh, S. (2014).
\newblock Kernelization lower bounds through colors and ids.
\newblock {\em {ACM} Trans. Algorithms}, 11(2):13:1--13:20.

\bibitem[Downey and Fellows, 2013]{downeyfellows}
Downey, R.~G. and Fellows, M.~R. (2013).
\newblock {\em Fundamentals of Parameterized Complexity}.
\newblock Springer.

\bibitem[Flum and Grohe, 2006]{flumgrohe}
Flum, J. and Grohe, M. (2006).
\newblock {\em Parameterized Complexity Theory}.
\newblock Springer.

\bibitem[Fomin et~al., 2010]{FGGKS2010}
Fomin, F.~V., Gaspers, S., Golovach, P.~A., Kratsch, D., and Saurabh, S.
  (2010).
\newblock Parameterized algorithm for eternal vertex cover.
\newblock {\em Information Processing Letters}, 110(16):702--706.

\bibitem[Fomin et~al., 2019]{kernelbook}
Fomin, F.~V., Lokshtanov, D., Saurabh, S., and Zehavi, M. (2019).
\newblock {\em Kernelization: Theory of Parameterized Preprocessing}.
\newblock Cambridge University Press.

\bibitem[Fujito and Nakamura, 2020]{FN2020}
Fujito, T. and Nakamura, T. (2020).
\newblock Eternal connected vertex cover problem.
\newblock In {\em Proceedings of the 16th International Conference on Theory
  and Applications of Models of Computation, {TAMC}}, pages 181--192. Springer.

\bibitem[Klostermeyer and Mynhardt, 2011]{KM2011}
Klostermeyer, W. and Mynhardt, C.~M. (2011).
\newblock Graphs with equal eternal vertex cover and eternal domination
  numbers.
\newblock {\em Discrete Mathematics}, 311(14):1371--1379.

\bibitem[Klostermeyer and Mynhardt, 2009]{KM09}
Klostermeyer, W.~F. and Mynhardt, C.~M. (2009).
\newblock Edge protection in graphs.
\newblock {\em Australas. {J} Comb.}, 45:235--250.

\bibitem[Niedermeier, 2006]{pcinvitation}
Niedermeier, R. (2006).
\newblock {\em Invitation to Fixed-Parameter Algorithms}.
\newblock Oxford University Press.

\end{thebibliography}


\end{document}